\documentclass[11pt]{article}
\usepackage{fullpage}

\usepackage{xcolor}
\usepackage[%dvips,
CJKbookmarks=true,
bookmarksnumbered=true,
bookmarksopen=true,
colorlinks=true,
citecolor=red,
linkcolor=blue,
anchorcolor=red,
urlcolor=blue
]{hyperref}

\usepackage{authblk}
\title{A Policy Gradient Algorithm for the Risk-Sensitive Exponential Cost MDP}
\author[1]{M. Moharrami}
\author[1]{Y. Murthy}
\author[2]{A. Roy}
\author[1]{R. Srikant}
\affil[1]{University of Illinois at Urbana-Champaign}
\affil[2]{Indian Institute of Technology Guwahati}

\usepackage{enumitem}
\usepackage{comment}
\usepackage[utf8]{inputenc} % allow utf-8 input
\usepackage[T1]{fontenc}    % use 8-bit T1 fonts
\usepackage{lmodern}
\usepackage{bbm}
\usepackage{graphicx}
\usepackage{amsmath,amsthm,amsfonts,amssymb,mathabx,mathtools}
\usepackage{enumerate,enumitem}
\usepackage{xspace}
\usepackage[capitalise]{cleveref}
\usepackage{autonum}
\usepackage{nicematrix}
\usepackage{tikz}

\newcommand{\real}{\mathbb{R}}
\newcommand{\nplus}{\mathbb{N}}
\newcommand{\nzero}{\mathbb{N}_0}

\newcommand{\e}{\mathrm{e}}
\newcommand{\U}{\mathcal{U}}
\newcommand{\ie}{i.e.\xspace}

\renewcommand{\check}{\widecheck}
\renewcommand{\hat}{\widehat}
\renewcommand{\tilde}{\widetilde}

\newcommand{\samplespace}{\mathcal{X}}
\newcommand{\actionspace}{\mathcal{A}}
\newcommand{\statespace}{\mathcal{S}}
\newcommand{\probspace}{\mathcal{P}}
\newcommand{\matspace}{\mathcal{M}}

\newcommand{\prob}[1]{\mathbb{P}{\left(#1\right)}}
\newcommand{\expect}[3]{\mathbb{E}{_{#1}^{#2}\left[#3\right]}}
\newcommand{\checkexpect}[3]{\check{\mathbb{E}}{_{#1}^{#2}\left[#3\right]}}
\newcommand{\expon}[1]{\exp\left(#1\right)}
\newcommand{\norm}[1]{\left\lVert#1\right\rVert}

\newcommand{\der}[1]{\mathrm{d}#1}
\newcommand{\dertwo}[1]{\mathrm{d}^2#1}
\newcommand{\derthree}[1]{\mathrm{d}^3#1}

\newcommand{\myquad}[1][1]{\hspace*{#1em}\ignorespaces}
\newcommand{\myeq}[1]{\stackrel{\mathclap{\normalfont\mbox{#1}}}{=}}

\makeatletter
\newcommand*{\xbar}{}%
\DeclareRobustCommand*{\xbar}{%
	\mathpalette\@xbar{}%
}
\newcommand*{\@xbar}[2]{%
	\sbox0{$#1\mathrm{X}\m@th$}%
	\sbox2{$#1\m@th$}%
	\rlap{%
		\hbox to\wd2{%
			\hfill
			$\overline{%
				\vrule width 0pt height\ht0 %
				\kern\wd0 %
			}$%
		}%
	}%
	\copy2 %
}
\newcommand*\xoverbar[1]{\hbox{\vbox{
			\hrule height 0.5pt % The actual bar
			\kern0.35ex%         % Distance between bar and symbol
			\hbox{\kern-0.05em%      % Shortening on the left side
				{#1}%
				\kern-0.06em%      % Shortening on the right side
	}}}
}
\newcommand*\xbarcheck[1]{\hbox{\vbox{
			\hrule height 0.5pt % The actual bar
			\kern0.4ex%         % Distance between bar and symbol
			\hbox{\kern-0.195em%      % Shortening on the left side
				\ensuremath{\check{#1}}%
				\kern-0.065em%      % Shortening on the right side
	}}}
}

\newcommand*\xcheckbar[1]{\ensuremath{\check{\hbox{\vbox{
					\hrule height 0.5pt % The actual bar
					\kern0.4ex%         % Distance between bar and symbol
					\hbox{\kern-0.195em%      % Shortening on the left side
						\ensuremath{{#1}}%
						\kern-0.075em%      % Shortening on the right side
	}}}}}
}
\makeatother

\newtheorem{theorem}{Theorem}
\newtheorem{lemma}{Lemma}
\newtheorem{proposition}{Proposition}
\newtheorem{corollary}{Corollary}
\newtheorem{assumption}{Assumption}
\newtheorem{claim}{Claim}
\newtheorem{remark}{Remark}

\Crefname{remark}{Remark}{Remarks}
\Crefname{claim}{Claim}{Claims}
\Crefname{assumption}{Assumption}{Assumptions}
\Crefname{lemma}{Lemma}{Lemmas}
\Crefname{theorem}{Theorem}{Theorems}
\Crefname{corollary}{Corollary}{Corollaries}
\Crefname{proposition}{Proposition}{Propositions}
\begin{document}
	\maketitle
	\begin{comment}
		-let's focus on the risk-sensitive exponential cost (abstract+intro)
		-our main result and key ideas (abtract+intro)
		-give a more general intro to risk-sensitive exp-cost MDPs, relationship to robustness in terms of KL distance. Ack the works of Borkar,Meyn,Ananatharam
		-Somewhere we should talk about standard MDPs
		-Talk about the RL challenges and the differences between Borkar's previous works
		-The rest of the paper is organized as follows:....
		-Related work (subsection): briefly compare and explain the limitations. Also, the Kaiqing Zhang, Basar work
	\end{comment}

	\begin{abstract}
		We study the risk-sensitive exponential cost MDP formulation and develop a trajectory-based gradient algorithm to find the stationary point of the cost associated with a set of parameterized policies. We derive a formula that can be used to compute the policy gradient from (state, action, cost) information collected from sample paths of the MDP for each fixed parameterized policy. Unlike the traditional average-cost problem, standard stochastic approximation theory cannot be used to exploit this formula. To address the issue, we introduce a truncated and smooth version of the risk-sensitive cost and show that this new cost criterion can be used to approximate the risk-sensitive cost and its gradient uniformly under some mild assumptions. We then develop a trajectory-based gradient algorithm to minimize the smooth truncated estimation of the risk-sensitive cost and derive conditions under which a sequence of truncations can be used to solve the original, untruncated cost problem.
%		Robustness in control theory and MDPs has one of the following two interpretations: (i) we would like to avoid very large costs with high probability and (ii) we would like the system to be robust to modeling uncertainty (called unmodelled dynamics in control theory and distributional shift in reinforcement learning). One formulation that captures both notions of robustness is the infinite-horizon risk-sensitive exponential cost MDP formulation. This risk-sensitive, exponential cost MDP has not been studied extensively in the RL literature due to the fact that the corresponding Bellman equation has a product form (unlike the sum form in the traditional average or discounted cost infinite-horizon formulations) and therefore, it has been difficult to derive easily implementable forms of widely used RL algorithms such as policy gradient, Q-learning, actor-critic algorithm, etc. We take an important step towards deriving implementable algorithms by deriving a policy gradient theorem which, for the first time, can be used to compute the policy gradient from (state, action, cost) information collected from sample paths of the MDP for each fixed parameterized policy. Our result holds for arbitrary policy parameterizations unlike prior work, which only holds for the tabular case.
	\end{abstract}

	\section{Introduction}\label{sec:intro}
	% One paragraph: general definition related to RL
% One paragraph: why risk sensitive is important and its connection to risk neutral MDP and
% Our contribution and comparision with the work of Borak
Reinforcement Learning (RL) has been remarkably successful in a wide variety of applications including video gaming, robotics, communications networks, pricing, transportation and product management. RL algorithms are trained on data from an environment which is assumed to correspond to a Markov Decision Process (MDP) \cite{bertsekas1995neuro,sutton2018reinforcement}, and the trained algorithm is then deployed with the hope that it will perform well in environments that are similar to the ones on which it was trained. In practice, it is possible that the environment in which the algorithm is deployed may have different statistical properties than the training environment. Whether RL algorithms are robust to such distributional shifts is an important consideration in practice. In the control theory and MDP literature, distributional shift has been studied using the exponential-cost infinite-horizon formulation \cite{whittle1990risk,whittle1982optimization,dai1996connections} which is the focus of this paper. The notion of robustness is closely connected to the notion of risk-sensitivity. In particular, it can be shown that the solution to the risk-sensitive cost formulation provides robustness to models which are within a Kullbak-Leibler ball around a nominal distribution; see \cite{dai1996connections, osogami2012robustness, Anantharam2017,Follmer2008,Follmer2011} and \cref{sec_app:background} for details.

RL algorithms for the risk-sensitive exponential cost MDP have not been studied extensively in the literature as the risk-neutral counterpart. One reason for this is the multiplicative form of dynamic programming equation  in risk-sensitive exponential cost MDP \cite{borkar2002risk} unlike the classical MDP literature involving an additive Bellman equation \cite{bertsekas1995neuro}.
A variant of Q-learning for the risk-sensitive exponential cost is proposed in \cite{borkar2002q} and an actor-critic algorithm for the risk-sensitive exponential cost is proposed in \cite{Borkar2001}. However, these schemes work only deal with the tabular case, i.e., no approximation is performed either in the value function domain or in the policy domain.
In \cite{basu2008learning}, a linear function approximation based scheme for the problem of estimating the value function and cost under a given policy is studied. Interestingly, if one were just interested in just policy evaluation, i.e., estimating the cost alone, the algorithm still requires one to estimate the value function. See \cite{borkar2010learning} for an overview on
existing works in exponential cost risk-sensitive RL.
\par
In this paper, using a fixed point representation of the risk-sensitive exponential cost, we first discuss the problem of policy evaluation. We show that standard stochastic approximation for a fixed point equation may not apply as the stochastic noise may not be summable. We develop a policy evaluation algorithm by considering an increasing sequence of truncated approximations of the risk-sensitive cost. Using the implicit function theorem and the fixed point representation of the risk-sensitive cost, we then derive a policy gradient theorem based on visits to the recurrent state. As in the case of policy evaluation, this form of the gradient cannot be used in practice. Moreover, the truncated approximation of risk-sensitive cost cannot be used either as it may not be smooth enough with respect to the policy parameterization. Hence, we introduce a smooth truncated approximation of the risk-sensitive cost. We show that this new cost criterion approximates the risk-sensitive cost and its gradient uniformly. We then generalize the result of \cite{Marbach2001} to develop a trajectory-based algorithm to minimize the smooth and truncated approximation of the risk-sensitive cost; we note that  this generalization to the risk-sensitive setting was considered to be challenging~\cite{Bhatnagar2021}.
%One of our key contributions is to show that the convergence does not depend on the specific structure average-reward problem, rather on more general properties which also hold in the risk-sensitive setting.
We also show that there exists a sequence of increasing truncations that recovers the same results for untruncated risk-sensitive exponential cost.

\begin{comment}
The rest of the paper is organized as follows. In \cref{sec:review}, we discuss risk-sensitive formulation, state a few assumptions, and review a few basic results. We then discuss the problem of policy evaluation in \cref{sec:costest}. In \cref{sec:policygrad}, we derive a risk-sensitive formulation based on visits to a recurrent state. In \cref{sec:heuristic}, we present a series of idealized algorithms, and discuss the intuition behind our trajectory-based gradient algorithm. We then discuss our algorithm and present our main result in \cref{sec:algorithm}. Finally, in \cref{sec:riskmdp}, we generalize the results presented in the previous sections to risk-sensitive exponential cost MDP. The proof of our main theorem and some supporting results are postponed to \cref{proof:thm:main,proof:supporting} respectively. We also present the necessary background for the rest of the paper and discuss the robustness of the risk-sensitive formulation in \cref{sec_app:background}, and an informed reader can skip this material.
\end{comment}

\noindent{\bf Related Work:}
Another notion of risk-sensitivity in the RL literature is to provide some trade-off between the mean and the variance (or higher order moments) of the long-term average cost \cite{prashanth2013actor,prashanth2016variance,tamar2013temporal,tamar2016learning,fu2018risk}. It is well-known that one can obtain the mean-variance optimization formulation from a Taylor's series expansion of the exponential cost formulation. % as follows: we observe that $\frac{1}{\alpha}\log \mathbb{E}[e^{\alpha D_T(\theta)}]=\mathbb{E}[D_T(\theta)]+\frac{\alpha}{2} Var[D_T(\theta)]+O(\alpha^2).$
%However, this formulation does not fully capture the strong robustness notion in terms of the distributional shift (measured by KL distance) that the exponential cost formulation possesses.
%We would also like to point out that the proof of convergence in \cite{prashanth2016variance} (and hence in \cite{prashanth2013actor}) appears to have a gap. In particular, the function $\hat{L}(\theta,\lambda)$ defined there is not coercive \cite[Theorem 4]{prashanth2016variance}, which seems to be required for their proof to hold.
Another robustness metric used in the MDP/RL literature is the notion of Conditional VaR (CVaR) \cite{chow2015risk,rockafellar2002conditional,chow2014algorithms,chow2017risk,tamar2015optimizing}. CVaR focuses on minimizing a conditional expected cost where the conditioning is over a fraction of sample paths which lead to high costs. For finite-horizon problems, CVaR is shown to have some robustness to modeling uncertainty \cite{chow2015risk} but not in the infinite-horizon case to the best of our knowledge. Recently, another approach known as distribution RL, has emerged, see \cite{bellemare2017distributional,singh2020improving}. In distributional RL, instead of learning the average value, one is interested in learning the value distribution. None of the above approaches fully capture the strong robustness notion in terms of the distributional shift (measured by KL distance) that the exponential cost formulation possesses.
A trajectory-based gradient algorithm has been studied in \cite{tamar2021derivative} for the special case of linear quadratic control problems; however, their analysis relies on the known explicit form of the optimal policy which is not the case for MDPs.

\noindent {\bf Notation:}
``$\textrm{Const}$'' is used to denote a positive constant independent of the parameters, and its value may change even in the same line. $\mathbb{R}_+$ and $\mathbb{R}_{++}$ denote the set of non-negative and strictly positive real numbers respectively.
$\nzero$ and $\nplus$ denote the set of non-negative and strictly positive integers respectively.

	\section{Markov Chains and the Risk-Sensitive Exponential Cost}\label{sec:review}
	Consider a discrete-time Markov chain $\{\Phi_i\}_{i\geq0}$ with a finite state space $\samplespace$. The transition probability of $\{\Phi_i\}$ is assumed to depend on a parameter vector $\theta \in \real^l$, and is denoted by
\begin{align}
	P_\theta(x,y) \coloneqq \prob{\Phi_{k+1} = y \vert \Phi_k = x, \theta} \qquad\forall x,y\in \samplespace \text{ and }k\in \nzero.
\end{align}
In an MDP, the parameter $\theta$ will parameterize the class of policies that one considers. Let $\probspace= \left\{ P_\theta:\theta\in\mathbb{R}^l \right\}$ denote the set of all transition probabilities, and let $\xbar{\probspace}$ denote its closure in the space of $|\samplespace|\times|\samplespace|$ matrices. Notice that elements of $\xbar{\probspace}$ are stochastic, and hence, they define a Markov chain on the same state space $\samplespace$.
\begin{assumption}\label[assumption]{ass:1}
	For each $P\in \xbar{\probspace}$, the Markov chain with transition probability $P$ is aperiodic and irreducible with a common recurrent state $x^*.$%~\cite[Assunption 1]{Marbach2001}.
\end{assumption}

Suppose for any parameter vector $\theta \in \real^l$, there is a one-step cost function $C_\theta:\samplespace\to \real$ which is the cost we incur at each state under the parameter $\theta$. The risk-sensitive cost of the Markov chain $\{\Phi_i\}_{i\geq 0}$ with probability transition kernel $P_\theta$ is defined as follows:
\begin{align}\label{eq:deflambda}
	\Lambda_\theta \coloneqq \lim_{n\to\infty}\frac{1}{n}\ln\expect{\Phi_0=x}{\theta}{\expon{\alpha \sum_{i=0}^{n-1}C_\theta(\Phi_i)}},\qquad \forall x\in \samplespace,
\end{align}
where $\alpha >0$ is called the risk factor and $\expect{\Phi_0=x}{\theta}{\cdot}$ denotes the expectation with respect to the probability transition kernel $P_\theta$ given $\Phi_0=x$.
By invoking the multiplicative ergodic theorem~\cite[Theorem 1.2]{Balaji2000}, which trivially holds for any aperiodic and irreducible finite-state Markov chain, it can be shown that the above limit exists and does not depend on $x\in\samplespace$. In particular, it can be shown that $\lambda_\theta\coloneqq\exp(\Lambda_\theta)$ is the largest eigenvalue of  $\hat{P}_\theta$ with multiplicity $1$, where $\hat{P}_\theta(x,y) \coloneqq \expon{\alpha C_\theta(x)} P_\theta(x,y)$ for all $x,y\in\samplespace$. In \cref{sec_app:back_risk}, we provide simple proofs of these results for the case of finite state-space Markov chains which is the focus of this paper.

\begin{assumption}\label[assumption]{ass:2}
	For each $x,y\in\samplespace$, the transition kernel $P_\theta(x,y)$ and the one-step cost function $C_\theta(x)$ are bounded, twice differentiable, and have bounded first and second derivatives.%~\cite[Assunption 2]{Marbach2001}.
\end{assumption}

Let $h_\theta$ denote the right-eigenvector corresponding to the eigenvalue $\lambda_\theta$ of $\hat{P}$. Notice that by the Perron-Frobenius theorem, $h_\theta$ is a strictly positive vector. Let us write the corresponding eigenequation as follows:
\begin{align}\label{eq:multpoiss}
	h_\theta(x) = \frac{\exp(\alpha C_\theta(x))}{\lambda_\theta} \sum_{y\in\samplespace}p_\theta(x,y)h_\theta(y)\qquad\forall x\in\samplespace,
\end{align}
The above equation is called the multiplicative Poisson equation, and it is the multiplicative analog of the Bellman equation for the average cost problem.  In particular, $h_\theta$ can be interpreted as the relative value function associated with the risk-sensitive cost problem. Notice that $h_\theta$ is the unique solution to the multiplicative Poisson equation up to a scaling factor.
Similar to the relative value function for the average cost problem, $h_\theta$ can be written in terms of visits to the recurrent state $x^*$. In particular,  $h_\theta$ is given by
\begin{align}
	h_\theta(x) \,\propto\, \expect{\Phi_0=x}{\theta}{\expon{\sum_{i=0}^{\tau_{x^*}-1}\left(\alpha C_\theta(\Phi_i)- \Lambda_\theta\right)}}\qquad\forall x\in\samplespace, \label{eq:recversion}
\end{align}
where $\tau_{x^*}$ is the first return time to the recurrent state $x^*\in \mathcal{X}$. Notice that by \cref{ass:1}, %we have
\begin{align}
	\expect{\Phi_0=x^*}{\theta}{\expon{\sum_{i=0}^{\tau_{x^*}-1}\left(\alpha C_\theta(\Phi_i)- \Lambda_\theta\right)}} = 1,\label{eq:fixedpoint}
\end{align}
which can be proved by twisting $\hat{P}$ and defining a twisted kernel $\check{P}_\theta(x,y) \coloneqq {\hat{P}_\theta(x,y) h_\theta(y)}\mathbin{/}{\lambda_\theta h_\theta(x)} = {\exp(\alpha C_\theta(x)) P_\theta(x,y) h_\theta(y)}\mathbin{/}{\lambda_\theta h_\theta(x)}.$
%\begin{align}
%	\check{P}_\theta(x,y) \coloneqq \frac{\hat{P}_\theta(x,y) h_\theta(y)}{\lambda_\theta h_\theta(x)} = \frac{\exp(\alpha C_\theta(x)) P_\theta(x,y) h_\theta(y)}{\lambda_\theta h_\theta(x)}.
%\end{align}
See \cref{sec_app:back_risk} for more details.

Our analysis is mostly based on the expression for $h_\theta$ given by \cref{eq:recversion} and the fact that $\Lambda_\theta$ is the unique fixed point of \cref{eq:fixedpoint}. In particular, we apply a stochastic approximation algorithm to estimate the cost of a policy $\theta\in\real^l$. We note that $h_\theta$ given by \cref{eq:recversion} is uniformly bounded (see \cref{lem:hbounded}).

Later, we use the twisted kernel to develop a sample-based approximation of $\nabla_\theta \Lambda_\theta$. Notice that $\check{P}_\theta$ is stochastic, and hence, it defines a Markov chain on the same state space $\samplespace$. Let $\check{\probspace}= \big\{ \check{P}_\theta:\theta\in\mathbb{R}^l \big\}$, and let $\,\xbarcheck{\probspace}$ denote its closure in the space of $|\samplespace|\times|\samplespace|$ matrices. It is easy to verify that by \cref{lem:hbounded}, $\check{\probspace}$ also satisfies \cref{ass:1} (see \cref{cor:Phat_commonrec}).

	\section{Sample-Based Approximation of the Risk-Sensitive Exponential Cost}\label{sec:costest}
	In this section, we present an algorithm for estimating the risk-sensitive cost from a single sample path and discuss the difficulties compared to the average cost problem.

Let us fix a parameter $\theta\in \real^l$. Suppose that we are given a single sample path $\{(\Phi_i,C_\theta(\Phi_i))\}_{i\geq 0}$ and the goal is to estimate the risk-sensitive cost $\Lambda_\theta$. In the case of average-cost problem, one can simply take the average of $\{C_\theta(\Phi_i)\}$ to get an accurate estimate of the average cost. This is due to the almost sure convergence of the sample path average cost to its limit. Notice that the same property does not hold for the risk-sensitive cost as the order of expectation and averaging cannot be interchanged in \cref{eq:deflambda}.

One approach for estimating the value of $\Lambda_\theta$ is to use the fixed point equation given by \cref{eq:fixedpoint}. Consider the function $g:\real^{l+1} \to \real_{++} \cup \{\infty\}$ defined as follows:
\begin{align}
	g(\theta,\Lambda) &\coloneqq \expect{\Phi_{0} = x^*}{\theta}{\expon{\sum_{i=0}^{\tau_{x^*}-1} \left(\alpha C_\theta({\Phi}_i) - \Lambda\right)}}= \checkexpect{\check{\Phi}_0=x^*}{\theta}{\e^{{\tau_{x^*}}\left(\Lambda_\theta - \Lambda\right)}}, \label{eq:defg}
\end{align}
where $\checkexpect{}{\theta}{\cdot}$ is the expectation with respected to the twisted kernel $\check{P}_\theta.$ Notice that $\Lambda_\theta$ is the unique fixed point of $g(\theta,\Lambda) = 1$. Hence, a natural stochastic approximation algorithm to estimate $\Lambda_\theta$ is
\begin{align}\label{eq:costest_vanilstoch}
	\tilde{\Lambda}_{m+1} = \tilde{\Lambda}_{m} + \gamma_m\left(\expon{\sum_{i=t_m}^{t_{m+1}-1} \left(\alpha C_\theta({\Phi}_i) - \tilde{\Lambda}_m\right)} - 1\right),
\end{align}
where $t_0 = 0,$ $t_m$ is the $m$th visit to $x^*,$ and $\gamma_m$ is chosen so that $\sum_m\gamma_m\to\infty$ and $\sum_{m} \gamma_m^2 < \infty$. However, the above algorithm may not converge for two reasons:
\begin{itemize}[leftmargin=*]
\item For any fixed $\theta\in\real^l$, $g(\theta,\Lambda)$ is $+\infty$ for all small enough values of $\Lambda$. Thus, one cannot approximate \cref{eq:costest_vanilstoch} by the ODE $\dot{\tilde{\Lambda}}(t) = (g(\theta,\tilde{\Lambda}(t)) - 1)$.
%    \begin{align}\label{eq:costest_ODE}
%	    \dot{\tilde{\Lambda}}(t) = (g(\theta,\tilde{\Lambda}(t)) - 1),\qquad \tilde{\Lambda}(0) \text{ is large enough so that } g(\theta,\tilde{\Lambda}(0)) < \infty.
%    \end{align}
\item It is possible that if we start with a sufficiently large value of $\tilde{\Lambda}_0$ such that $g(\theta,\tilde{\Lambda}_0)<\infty,$ then $g(\theta,\tilde{\Lambda}_m)<\infty$ for all $m\geq 0$ for sufficiently small $\alpha.$ One condition under which this will hold is when $g(\theta,\underline{C} - \gamma_{\max}) < \infty$, where $\underline{C}\coloneqq \min_{x\in\samplespace}C_\theta(x)$ and $\gamma_{\max} \coloneqq \max_{m}\gamma_m$. But even in such a case, standard convergence proofs for stochastic approximation cannot be applied because the stochastic noise may be too large. To see this, let us rewrite the above update equation as $\tilde{\Lambda}_{m+1} = \tilde{\Lambda}_{m} + \gamma_m\left(g(\theta,\tilde{\Lambda}_m) - 1\right)  + \epsilon_m,$ %follows:
%\begin{align}
	%\tilde{\Lambda}_{m+1} &= \tilde{\Lambda}_{m} + \gamma_m\left(g(\theta,\tilde{\Lambda}_m) - 1\right) +  \gamma_m\left(\expon{\sum_{i=t_m}^{t_{m+1}-1} \left(\alpha C_\theta({\Phi}_i) - \tilde{\Lambda}_m\right)} - g(\theta,\tilde{\Lambda}_m)\right)\\
%	\tilde{\Lambda}_{m+1} &= \tilde{\Lambda}_{m} + \gamma_m\left(g(\theta,\tilde{\Lambda}_m) - 1\right)  + \epsilon_m,
%\end{align}
where $\epsilon_m$ is the stochastic error. It is easy to see that $\left\{\epsilon_m \right\}$ is a martingale difference sequence. However, $\{\epsilon_m\}$ may not be summable as $\expect{}{}{\left|\epsilon_m\right|^2}$ might be infinity, depending on the value of $\alpha$. In particular, even if $\tilde{\Lambda}_{m} = \Lambda_\theta$, we may still have $\expect{}{}{\left|\epsilon_m\right|^2} = \infty$.

\end{itemize}

Hence, the vanilla form of the update equation given by \cref{eq:costest_vanilstoch} may not work for large values of $\alpha$, and we have no information to determine for which values of $\alpha$ the estimation converges to the correct value a priori. The natural solution is to use a truncation which results in the following update equation:
\begin{align}\label{eq:costest_fixtruncstoch}
	\tilde{\Lambda}_{m+1} = \tilde{\Lambda}_{m} + \gamma_m\left(\expon{\sum_{i=t_m}^{t_{m+1}-1} \left(\alpha C_\theta({\Phi}_i) - \tilde{\Lambda}_m\right)} \wedge M - 1\right),
\end{align}
where $M > 1$ is a fixed constant. Clearly, the above iteration will not converge to $\Lambda_\theta$. Consider the function $g^{(M)}:\real^{l+1} \to \real_{++}$ defined as follows:
\begin{align}
	g^{(M)}(\theta,\Lambda) \coloneqq \expect{\Phi_0=x^*}{\theta}{\expon{\sum_{i=0}^{\tau_{x^*}-1}\left(\alpha C_\theta(\Phi_i)- \Lambda\right)}\wedge M}.
\end{align}
It is easy to verify that \cref{eq:costest_fixtruncstoch} converges to $\Lambda_\theta^{(M)}$, where $\Lambda_\theta^{(M)}$ is the unique fixed point of $g^{(M)}(\theta,\Lambda) = 1$. 
%However, a natural question to ask is how good of approximation is $\Lambda_\theta^{(M)}$ for $\Lambda_\theta$. 
In \cref{lem:unifapprox_truncost}, we show that $\Lambda_\theta^{(M)}$ approximates $\Lambda_\theta$ uniformly.
%In view of \cref{lem:unifapprox_truncost}, for all practical purposes, it may sufficient to choose a large $M$ and obtain an approximation to $\Lambda_\theta$ using the \cref{eq:costest_fixtruncstoch}. 
In \cref{prop:costest} we show that $\tilde{\Lambda}_{m} \to \Lambda_\theta$ for an increasing sequence of truncations; the proof is given in \cref{proof:prop:costest}.
\begin{proposition}\label[proposition]{prop:costest}
	Let \cref{ass:1,ass:2} hold. Suppose that the step-size $\{\gamma_m\}_{m\geq0}$ is given such that $\sum_{m}\gamma_m \to\infty$ and $\sum_{m}\gamma_m^{2(1-\beta)} <\infty$ for some fixed $\beta \in (0,0.5)$. Consider the following update equation
	\begin{align}\label{eq:prop:costest_vartrunc}
		\tilde{\Lambda}_{m+1} = \tilde{\Lambda}_{m} + \gamma_m\left(\expon{\sum_{i=t_m}^{t_{m+1}-1} \left(\alpha C_\theta({\Phi}_i) - \tilde{\Lambda}_m\right)} \wedge M_m - 1\right),
	\end{align}
	where $M_m = \gamma_m^{-\beta}$. Then, we have $\tilde{\Lambda}_{m}\to\Lambda_\theta$ almost surely.
\end{proposition}
%It is not difficult to see that the condition in the above proposition can be replaced by $\sum_m\gamma_m^2M_m^2<\infty;$ this is a slightly stronger condition on the step-size sequence compared to traditional stochastic approximation where the usual condition is $\sum_m\gamma_m^2<\infty.$

	\section{Policy Gradient Theorem}\label{sec:policygrad}
	In this section, we present a simple approach to derive the policy gradient theorem for risk-sensitive cost problem.
%using \cref{eq:fixedpoint} and the implicit function theorem. We show that our derivation results in the same sensitivity formula presented in \cite{Borkar2001}, however, it suggests a way to approximate the gradient of a truncated and smooth approximation of the risk-sensitive cost in the non-tabular case which we will present the next section.

Recall that for any fixed $\theta\in\real^l$, $\Lambda_\theta$ is the unique fixed point of $g(\theta,\Lambda) = 1$, where $g:\real^{l+1} \to \real_{++} \cup \{\infty\}$ is given as follows:
\begin{align}
	g(\theta,\Lambda) &= \expect{\Phi_0=x^*}{\theta}{\expon{\sum_{i=0}^{\tau_{x^*}-1}\left(\alpha C_\theta(\Phi_i)- \Lambda\right)}}= \checkexpect{\check{\Phi}_0=x^*}{\theta}{\e^{{\tau_{x^*}}\left(\Lambda_\theta - \Lambda\right)}}.
\end{align}
It is easy to verify that given \cref{ass:1,ass:2}, for any $\theta\in\real^l$ there exists a small neighborhood around $(\theta,\Lambda_\theta)$ for which $g$ is bounded. Now to get a policy gradient theorem, the idea is to apply the implicit function theorem to $g - 1$ at $(\theta,\Lambda_\theta)$. To do this, we need the following technical assumption.
\begin{assumption}\label[assumption]{ass:3}
	For any $x,y\in\samplespace$, there exist bounded functions $L_\theta(x,y)$ and $L^{(2)}_\theta(x,y)$ such that $({\normalfont a})~\nabla_\theta P_\theta(x,y) = P_\theta(x,y) L_\theta(x,y),$ and $({\normalfont b})~ \nabla^2_\theta P_\theta(x,y) = P_\theta(x,y) L^{(2)}_\theta(x,y).$
\end{assumption}

We note that Assumptions~\ref{ass:1}-\ref{ass:3} are also used in \cite{Marbach2001}.
Now, let us consider a fixed $\theta\in\real^{l}$. By \cref{ass:1,ass:2,ass:3}(a), $g$ is differentiable in a neighborhood of $(\theta,\Lambda_\theta)$. In particular, for any $(\theta,\Lambda)$ for which $g(\theta,\Lambda) < \infty$, by simple algebra we can derive a formula for $\nabla_\theta g(\theta,\Lambda)$ and ${\der{g(\theta,\Lambda)}}\mathbin{/}{\der{\Lambda}}$ (see \cref{lem:propgf}).

Invoking the implicit function theorem, we get the following risk-sensitive formula for $\nabla_\theta \Lambda_\theta$ in terms of visits to the recurrent state $x^*$:
\begin{align}
	\nabla_\theta \Lambda_\theta &=-\left(\frac{\der{g(\theta,\Lambda_\theta)}}{\der{\Lambda}}\right)^{-1}  \nabla_\theta g(\theta,\Lambda_\theta) = \frac{\checkexpect{\check{\Phi}_{0} = x^*}{\theta}{\sum_{i=0}^{\tau_{x^*}-1} \big(\alpha \nabla_\theta C_\theta(\check{\Phi}_i) + L_\theta(\check{\Phi}_i,\check{\Phi}_{i+1})}}{\checkexpect{\check{\Phi}_{0} = x^*}{\theta}{\tau_{x^*}}}\\
	&\propto\,\expect{\Phi_{0} = x^*}{\theta}{\sum_{i=0}^{\tau_{x^*}-1} \big(\alpha \nabla_\theta C_\theta({\Phi}_i) + L_\theta(\Phi_i,\Phi_{i+1}) \big)\expon{\sum_{i=0}^{\tau_{x^*}-1} \left(\alpha C_\theta({\Phi}_i) - \Lambda_\theta\right)}}\label{eq:policygrad_reg},
\end{align}
where we used the proportionality symbol somewhat loosely to denote that the vectors on the left and right hand sides are in the same direction. In later applications of this formula to obtain a gradient descent algorithm, we neglect the ${\checkexpect{\check{\Phi}_{0} = x^*}{\theta}{\tau_{x^*}}}$ term since it will not affect the convergence of the algorithm. We note that, under \cref{ass:1,ass:2,ass:3}(a), it is not difficult to see that this formula is the same as the policy gradient in \cite{Borkar2001}. However, we need \cref{ass:3}(a) for our derivation but it turns out that this assumption is critical to deriving a trajectory-based algorithm for the non-tabular case. We will see in the next section that the approach presented here can be used to estimate the gradient of a truncated and smooth approximation of the risk-sensitive cost. It is worth noting that the approximated risk-sensitive cost does not satisfy any multiplicative Poisson's equation, and hence, the approach in \cite{Borkar2001} cannot be used.

	\section{Intuition Behind the Trajectory-Based Algorithm}\label{sec:heuristic}
	Our goal is to find a stationary point of the risk-sensitive cost over $\theta\in\real^l$ by implementing a trajectory-based gradient-descent algorithm.  We present a series of idealized algorithms, which will lead to our policy gradient algorithm in the next section. 
%Our algorithm is motivated by \cite{Marbach2001}, but there are significant differences and difficulties due to the risk-sensitive exponential cost which we will point out in this section.

\subsection{Idealized Gradient Algorithm}\label{subsec:heuristic_main}
Suppose that for any $\theta\in\real^l$, we have access to $\nabla_\theta \Lambda_\theta$. Starting from $\theta_0$, a natural way to minimize the risk-sensitive cost is to run the following system of ODEs:
\begin{align}\label{eq:accnabla_main}
	\dot{\theta}(t) = -\nabla_\theta\Lambda_{\theta(t)},\qquad\theta(0) = \theta_0.
\end{align}
By the chain rule, we have $\dot{\Lambda}_{\theta(t)} = -\norm{\nabla_\theta\Lambda_{\theta(t)}}^2$. By \cref{ass:2}, ${\Lambda}_\theta$ is bounded over $\theta\in\real^l$. Hence, $\Lambda_{\theta(t)}$ decreases and converges to some $\Lambda^*,$ and $\lim_{t\to\infty} \nabla_\theta\Lambda_{\theta(t)} = 0$. Notice that the value of $\Lambda^*$ depends on $\theta_0$. Also, notice that $\theta(t)$ may not converge.

Next, let us consider the discrete counterpart of \cref{eq:accnabla_main}, i.e.,
\begin{align}
	\theta_{m+1} = \theta_m - \gamma_m \nabla_\theta\Lambda_{\theta_m},
\end{align}
and apply the ideas in \cite{Marbach2001} to study its convergence.
Later, we show that given \cref{ass:1,ass:2,ass:3}, $\nabla_\theta\Lambda_{\theta}$ and $\nabla^2_\theta\Lambda_{\theta}$ are uniformly bounded (see \cref{cor:unfbound_lambdagrad}). Following the same idea as in the above ODE analysis, let us study the changes in $\Lambda_{\theta_m}$ based on the update rule of $\theta_m$. Using the Taylor expansion of $\Lambda_{\theta_m}$, we have
\begin{align}\label{eq:accnabla_Lambdatheta}
	\Lambda_{\theta_{m+1}} = \Lambda_{\theta_{m}} - \gamma_m \norm{\nabla_\theta\Lambda_{\theta_m}}^2 + O(\gamma_m^2).
\end{align}
Notice that $\norm{\theta_{m+1} - \theta_m}^2 = \gamma_m^2 \norm{\nabla_\theta\Lambda_{\theta_m}}^2 = O(\gamma_m^2)$ and that $\nabla_\theta\Lambda_{\theta}$ and $\nabla^2_\theta\Lambda_{\theta}$ are uniformly bounded. Suppose that $\sum_m \gamma_m^2 < \infty$. Hence, for any $\epsilon > 0$, there exists $m_0 > 0$ large enough so that $\Lambda_{\theta_{n}} < \Lambda_{\theta_{m}} + \epsilon$, for any $m > n > m_0$. Taking $\limsup$ from the left-hand side, and then $\liminf$ from the right-hand side of the previous equation, noticing that the choice of $\epsilon$ was arbitrary, we have $\Lambda_{\theta_m} \to \Lambda^*$. As before, notice that the value of $\Lambda^*$ depends on $\theta_0$ and $\{\gamma_m\}_{m\geq0}$. We want to show that $\nabla_\theta\Lambda_{\theta_m}\to 0$ as in the case of the ODE.

Assuming $\sum_{m}\gamma_m \to\infty$, $\cref{eq:accnabla_Lambdatheta}$ together with the fact that $\lim_{m\to\infty}\gamma_m = 0$ and the fact that $\Lambda_{\theta_m}$ converges, imply that $\liminf_{m\to\infty} \nabla_\theta\Lambda_{\theta_m} \to 0$. Next, we show that $\limsup_{m\to\infty} \nabla_\theta\Lambda_{\theta_m} \to 0$. Notice that for any $n > m > 0$, since $\nabla^2_\theta \Lambda_\theta$ is uniformly bounded,
\begin{align}\label{eq:accnabla_nablaLambdatheta}
	\norm{\nabla_\theta \Lambda_{\theta_n} - \nabla_\theta\Lambda_{\theta_m}} &\leq \text{Const}\times \norm{\theta_n - \theta_m} \leq \text{Const} \times \sum_{i=m}^{n-1} \gamma_i \norm{\nabla_\theta\Lambda_{\theta_i}}.
\end{align}
If $\limsup_{m\to\infty}\allowbreak \norm{\nabla_\theta\Lambda_{\theta_m} }= \epsilon > 0$, then there are infinity many $n>m$, such that $\norm{\nabla_\theta \Lambda_{\theta_m}} < \epsilon/3$, $\norm{\nabla_\theta \Lambda_{\theta_n}} > 2\epsilon/3$, and for all $i\in\{m+1,\cdots,n-1\}: \norm{\nabla_\theta\Lambda_{\theta_i}}\in(\epsilon/3,2\epsilon/3)$. For any such $m$ and $n$, by \cref{eq:accnabla_nablaLambdatheta} and the fact that $\gamma_i\to0$, we have $\sum_{i=m+1}^{n-1} \gamma_i \geq \text{Const}$, which implies
\begin{align}
	\Lambda_{\theta_n} &= \Lambda_{\theta_m} - \sum_{i=m}^{n-1} \gamma_i \norm{\nabla_\theta\Lambda_{\theta_i}}^2 + \sum_{i=m}^{n-1}O(\gamma_i^2)\leq \Lambda_{\theta_m} - \frac{\epsilon^2}{9} \text{Const} + \sum_{i=m}^{n-1}O(\gamma_i^2).
\end{align}
Contradiction follows by the fact that $\Lambda_{\theta_m}\to\Lambda^*$. Hence, $\nabla_\theta\Lambda_{\theta_m}\to 0$.% as in the case of the ODE.
\begin{remark}\label[remark]{rem:summability}
	As noted in \cite{Marbach2001} for their algorithm, the above analysis is not based on the convergence of the discrete-time algorithm to its continuous counterpart, since $\theta_m$ may not converge at all. Starting from the next section, we need more tools beyond the ones in \cite{Marbach2001} to handle the risk-sensitive cost.
\end{remark}

\subsection{Stochastic Gradient Algorithm}\label{subsec:heuristic_approxnabla}
Next we will assume that we have access to a noisy version of $\nabla_\theta\Lambda_\theta$ for any $\theta\in\real^l$. Notice that by \cref{eq:policygrad_reg}, we have $\nabla_\theta \Lambda_\theta \checkexpect{\check{\Phi}_0=x^*}{\theta}{\tau_{x^*}} = \expect{\Phi_0=x^*}{\theta}{F(\theta,\Lambda_\theta)}$, where the random function $F(\theta,\Lambda)$ is defined as follows:
\begin{align}
	F(\theta,\Lambda)\coloneqq\sum_{i=0}^{\tau_{x^*}-1} \big(\alpha \nabla_\theta C_\theta({\Phi}_i) + L_\theta(\Phi_i,\Phi_{i+1}) \big)\expon{\sum_{i=0}^{\tau_{x^*}-1} \left(\alpha C_\theta({\Phi}_i) - \Lambda\right)}.
\end{align}
Let $f(\theta,\Lambda) \coloneqq \expect{\Phi_0=x^*}{\theta}{F(\theta,\Lambda)} = \nabla_\theta g(\theta,\Lambda)$. Notice that by \cref{lem:propgf}, for any $(\theta,\Lambda)\in\real^{l+1}$ at which $g(\theta,\Lambda)<\infty$, $f(\theta,\Lambda)$ is well-defined. Suppose that for any $\theta\in\real^l$ the oracle generates an independent random sample distributed as $F(\theta,\Lambda_\theta)$. The discrete counterpart of \cref{eq:accnabla_main}, using a stochastic approximation of the gradient is then given by
\begin{align}
	\theta_{m+1} = \theta_m - \gamma_m F_m(\theta_m,\Lambda_{\theta_m})
\end{align}
where $\{F_m(\theta_m,\Lambda_{\theta_m})\}_{m}$ are independent realizations provided by the oracle. Note that this algorithm is still not realistic since the oracle needs the knowledge of $\Lambda_{\theta_m}$ to generate $F_m(\theta_m,\Lambda_{\theta_m}).$ One may attempt to use the same argument as in \cref{subsec:heuristic_main}, writing
\begin{align}
	%\Lambda_{\theta_{m+1}} &= \Lambda_{\theta_{m}} - \gamma_m \nabla_\theta\Lambda_{\theta_m}\cdot F_m(\theta_m,\Lambda_{\theta_m}) + O(\gamma_m^2 \norm{F_m(\theta_m,\Lambda_{\theta_m})}^2)\\
	\Lambda_{\theta_{m+1}}&=\Lambda_{\theta_{m}} - \gamma_m \nabla_\theta\Lambda_{\theta_m}\cdot f(\theta_m,\Lambda_{\theta_m}) \\
	&\myquad[4]- \gamma_m\nabla_\theta\Lambda_{\theta_m}\cdot (F_m(\theta_m,\Lambda_{\theta_m}) - f(\theta_m,\Lambda_{\theta_m})) + O(\gamma_m^2 \norm{F_m(\theta_m,\Lambda_{\theta_m})}^2)
\end{align}
Unless $\alpha$ is small, $\expect{\Phi_0=x^*}{\theta_m}{\norm{F_m(\theta_m,\Lambda_{\theta_m})}^2}$ can be $\infty$, and the error term may not be summable as in \cref{sec:costest} for policy evaluation.

Similar to \cref{sec:costest}, a remedy to the above problem is to consider a truncated version of the stochastic approximation of the gradient. However, truncating $F_m(\theta_m,\Lambda_{\theta_m})$ can be problematic as the expected value of the truncated version is no longer proportional to $\nabla_\theta\Lambda_\theta$. Another approach is to follow the direction of $\nabla_\theta\Lambda_\theta^{(M)}$ where $\Lambda_\theta^{(M)}$ was defined in \cref{sec:costest}. However, the hard truncation in the definition of $\Lambda_\theta^{(M)}$ is problematic: (i) the gradient may not exist and (ii) even if it exists, it may not be sufficiently smooth to ensure convergence when used in a stochastic approximation algorithm. Hence, we need to consider a smooth truncation and then develop a stochastic approximation of the gradient using a similar argument as in \cref{sec:policygrad}.

Abusing notation, for any $M>1$ let us redefine the function $g^{(M)}:\real^{l+1}\to\real_{++}$ by $g^{(M)}(\theta,\Lambda)\coloneqq \expect{\Phi_0=x^*}{\theta}{G^{(M)}(\theta,\Lambda)}$ where the random function $G^{(M)}(\theta,\Lambda)$ is given by
\begin{align}
		&G^{(M)} (\theta,\Lambda) \coloneqq  H(\theta,\Lambda) \mathbbm{1}\left\{H(\theta,\Lambda) \leq M\right\}+M \left[\sum_{i=0}^4 \frac{1}{i!}\left(\ln\left(M^{-1}H(\theta,\Lambda)\right)\right)^i\right]\mathbbm{1}\left\{H(\theta,\Lambda) > M\right\},
\end{align}
and the random function $H(\theta,\Lambda)$ is defined as $H(\theta,\Lambda) \coloneqq \exp\big(\sum_{i=0}^{\tau_{x^*}-1} \left(\alpha C_\theta({\Phi}_i) - \Lambda\right)\big).$
%\begin{align}
%	H(\theta,\Lambda) \coloneqq \expon{\sum_{i=0}^{\tau_{x^*}-1} \left(\alpha C_\theta({\Phi}_i) - \Lambda\right)}.
%\end{align}
In particular, $G^{(M)} (\theta,\Lambda)$ includes the first five terms in the Taylor expansion of $H(\theta,\Lambda)$ at $H(\theta,\Lambda) = M$. The reason why we need this many terms is somewhat technical and becomes clear later in the paper (see the discussion before \cref{cor:deltagmfm}).

Notice that $0 < G^{(M)} (\theta,\Lambda) < H(\theta,\Lambda)$, and for any fixed $\theta\in\real^l$ and $M > 1$, $G^{(M)} (\theta,\Lambda)$ is strictly decreasing in $\Lambda$ . Abusing the notation, let $\Lambda_\theta^{(M)}$ denote the unique solution of $g^{(M)} (\theta,\Lambda) =1$. Note that this new version of $\Lambda_\theta^{(M)}$ is larger than the one given in \cref{lem:unifapprox_truncost} and smaller than $\Lambda_\theta$. In particular, $\Lambda_\theta^{(M)}$ estimates $\Lambda_\theta$ uniformly (see \cref{cor:unifapprox_truncost}).

Next, we derive a sensitivity formula for $\Lambda_\theta^{(M)}$. Notice that $\Lambda_\theta^{(M)}$ does not have an eigenvalue interpretation, however, we can still derive a formula for $\nabla_\theta\Lambda_\theta^{(M)}$ following the same argument as in \cref{sec:policygrad}. For any $M>1$, define the function $f^{(M)}:\real^{l+1}\to\real^l$ by $f^{(M)} (\theta,\Lambda) \coloneqq \expect{\Phi_{0} = x^*}{\theta}{F^{(M)} (\theta,\Lambda)}$ where the random function $F^{(M)}(\theta,\Lambda)$ is given by
\begin{align}
	&F^{(M)} (\theta,\Lambda) \coloneqq  G^{(M)} (\theta,\Lambda) \sum_{i=0}^{\tau_{x^*}-1} L_\theta(\Phi_i,\Phi_{i+1})+\nabla_\theta G^{(M)} (\theta,\Lambda)\allowdisplaybreaks\\
	&\myquad[4]=G^{(M)} (\theta,\Lambda) \sum_{i=0}^{\tau_{x^*}-1} L_\theta(\Phi_i,\Phi_{i+1}) + \nabla_\theta H(\theta,\Lambda) \mathbbm{1}\left\{H(\theta,\Lambda) \leq M\right\}\allowdisplaybreaks\\
	&\myquad[6]  +M \frac{\nabla_\theta H(\theta,\Lambda)}{H(\theta,\Lambda)}\left[\sum_{i=0}^3 \frac{1}{i!}\left(\ln\left(M^{-1}H(\theta,\Lambda)\right)\right)^i\right]\mathbbm{1}\left\{H(\theta,\Lambda) > M\right\}.
\end{align}
Invoking the implicit function theorem and using a similar argument as in \cref{lem:propgf}, we get the following sensitivity formula: $f^{(M)} (\theta,\Lambda_\theta^{(M)}) = -\der{g^{(M)}(\theta,\Lambda)}\mathbin{/}\der{\Lambda}\big|_{\Lambda = \Lambda_\theta^{(M)}} \nabla_\theta \Lambda_\theta^{(M)}$  (see \cref{lem:propgmfm}). Notice that in order to use $f^{(M)} (\theta,\Lambda)$ to follow the direction of the gradient, we need to ensure that the term $-{\der{g^{(M)}(\theta,\Lambda)}}{\mathbin{/}}{\der{\Lambda}}$ is uniformly positive (see \cref{lem:uniflowerbd_dg}).

In order to follow the same steps as in \cref{subsec:heuristic_main}, we need to ensure that $\nabla_\theta \Lambda_\theta^{(M)}$ and $\nabla^2_\theta \Lambda_\theta^{(M)}$ are uniformly bounded. In \cref{lem:gradgfgmfm}, we first present a more general result which shows that many of the functions that we are interested in are uniformly bounded in a certain region. We then argue that $\nabla_\theta\Lambda_\theta^{(M)}$ and $\nabla^2_\theta\Lambda_\theta^{(M)}$ are uniformly bounded over $\theta\in\real^l$ and $M>1$ (see \cref{cor:unfbound_lambdagrad}).

Suppose that $M>1$ is fixed and the oracle for any $\theta\in\real^l$ generates independent random variables distributed according to $F^{(M)}(\theta,\Lambda_{\theta})$. Suppose that we update $\theta_m$ using
\begin{align}
{\theta_{m+1}} = {\theta_{m}} - \gamma_m F^{(M)}(\theta_m,\Lambda_{\theta_m}).    
\end{align}
Following the same argument as in \cref{subsec:heuristic_main}, using \cref{cor:unfbound_lambdagrad}, we can write
\begin{align}
	\Lambda^{(M)}_{\theta_{m+1}} &= \Lambda^{(M)}_{\theta_{m}} - \gamma_m \nabla_\theta\Lambda^{(M)}_{\theta_m}\cdot f^{(M)}(\theta_m,\Lambda^{(M)}_{\theta_m})\\
	&\myquad[4] - \gamma_m \nabla_\theta\Lambda^{(M)}_{\theta_m}\cdot \left(F^{(M)}(\theta_m,\Lambda^{(M)}_{\theta_m}) - f^{(M)}(\theta_m,\Lambda^{(M)}_{\theta_m})\right) \\
	&\myquad[4] + O\big(\gamma_m^2 \norm{F^{(M)}(\theta_m,\Lambda^{(M)}_{\theta_m})}^2\big).
\end{align}
Assuming $\sum_{m}\gamma_m\to \infty$ and $\sum_m \gamma_m^2 < \infty$, using the martingale convergence theorem, it is easy to show that both the stochastic noise as well as the error associated with the Taylor approximation are summable. Following the same argument as in \cref{subsec:heuristic_main}, we have $\Lambda^{(M)}_{\theta_m}$ converges to some $\Lambda^{*(M)}$ and that $\nabla_\theta\Lambda^{(M)}_{\theta_m}\to 0,$ both almost surely.

A natural question is whether $\nabla_\theta \Lambda_\theta^{(M)}$ is a good approximation to $\nabla_\theta \Lambda_\theta$. Same as in \cref{cor:unifapprox_truncost}, we show that $\nabla_\theta \Lambda_\theta^{(M)}$ approximates $\nabla_\theta \Lambda_\theta$ uniformly (see \cref{cor:unifapprox_nablatheta}). In particular, given $\epsilon>0,$ there exists a sufficiently large $M$ such that
$\lim_{m\rightarrow\infty}\nabla_\theta \Lambda^{(M)}_{\theta_m}=0$ and $\lim_{m\rightarrow\infty}\|\nabla_\theta \Lambda_{\theta_m}\|\leq \epsilon$ both almost surely.

As in the case of policy evaluation, we now consider the question of whether $\nabla_\theta\Lambda_{\theta_m}\rightarrow 0$ a.s. by considering an increasing sequence $\{M_m\}_{m\geq 0}$. Suppose that the oracle for any $\theta\in\real^l$ and $M>1$, generates independent random variables distributed according to $F^{(M)}(\theta,\Lambda_{\theta})$. Following the same steps as before, using the truncation sequence $\{M_m\}$, we have
\begin{align}
	%\Lambda^{(M_{m+1})}_{\theta_{m+1}} &= \Lambda^{(M_{m})}_{\theta_{m+1}} + \Lambda^{(M_{m+1})}_{\theta_{m+1}} - \Lambda^{(M_{m})}_{\theta_{m+1}}\\
	\Lambda^{(M_{m+1})}_{\theta_{m+1}}&= \Lambda^{(M_{m})}_{\theta_{m}} - \gamma_m \nabla_\theta\Lambda^{(M_{m})}_{\theta_m}\cdot f^{(M_m)}(\theta_m,\Lambda^{(M_{m})}_{\theta_m})\\
	&\myquad[4] - \gamma_m \nabla_\theta\Lambda^{(M_{m})}_{\theta_m}\cdot \left(F^{(M_m)}(\theta_m,\Lambda^{(M_{m})}_{\theta_m}) - f^{(M_m)}(\theta_m,\Lambda^{(M_{m})}_{\theta_m})\right) \\
	&\myquad[4] + O\big(\gamma_m^2 \norm{F^{(M_m)}(\theta_m,\Lambda^{(M_{m})}_{\theta_m})}^2\big) +  \Lambda^{(M_{m+1})}_{\theta_{m+1}} - \Lambda^{(M_{m})}_{\theta_{m+1}}.
\end{align}
Hence, in order to ensure the error term is summable, we need to have $\sum_m\Lambda^{(M_{m+1})}_{\theta_{m+1}} - \Lambda^{(M_{m})}_{\theta_{m+1}} < \infty$ a.s. in addition to other assumptions on the step sizes. Since we do not know the sequence $\{\theta_m\}$ a priori, we need the following assumption on the truncation sequence.
\begin{assumption}\label[assumption]{ass:4}
	$\{M_m\}_{m\geq 0}$ satisfies $M_m\uparrow\infty$ and $\sum_{m=0}^\infty \sup_{\theta\in\real^l} \Lambda_{\theta}^{(M_{m+1})} - \Lambda_{\theta}^{(M_{m})} < \infty$.
\end{assumption}
Given \cref{ass:4} and assuming $\sum_m\gamma_m^2M_m^2<\infty$, it is easy to see that the exact same argument together with \cref{cor:unifapprox_truncost,cor:unifapprox_nablatheta}, yields $\Lambda_{\theta_m} \to \Lambda^*$ and that $\nabla_\theta\Lambda_{\theta_m}\to 0$. It is easy to see that there are truncation sequences for which \cref{ass:4} holds. In particular, let $M_m$ to be such that $\sup_\theta \Lambda_{\theta} - \Lambda_{\theta}^{(M_{m})} < {(m+1)}^{-2}$. The existence of such $M_m$ is guaranteed by \cref{cor:unifapprox_truncost}. Then, $\{M_m\}_{m}$ satisfies \cref{ass:4}. 
%An open question is how to find such a sequence. In practice, however, one can use a sufficiently large fixed $M$ and obtain a $\theta$ where the policy gradient is close to zero which is the best that one can do anyway if one were to run any algorithm for a finite amount of time. But the open question is nevertheless interesting and a topic for future work.

\subsection{An ODE Version of Our Trajectory-Based Algorithm}\label{subsec:heuristic_fg}

The main algorithm considered in this paper will be presented in the next section. Here we present an ODE version of the algorithm where we have to estimate the cost and optimize over $\theta\in\real^l$ simultaneously.
Suppose that we have access to an oracle that for any $\theta\in\real^l$ and $\Lambda\in\real$ returns the value of $f(\theta,\Lambda)$ and $g(\theta,\Lambda)$. Notice that in this scenario, we need to approximate $\Lambda_\theta$ as we update the value of $\theta$. Starting from $\theta_0$ and setting $\tilde{\Lambda}_0$ large enough so that $g(\theta_0,\tilde{\Lambda}_0)<\infty$, a natural way to optimize the risk-sensitive cost is to combine the ODEs given in \cref{sec:costest} and \cref{subsec:heuristic_main}, and run the following system of ODEs:
\begin{align}
	\begin{cases}
		\dot{\theta}(t) &= -f(\theta(t),\tilde{\Lambda}(t))\\
		\dot{\tilde{\Lambda}}(t) &= g(\theta(t),\tilde{\Lambda}(t)) -1
	\end{cases},\qquad (\theta(0),\tilde{\Lambda}(0)) = (\theta_0,\Lambda_0).
\end{align}
Following the same argument as in \cref{subsec:heuristic_main}, we want to study changes in $\Lambda_{\theta(t)}$ as we update $\theta(t)$. By the Chain rule, we have $\dot{\Lambda}_{\theta(t)} = -f(\theta(t),\tilde{\Lambda}(t)) \cdot \nabla_\theta\Lambda_{\theta(t)}$.
%\begin{align}
%	\dot{\Lambda}_{\theta(t)} = -f(\theta(t),\tilde{\Lambda}(t)) \cdot \nabla_\theta\Lambda_{\theta(t)}
%\end{align}
In particular, we can follow the exact same argument as in \cref{subsec:heuristic_main} if $\tilde{\Lambda}(t) = \Lambda_{\theta(t)}$. Hence, our strategy is to first show that $\tilde{\Lambda}(t) - \Lambda_{\theta(t)} \to 0$.

Notice that if at any point $\tilde{\Lambda}(t_0) = \Lambda_{\theta(t_0)}$, then $\dot{\Lambda}_{\theta(t_0)} = -\checkexpect{\check{\Phi}_0=x^*}{\theta}{\tau_{x^*}} \norm{\nabla_\theta\Lambda_{\theta(t_0)}}^2 \leq 0$ and $\dot{\tilde{\Lambda}}(t_0) = 0$. If $\norm{\nabla_\theta\Lambda_{\theta(t_0)}} = 0$, we have then reached a stationary point of the system. Otherwise, $\der{\big(\tilde{\Lambda}_{t} - \Lambda_{\theta(t)}\big)}\mathbin{/}\der{t}\big|_{t = t_0} > 0$,
%\begin{align}
%	\frac{\der{\left(\tilde{\Lambda}_{t} - \Lambda_{\theta(t)}\right)}}{\der{t}}\bigg|_{t = t_0} > 0
%\end{align}
and for all small enough $\epsilon > 0$, we have $\tilde{\Lambda}(t_0+\epsilon)> \Lambda_{\theta(t_0+\epsilon)}$. Hence, we have $\tilde{\Lambda}(t)\geq \Lambda_{\theta(t)} $ for all $t\geq t_0$.

In particular, either $\tilde{\Lambda}(t) < \Lambda_{\theta(t)}$ for all $t\geq 0$, or there exists $t_0 \geq 0$ such that $ \tilde{\Lambda}(t) \geq \Lambda_{\theta(t)} $ for all $t\geq t_0$. Either case, $\tilde{\Lambda}(t)$ is monotone after some time and since it is bounded, it converges. This implies that $\dot{\tilde{\Lambda}}(t) \to 0$, which in turn implies that $\tilde{\Lambda}(t)-\Lambda_{\theta(t)}\to0$.

While this ODE proof is very similar to the proof in \cite{Marbach2001}, the discrete time stochastic counterpart of the above algorithm faces the same difficulties (and more) that were discussed in \cref{sec:costest,subsec:heuristic_approxnabla}. %Therefore, we study a sequence of truncations of the discrete time stochastic counterpart in the next section.

	\section{A Trajectory-Based Gradient Algorithm}\label{sec:algorithm}
	In this section, we propose our trajectory-based gradient algorithm which is a discrete-time stochastic counterpart of the algorithm discussed in \cref{subsec:heuristic_fg}. Following the same logic as in \cref{subsec:heuristic_approxnabla}, we consider a smooth truncated approximation of the risk-sensitive cost and its gradient. We study both fixed and varying truncations.

Suppose we start with an initial parameter $\theta_0$ and an estimate of the risk-sensitive cost $\tilde{\Lambda}_0$. We sample the Markov chain $\{\Phi_i\}_{i\geq 0}$ according to the policy $\theta_0$ up to the time of the first visit to the recurrent state $x^*$. We then update our policy and our estimation of the risk-sensitive cost and then repeat this cycle. Let $t_0 = 0$ and $t_m$ denote $m$th visit to the recurrent state $x^*$. Let $(\theta_m,\widetilde{\Lambda}_m)$ for $m\in\nzero$ denote the pair of policy and estimated cost at time-step $t_m$. Notice that the Markov chain follows the policy $\theta_m$ during the time-interval $\{t_m,\cdots,t_{m+1}-1\}$. At time-step $t_{m+1}$ we update the policy and estimate the cost as follows:
\begin{align}
	\begin{aligned}
		&\theta_{m+1} = \theta_m - \gamma_m F_m \left(\theta_m,\tilde{\Lambda}_m\right),\\
		&\tilde{\Lambda}_{m+1} = \tilde{\Lambda}_m + \eta \gamma_m \left(G_m\left(\theta_m,\tilde{\Lambda}_m\right) - 1\right),
	\end{aligned}\label{eq:algorithm}
\end{align}
where $\eta > 0$ is a tunable parameter, $G_m (\theta,\Lambda)$ and $F_m (\theta,\Lambda)$ are trajectory-based versions of $G^{(M_m)} (\theta,\Lambda)$ and $F^{(M_m)} (\theta,\Lambda)$ (see \cref{subsec:heuristic_approxnabla}) respectively, obtained by observing the trajectory from time-step $t_m$ to time-step $t_{m+1}-1$, and $\{M_m\}$ is a truncation sequence.
%\begin{align}
%	&G_m (\theta,\Lambda) \coloneqq  H_m(\theta,\Lambda) \mathbbm{1}\left\{H_m(\theta,\Lambda) \leq M_m\right\} \\
%	&\myquad[10]+ M_m \left[\sum_{i=0}^4 \frac{1}{i!}\left(\ln\left(M_m^{-1}H_m(\theta,\Lambda)\right)\right)^i\right]\mathbbm{1}\left\{H_m(\theta,\Lambda) > M_m\right\},\allowdisplaybreaks\\
%	&F_m (\theta,\Lambda) \coloneqq \nabla_\theta G_m (\theta,\Lambda) + G_m (\theta,\Lambda) \sum_{i=t_m}^{t_{m+1}-1} L_\theta(\Phi_i,\Phi_{i+1}),\allowdisplaybreaks\\
%	&H_m(\theta,\Lambda) \coloneqq \expon{\sum_{i=t_m}^{t_{m+1}-1} \left(\alpha C_\theta({\Phi}_i) - \Lambda\right)},
%\end{align}
Notice that $G_m (\theta,\Lambda) \myeq{d} G^{(M_m)} (\theta,\Lambda)$ and $F_m (\theta,\Lambda) \myeq{d} F^{(M_m)} (\theta,\Lambda)$. Let $g_m(\theta,\Lambda)$ and $f_m (\theta,\Lambda)$ denote the expected values of $G_m (\theta,\Lambda)$ and $F_m (\theta,\Lambda)$ respectively. We assume that the step-size sequence $\{\gamma_m\}$ satisfies the following standard assumption.
\begin{assumption}\label[assumption]{ass:5}
	$\{\gamma_m\}_{m\geq0}$ satisfies $\sum_m \gamma_m \to \infty$, $\sum_m \gamma_m^2 < \infty$, and $\gamma_{\max} \coloneqq \sup_m \gamma_m < 1$.
\end{assumption}

To avoid large jumps and improve the numerical stability of the algorithm, we may use the following update rule instead:
\begin{align}\label{eq:algorithm_project}
	\begin{aligned}
		&\theta_{m+1} = \theta_m - \Gamma_{\theta_{\text{thresh}}}\left(\gamma_m F_m \left(\theta_m,\tilde{\Lambda}_m\right)\right),\\
		&\tilde{\Lambda}_{m+1} = \tilde{\Lambda}_m + \Gamma_{\Lambda_{\text{thresh}}}\left(\eta \gamma_m \left(G_m\left(\theta_m,\tilde{\Lambda}_m\right) - 1\right)\right),
	\end{aligned}
\end{align}
where $\theta_{\text{thresh}}$ and $\Lambda_{\text{thresh}}$ are fixed positive constants, and $\Gamma_a(x)$ is the element-wise projection of $x$ to $[-a,a]$. Either case, based on whether the $\{M_m\}$ is a fixed sequence or an increasing one, we get a similar result as in \cref{subsec:heuristic_approxnabla}. The proof of \cref{thm:main} is presented in \cref{proof:thm:main}.
\begin{theorem}\label[theorem]{thm:main}
	Let \cref{ass:1,ass:2,ass:3,ass:5} hold. Let $\{\theta_m\}$ be the sequence generated by \cref{eq:algorithm} (or \cref{eq:algorithm_project}).\newline
	{\normalfont (i)} Suppose that $\{M_m\}_{m\geq0}$ is a fixed sequence, i.e., $M_m = M$ for all $m\geq0$. Then $\Lambda^{(M)}_{\theta_{m}}$ converges and $\nabla_\theta \Lambda^{(M)}_{\theta_{m}} \to 0$ with probability $1$.\newline
	{\normalfont (ii)}  Suppose that $\sum_m \gamma_m^{2(1-\beta)} < \infty$ for some fixed $\beta\in(0,0.5)$, and the truncation sequence is given by $M_m \coloneqq \sup\left\{N_i: N_i \leq \frac{1}{\gamma_m^{\beta}} \right\}$, where $\{N_i\}_{i\geq0}$ is a sequence for which \cref{ass:4} hold. Then $\Lambda_{\theta_{m}}$ converges and $\nabla_\theta \Lambda_{\theta_{m}} \to 0$ with probability $1$.
\end{theorem}
The results in the paper so far have been presented for a Markov chain with a parameter $\theta.$ It is straightforward to add the formalism of an MDP and derive an algorithm for finding the stationary point of a parameterized set of control policies. This is presented in \cref{sec:riskmdp}.

	\section{Conclusion}
	We have derived a policy gradient algorithm for the exponential-cost infinite horizon risk-sensitive MDP. A key challenge in the risk-sensitive case is that the stochastic noise may not be summable and hence, standard stochastic approximation theory does not apply. Therefore, we consider a truncated version of the risk-sensitive cost. This alone does not solve the problem either because again standard stochastic approximation theory is not applicable due to the fact that the truncated version is not sufficiently smooth. So we define a truncated and smooth approximation to the cost and show that this version of the cost provides a uniform bound on the original risk-sensitive cost and further, one can use a variant of stochastic approximation presented in \cite{Marbach2001} to prove convergence of the policy gradient algorithm. An additional interesting question that we consider is whether we can obtain a solution to the untruncated version of the problem: we show that there exists a sequence of truncations such that the policy gradient algorithm will achieve a stationary point of the original risk-sensitive cost. In preliminary simulations, we have observed an increasing sequence of truncations of the form $\gamma_m^{-\beta}$ always works. An interesting open question is to show that this sequence satisfies the conditions provided in our main theorem. We also note that this sequence works for the problem of policy evaluation as shown in \cref{prop:costest}.

	\bibliographystyle{plain}
	\bibliography{bibliography}

	\appendix
	\section{Background}\label{sec_app:background}
	In this section, we present the necessary background to ensure that the paper is self-contained. Most of these results hold in a more general setting; however, we present a more straightforward and intuitive proofs by restricting our attention to finite-state Markov chains. In \cref{sec_app:back_risk}, we define the risk-sensitive exponential cost and derive a multiplicative counterpart of the traditional Bellman equation. We also define the risk-sensitive value function and provide a fixed point equation based on visits to a recurrent state. This fixed point equation is the basis of our approximation in the paper. The results presented in this section holds for countable state-space Markov chains \cite{Balaji2000,borkar2002risk} as well as continuous state-space Markov processes \cite{Kontoyiannis2003}. In \cref{sec_app:back_robust}, we derive a variational formula for the risk-sensitive Markov chain and discuss the connection between different risk measures. In particular, we show that the risk-sensitive cost takes into account model uncertainties. A similar formula holds in the case of continuous state-space \cite{Anantharam2017}, under some restrictions. For the sake of notational simplicity, we focus on Markov chains only; generalization to Markov decision processes follows readily.

In the following subsections, $v(x)$ and $[v]_x$ are used interchangeably to denote the $x$'th element of a vector $v$. For a matrix $P$, we use $P(x,:)$ to denote its $x$'th row, and $P(:,y)$ to denote its $y$'th column. Real-valued functions are sometimes treated as vectors. The set of all recurrent and aperiodic transition matrices on the state-space $\samplespace$ are denoted by $\matspace(\samplespace)$. For any $P\in\matspace(\samplespace)$, we use $\pi_P$ to denote its unique stationary distribution.
	\subsection{Risk-Sensitive Exponential Cost}\label{sec_app:back_risk}
	Let $\{\Phi_i\}_{i\geq 0}$ denote a discrete-time Markov chain on a finite space $\samplespace = \{1,2,\cdots,\left|\samplespace\right|\}$ that is recurrent and aperiodic. Suppose that the transition probability of $\{\Phi_i\}_{i\geq 0}$ is given by
\begin{align}
	P(x,y)\coloneq\prob{\Phi_{i+1} = y | \Phi_i = x}\qquad\forall x,y\in\samplespace \text{ and }k\in\nplus.
\end{align}
Let $C:\samplespace \to \real$ denote the one-step cost function, \ie, $C(x)$ is the cost we incur each time the chain visits the state $x\in\samplespace$.
The risk-sensitive cost of the Markov chain $\{\Phi_i\}_{i\geq 0}$ with risk factor $\alpha > 0$ is defined as follows:
\begin{align}\label{eq_app:riskcost}
	\Lambda(\alpha) \coloneq \lim_{n\to\infty} \frac{1}{n} \ln\expect{\Phi_0 = x}{P}{\exp\left(\alpha \sum_{i=0}^{n-1} C(\Phi_i)\right)},\qquad \forall x\in\samplespace
\end{align}
where $\expect{\Phi_0 = x}{P}{\cdot}$ denote the expectation with respect to the transition probability $P$ given $\Phi_0 = x$. Notice that the risk-sensitive cost penalizes sample paths with high costs.
\begin{proposition}\label[proposition]{prop_app:riskcost}
	For any $\alpha > 0$, $\expon{\Lambda(\alpha)}$ is the largest eigenvalue of $\hat{P} \coloneqq \text{diag}\left(\e^{\alpha C}\right)  P$. In particular, the right-hand side of \cref{eq_app:riskcost} converges to the same value, irrespective of initial state.
\end{proposition}
\begin{proof}
	Let us rewrite the right-hand side of \cref{eq_app:riskcost} in the matrix form:
	\begin{align}\label{eq_app:thm_risk_tmp}
		\frac{1}{n}\ln\expect{\Phi_0 = x}{P}{\exp\left(\alpha \sum_{i=0}^{n-1} C(\Phi_i)\right)} &= \frac{1}{n}\ln\left[\left(\text{diag}\left(\e^{\alpha C}\right)  P\right)^{n-1}  \e^{\alpha C} \right]_{x}\\
		&= \frac{1}{n}\ln\left[\hat{P}^{n-1}  \e^{\alpha C} \right]_{x}
	\end{align}
	where $P$ is the transition probability matrix and $\e^{\alpha C}$ is a vector with elements $\left[\e^{\alpha C}\right]_{x} = \e^{\alpha C(x)}$ for any $x\in\samplespace$.  Notice that $\hat{P}$ is non-negative, irreducible and aperiodic, i.e., primitive. Hence, by Perron–Frobenius theorem, $\hat{P}$ has a simple positive eigenvalue $\hat{\lambda}_1$ which is larger than all the other eigenvalues.

	Let $J$ denote the Jordan normal form of $\hat{P}$. In particular, there exists a matrix $Q$ such that $\hat{P} = Q J Q^{-1}$. By the Perron–Frobenius theorem, $J$ is given as follows:
	\begin{equation}
		J=
		\begin{bNiceArray}{>{\strut}cccccccc}[margin,rules/color=black]
			\hat{\lambda}_1     		&                       	 && &\Block{3-3}<\LARGE>{\boldsymbol{0}}&      					&& \\
										& \Block[draw]{2-3}<\large>{B_2}     && &                        	&                    	&& \\
										&                            && &				   	 	   	&				       	&& \\
			\Block{3-3}<\LARGE>{\boldsymbol{0}}	&                            && &\ddots        		   		&				       	&& \\
									    &                            && &                        	& \Block[draw]{2-3}<\large>{B_k}&& \\
										&                            && &                        	&                     	&&
		\end{bNiceArray},~
		B_l =
		\begin{bNiceMatrix}[nullify-dots,xdots/line-style=loosely dotted]
			\hat{\lambda}_l & 1               & 0              &                 & \Cdots          & 0 \\
			0               & \hat{\lambda}_l & 1              & \Ddots          &                 & \Vdots \\
			\Vdots          & \Ddots          & \Ddots         &\Ddots           &                 &  \\
			                &                 &                &                 &                 & 0\\
   			                &                 &                &                 &                 & 1\\
 			0               & \Cdots          &                &                 & 0               & \hat{\lambda}_l
		\end{bNiceMatrix}
	\end{equation}
	where for any $l$, $B_l$ is the Jordan block matrix corresponding to eigenvalue $\hat{\lambda}_l$ of $\hat{P}$. Notice that $\left|\hat{\lambda}_l\right| < \hat{\lambda}_1$ for any $l>1$. Rewriting \cref{eq_app:thm_risk_tmp} using the Jordan normal form, we have
	\begin{align}
		\frac{1}{n}\ln\expect{\Phi_0 = x}{P}{\exp\left(\alpha \sum_{i=0}^{n-1} C(\Phi_i)\right)}
		&= \frac{1}{n}\ln\left[Q J^{n-1} Q^{-1}  \e^{\alpha C} \right]_{x}\\
		&= \frac{n-1}{n}\ln\hat{\lambda}_1   + \frac{1}{n}\ln\left[Q  \left( \frac{J}{\hat{\lambda}_1} \right)^{n-1} Q^{-1}  \e^{\alpha C} \right]_{x},
	\end{align}
	It is easy to verify that for all large enough $n$, all elements of $\left( {J}\middle/{\hat{\lambda}_1} \right)^{n-1}$ are bounded by $1$. In particular, for all $x,y\in\samplespace$ with $y\geq x$ we have
	\begin{align}
		\left[\left( \frac{B_l}{\hat{\lambda}_1} \right)^{m}\right]_{x,y} &= {m \choose y-x }  \frac{\hat{\lambda}_l^{m - y + x}}{\hat{\lambda}_1^m}\\
		&\leq  \left(\frac{\e m}{y-x}\right)^{y-x}   \frac{\hat{\lambda}_l^{m - y + x}}{\hat{\lambda}_1^m} \xrightarrow{m\to\infty} 0,\qquad\forall l > 1.
	\end{align}
	Hence, we have
	\begin{align}
		\lim_{n\to\infty}\frac{1}{n}\ln\expect{\Phi_0 = x}{P}{\exp\left(\alpha \sum_{i=0}^{n-1} C(\Phi_i)\right)} =\ln \hat{\lambda}_1.
	\end{align}
\end{proof}
Notice that by Perron–Frobenius theorem, the eigenvectors corresponding to the eigenvalue $\expon{\Lambda(\alpha)}$ of $\hat{P}$ are strictly positive and unique up to a scaling factor. Let us denote this eigenvector with $h$. We have
\begin{align}\label{eq_app:risk_valuefunc}
	h(x) &= \frac{\expon{\alpha C(x)}}{\expon{\Lambda(\alpha)}}   \sum_{y\in\samplespace} P(x,y)h(y)\\
	&= \frac{\expon{\alpha C(x)}}{\expon{\Lambda(\alpha)}}   \expect{\Phi_0 = x}{P}{h(\Phi_1)},\qquad\forall x\in\samplespace.
\end{align}
The above equation is the multiplicative Poisson equation, and $h(\cdot)$ is its unique non-zero solution up to a scaling factor. Notice that the multiplication Poisson equation resembles the Bellman equation for the average cost problem, replacing `$\times$' with `$+$' and `$\div$' with `$-$'. Next, we show that $h(\cdot)$ is indeed the relative value function associated with the risk-sensitive cost problem, defined similarly to the relative value function in the average cost problem.
\begin{proposition}\label[proposition]{prop_app:valuefunc}
	Any solution of the multiplicative Poisson equation satisfies
	\begin{align}\label{eq_app:risk_valuefunc_limit}
		h(x) \,\propto\, \lim_{n\to\infty} \expect{\Phi_0 = x}{P}{\exp\left( \sum_{i=0}^{n-1} \left(\alpha C(\Phi_i) - \Lambda(\alpha)\right)\right)}
	\end{align}
\end{proposition}
\begin{proof}
	Let us ``twist'' the matrix $\hat{P}$, and define a ``twisted kernel'' as follows:
	\begin{align}
		\check{P}(x,y)\coloneqq \frac{\hat{P}(x,y) h(y)}{\hat{\lambda}_1 h(x)} = \frac{\e^{\alpha C(x)}P(x,y) h(y)}{\e^{\Lambda(\alpha)}h(x)}
	\end{align}
	where $\hat{\lambda}_1 = \expon{\Lambda(\alpha)}$ is the largest eigenvalue of $\hat{P}$, with all positive right eigenvector $h$. Notice that by \cref{eq_app:risk_valuefunc}, $\check{P}$ is a stochastic matrix. Let us rewrite \cref{eq_app:risk_valuefunc_limit} in terms of the twisted kernel $\check{P}$. We have
	\begin{align}
		\expect{\Phi_0 = x}{P}{\exp\left( \sum_{i=0}^{n-1} \left(\alpha C(\Phi_i) - \Lambda(\alpha)\right)\right)} = h(x) \expect{\check{\Phi}_0 = x}{\check{P}}{\frac{1}{h(\check{\Phi}_n)}},
	\end{align}
	where $\{\check{\Phi}_i\}_{i\geq 0}$ denotes a Markov chain with transition probability $\check{P}$, and $\expect{\check{\Phi}_0 = x}{\check{P}}{\cdot}$ denote the expectation with respect to $\check{P}$ given $\check{\Phi}_0 = x$. Notice that $\check{P}$ is recurrent and aperiodic. Hence, $\{\check{\Phi}_i\}_{i\geq 0}$ is ergodic and we have $\expect{\check{\Phi}_0 = x}{\check{P}}{{1}/{h(\check{\Phi}_n)}} \to \sum_{y\in\samplespace} \check{\pi}(y)/h(y)$, where $\check{\pi}$ is the stationary distribution of $\{\check{\Phi}_i\}_{i\geq 0}$. Alternatively, we can write
	\begin{align}
		\expect{\check{\Phi}_0 = x}{\check{P}}{\frac{1}{h(\check{\Phi}_n)}} = \left[\check{P}^n  v\right]_x
	\end{align}
	where $v(y) \coloneqq 1/h(y)$ for any $y\in\samplespace$. Notice that $v\cdot \check{\pi}\neq 0 $, and that $\check{\pi}$ is the left eigenvector corresponding to the eigenvalue $1$ of $\check{P}$. Invoking Perron–Frobenius theorem for $\check{P}$ and applying the power iteration analysis, we have $\check{P}^n  v \to \sum_{y\in\samplespace} \check{\pi}(y)/h(y) \boldsymbol{1}$, where $\boldsymbol{1}$ is the all-ones vector.
\end{proof}
Notice that \cref{prop_app:riskcost,prop_app:valuefunc} do not suggest any practical way to estimate the cost and the relative value function associated with the risk-sensitive cost problem. Similar to the average cost problem, one may hope that visits to a recurrent state might be useful for estimating these quantities. Next, using the twisted kernel $\check{P}(x,y) = {\hat{P}(x,y) h(y)}\mathbin{/}{\hat{\lambda}_1 h(x)}$, we show that this intuition is indeed correct.
\begin{corollary}\label[corollary]{cor_app:fixedpoint}
	Let $x^*\in\samplespace$ denote a recurrent state. The risk-sensitive cost $\Lambda(\alpha)$ is the unique fixed point of the following equation:
	\begin{align}
		\expect{\Phi_0 = x^*}{P}{\exp\left( \sum_{i=0}^{\tau_{x^*}-1} \left(\alpha C(\Phi_i) - \Lambda\right)\right)} = 1,
	\end{align}
	where $\tau_{x^*}\coloneqq \inf \{i\in\nplus:\Phi_i = x^* \}$ is the first return time to $x^*$. Moreover, for any $x\in\samplespace$, the risk-sensitive relative value function is given by
	\begin{align}
		h(x) \,=\, \expect{\Phi_0 = x}{P}{\exp\left( \sum_{i=0}^{\tau_{x^*}-1} \left(\alpha C(\Phi_i) - \Lambda(\alpha)\right)\right)}.
	\end{align}
\end{corollary}
\begin{proof}
	Notice that
	\begin{align}
		\expect{\Phi_0 = x^*}{P}{\exp\left( \sum_{i=0}^{\tau_{x^*}-1} \left(\alpha C(\Phi_i) - \Lambda\right)\right)} = \expect{\check{\Phi}_0 = x^*}{\check{P}}{\exp \left(\Lambda(\alpha) - \Lambda\right)},
	\end{align}
	which proves the first part. Similarly, by the definition of the twisted kernel, we have
	\begin{align}
		\expect{\Phi_0 = x}{P}{\exp\left( \sum_{i=0}^{\tau_{x^*}-1} \left(\alpha C(\Phi_i) - \Lambda(\alpha)\right)\right)} = h(x) \expect{\check{\Phi}_0 = x}{\check{P}}{1/h(\check{\Phi}_{\tau_{x^*}})}.
	\end{align}
\end{proof}
We wrap up this section by characterizing the limiting behavior of the risk-sensitive cost, as the risk factor $\alpha$ approaches $0$ and $\infty$.
\begin{corollary}
	$\lim_{\alpha\to0} \Lambda(\alpha)/\alpha$ and $\lim_{\alpha\to\infty} \Lambda(\alpha)/\alpha$, respectively, are the average cost and the maximum cost associated with the Markov chain $\{\Phi_i\}_{i\geq0}$. Moreover, $\Lambda(\alpha)/\alpha$ is an increasing function of $\alpha$.
\end{corollary}
\begin{proof}
	Notice that $\expon{\Lambda(\alpha)}$ is the largest eigenvalue of $\hat{P}$, and it is a differentiable function of $\alpha$. Moreover, $\lim_{\alpha\to 0} \expon{\Lambda(\alpha)} = 1$. Hence, $\lim_{\alpha\to0} \Lambda(\alpha)/\alpha = \frac{\der{\Lambda(\alpha)}}{\der{\alpha}}\big|_{\alpha=0}$. By \cref{cor_app:fixedpoint}, we have
	\begin{align}
		0 &= \frac{\der{}}{\der{\alpha}}\expect{\Phi_0 = x^*}{P}{\exp\left( \sum_{i=0}^{\tau_{x^*}-1} \left(\alpha C(\Phi_i) - \Lambda(\alpha)\right)\right)}\\
		&= \expect{\Phi_0 = x^*}{P}{\sum_{i=0}^{\tau_{x^*}-1} \left( C(\Phi_i) - \frac{\der{\Lambda(\alpha)}}{\der{\alpha}}\right)\exp\left( \sum_{i=0}^{\tau_{x^*}-1} \left(\alpha C(\Phi_i) - \Lambda(\alpha)\right)\right)} \\
		&\xrightarrow{\alpha\to0} \expect{\Phi_0 = x^*}{P}{\sum_{i=0}^{\tau_{x^*}-1} \left( C(\Phi_i) - \frac{\der{\Lambda(\alpha)}}{\der{\alpha}}\bigg|_{\alpha=0}\right)}
	\end{align}
	Hence,
	\begin{align}
		\frac{\der{\Lambda(\alpha)}}{\der{\alpha}}\bigg|_{\alpha=0} =  \expect{\Phi_0 = x^*}{P}{\sum_{i=0}^{\tau_{x^*}-1} C(\Phi_i)}\mathbin{\bigg/}\expect{\Phi_0 = x^*}{P}{\tau_{x^*}}
	\end{align}
	which is the the average cost associated with the Markov chain $\{\Phi_i\}_{i\geq 0}$. For the other limit, notice that
	\begin{align}
		\lim_{\alpha\to\infty} \hat{P}/\e^{\alpha C_{\max}} = \text{diag}(u)P,\qquad
		u(x) =
		\begin{cases}
			1&\text{if } C(x) = C_{\max}\\
			0&\text{o.w.}
		\end{cases}
	\end{align}
	where $C_{\max} \coloneqq \max_{x\in\samplespace} C(x)$. It is easy to verify that the largest eigenvalue of $\text{diag}(u)P$ is $1$. Hence, $\lim_{\alpha\to\infty} \expon{\Lambda(\alpha) - \alpha C_{\max}} = 1$. Finally, for any $0 < \alpha_1 < \alpha_2$, we have
	\begin{align}
		\frac{\Lambda(\alpha_1)}{\alpha_1}  &= \lim_{n\to\infty} \frac{1}{\alpha_1 n} \ln\expect{\Phi_0 = x}{P}{\exp\left(\alpha_1 \sum_{i=0}^{n-1} C(\Phi_i)\right)}\\
		&= \lim_{n\to\infty} \frac{1}{\alpha_2 n} \ln\left(\expect{\Phi_0 = x}{P}{\exp\left(\alpha_1 \sum_{i=0}^{n-1} C(\Phi_i)\right)}\right)^{\frac{\alpha_2}{\alpha_1}}\\
		&\geq \lim_{n\to\infty} \frac{1}{\alpha_2 n} \ln\expect{\Phi_0 = x}{P}{\exp\left(\alpha_2 \sum_{i=0}^{n-1} C(\Phi_i)\right)} = \frac{\Lambda(\alpha_2)}{\alpha_2}\\
	\end{align}
	where the inequality follows by Jensen's inequality.
\end{proof}

	\subsection{Robustness and Risk-Sensitive Exponential Cost}\label{sec_app:back_robust}
	As we pointed out, $\Lambda(\alpha)/\alpha$ is a bounded continuous increasing function of $\alpha$. This quantity is called the entropic risk measure \cite{Follmer2008,Follmer2011}, and is denoted by $e(\alpha)$:
\begin{align}\label{eq_app:def_entropicrisk}
	&e(\alpha) \coloneqq \frac{1}{\alpha} \Lambda(\alpha) = \lim_{n\to\infty} \frac{1}{\alpha n} \ln\expect{\Phi_0 = x}{P}{\exp\left(\alpha \sum_{i=0}^{n-1} C(\Phi_i)\right)},\qquad \forall \alpha > 0\\
	&e(0) \coloneqq \lim_{\alpha\to 0} \frac{1}{\alpha} \Lambda(\alpha).
\end{align}
A related risk measure is called the coherent risk measure. For any $\beta \geq 0$, the coherent risk measure is denoted by $\rho(\beta)$ and is defined as follows:
\begin{align}\label{eq_app:def_coherentrisk}
	&\rho(\beta) \coloneqq \sup_{\substack{Q\in\matspace(\samplespace), Q\ll P:\\\expect{}{\pi_{Q}}{D_{KL}(Q(X,:)\Vert P(X,:))} \leq \beta}}  \expect{}{\pi_Q}{C(X)},
\end{align}
where $\matspace(\samplespace)$ is the set of aperiodic and recurrent transition matrices on $\samplespace$, and $\pi_Q$ is the unique stationary distribution of a Markov chain with transition matrix $Q\in\matspace(\samplespace)$. In particular, $\rho(\beta)$ takes into account model uncertainties in terms of the transition probability, i.e., $\rho(\beta)$ is the worst case average cost associated with the Markov chain $\{\Phi_i\}_{i\geq0}$, assuming its transition probability is $Q$ instead of $P$ and that $\expect{}{\pi_{Q}}{D_{KL}(Q(X,:)\Vert P(X,:))} \leq \beta$. We will establish the connection between these two risk measures. We begin with proving the Donsker-Varadhan variational formula, also known as,
Gibbs variational formula.
\begin{proposition}[Donsker-Varadhan variational formula]\label[proposition]{prop_app:donskervaradhan}
	{\normalfont (i)} For any distributions $p$ and $q$ on $\samplespace$ with $q\ll p$, we have
	\begin{align}
		D_{KL}(q\Vert p) = \sup_{\phi:\samplespace\to\real} \expect{}{q}{\phi(X)} - \ln\expect{}{p}{\e^{\phi(X)}}
	\end{align}
	where equality is obtained by $\phi^*(x) \coloneqq \ln \left(q(x)/p(x)\right)$ for all $x\in\samplespace$.\newline
	{\normalfont (ii)} For any $\alpha > 0$, any function $\phi:\samplespace\to\real$, and any distribution $p$ on $\samplespace$, we have
	\begin{align}
		\frac{1}{\alpha} \ln\expect{}{p}{\e^{\alpha \phi(X)}} = \sup_{q\ll p} \expect{}{q}{\phi(X)} - \frac{1}{\alpha} D_{KL}(q\Vert p)
	\end{align}
	where equality is obtained by $q^*(x) \coloneqq p(x)\e^{\alpha\phi(x)}\mathbin{/}\sum_{y\in\samplespace}p(y)\e^{\alpha\phi(y)}$ for all $x\in\samplespace$.
\end{proposition}
\begin{proof}
	For any $\phi:\samplespace\to\real$, define a distribution $r_\phi$ on $\samplespace$ as follows:
	\begin{align}
		r_\phi(x) = \frac{p(x) \e^{\phi(x)}}{\sum_y p(y)\e^{\phi(y)}},\qquad\forall x\in\samplespace.
	\end{align}
	We have
	\begin{align}
		D_{KL}(q\Vert p) &= D_{KL}(q\Vert r_\phi) + \expect{}{q}{\ln r_\phi(X)} - \expect{}{q}{\ln p(X)}\\
		&= D_{KL}(q\Vert r_\phi) + \expect{}{q}{\phi(X)} - \ln\expect{}{p}{\e^{\phi(X)}}\\
		&\geq \expect{}{q}{\phi(X)} - \ln\expect{}{p}{\e^{\phi(X)}},
	\end{align}
	which proves both {\normalfont (i)} and {\normalfont (ii)}.
\end{proof}
By the Donsker-Varadhan variational formula, for any fixed $n\in\nplus$, any initial state $x\in\samplespace$, and any transition matrix $Q\in\matspace(\samplespace)$ for which $Q\ll P$, we have
\begin{align}\label{eq_app:proof_varform}
	\frac{1}{\alpha} \ln\expect{\Phi_0=x}{P}{\exp\left(\alpha \sum_{i=1}^{n} C(\Phi_i)\right)}  \geq \expect{\Phi_0=x}{Q}{\sum_{i=1}^{n} C(\Phi_i)} - \frac{1}{\alpha}D_{KL}(Q^{1:n}_x\Vert P^{1:n}_x),
\end{align}
where $Q^{1:n}_x$ denotes the distribution of $(\Phi_{1},\Phi_{1},\cdots,\Phi_{n})$ with transition probability matrix $Q$, given $\Phi_0 = x$. Notice that
\begin{align}
	D_{KL}(Q^{1:n}_x\Vert P^{1:n}_x) = D_{KL}(Q^{1:n-1}_x\Vert P^{1:n-1}_x) + \expect{\Phi_0=x}{Q}{D_{KL}(Q(\Phi_{n-1},:)\Vert P(\Phi_{n-1},:))}.
\end{align}
Since the distribution of $\Phi_{n-1}$ converges to $\pi_Q$, we have
\begin{align}
	\lim_{n\to\infty} \expect{\Phi_0=x}{Q}{D_{KL}(Q(\Phi_{n-1},:)\Vert P(\Phi_{n-1},:))} = \expect{}{\pi_Q}{D_{KL}(Q(X,:)\Vert P(X,:))}.
\end{align}
Dividing both sides of \cref{eq_app:proof_varform} by $n$ and letting $n\to\infty$, yields
\begin{align}\label{eq_app:varineq_lowerbound}
	e(\alpha) \geq \expect{}{\pi_Q}{C(X)} - \frac{1}{\alpha} \expect{}{\pi_Q}{D_{KL}(Q(X,:)\Vert P(X,:))}.
\end{align}
Notice that by \cref{prop_app:donskervaradhan}, equality in \cref{eq_app:proof_varform} cannot be achieved by $Q^{1:n}_x$ for any $Q\in\matspace(\samplespace)$. However, one may hope that in the limit, equality in \cref{eq_app:varineq_lowerbound} can be achieved for some $Q\in\matspace(\samplespace)$. This results in the following variational formula for the entropic risk measure.
\begin{proposition}\label[proposition]{prop_app:varform} We have
	\begin{align}\label{eq_app:varform}
		e(\alpha) = \sup_{Q\in\matspace(X), Q\ll P} \expect{}{\pi_Q}{C(X)} - \frac{1}{\alpha} \expect{}{\pi_Q}{D_{KL}(Q(X,:)\Vert P(X,:))}.
	\end{align}
\end{proposition}
\begin{proof}
	It is easy to verify that all eigenvalues of $ P\, \text{diag}\left(\e^{\alpha C}\right) $ and $ \text{diag}\left(\e^{\alpha C}\right) P$ are the same. Hence, by \cref{prop_app:riskcost} and Collatz-Wielandt formula, we have
	\begin{align}
		\Lambda(\alpha) &= \ln\left(\max_{v\in\real_+^{\left|\samplespace\right|}} \min_{x\in\samplespace} \frac{\left[P\, \text{diag}\left(\e^{\alpha C}\right) v\right]_x}{\left[v\right]_x} \right) \\
		&= \max_{v\in\real_+^{\left|\samplespace\right|}} \min_{x\in\samplespace} \ln\left(\frac{\left[P\, \text{diag}\left(\e^{\alpha C}\right) v\right]_x}{\left[v\right]_x} \right),
	\end{align}
	where the last equality follows by the fact that $\ln(\cdot)$ is a strictly increasing function. Since all element of $P\, \text{diag}\left(\e^{\alpha C}\right)$ are non-negative, it is easy to verify that the maximum in above is achieved by a strictly positive vector $v\in\real_+^{\left|\samplespace\right|}$. Changing the maximization variable to $w=\ln(v)$, we get
	\begin{align}
		&= \max_{w\in\real^{\left|\samplespace\right|}} \min_{x\in\samplespace} \ln\left({\left[P\, \text{diag}\left(\e^{\alpha C+w}\right)\right]_x}\right) - w(x)\\
		&=\max_{w\in\real^{\left|\samplespace\right|}} \min_{x\in\samplespace} \ln\expect{\Phi_0=x}{P}{\e^{\alpha C(\Phi_1)+w(\Phi_1)}} - w(x).
	\end{align}
	where $\left[\e^{\alpha C+w}\right]_x = \e^{\alpha C(x) + w(x)}$. Using Donsker-Varadhan variational formula, we get
	\begin{align}
		&=\max_{w\in\real^{\left|\samplespace\right|}} \min_{x\in\samplespace} \sup_{q\ll P(x,:)}\expect{}{q}{{\alpha C(X)+w(X)}} - D_{KL}(q\Vert P(x,:)) -  w(x).
	\end{align}
	Suppose that for any $x\in\samplespace$, the supremum in the above equality is obtained by $q^*_x$. Notice that $q^*_x$ is equivalent to $P(x,:)$, i.e., $P(x,y) = 0$ iff $q^*_x(y) = 0$ for any $y\in\samplespace$. Also, notice that $q^*_x$ depends on $w$. Define a transition matrix $Q_w\in\matspace(\samplespace)$ using $\{q^*_x\}_{x\in\samplespace}$ as its rows, i.e., $Q_w(x,:) = q^*_x$ for any $x\in\samplespace$. We have
	\begin{align}
			&=\max_{w\in\real^{\left|\samplespace\right|}} \min_{x\in\samplespace} \expect{}{Q_w(x,:)}{{\alpha C(X)+w(X)}} - D_{KL}(Q_w(x,:)\Vert P(x,:)) -  w(x)\\
			&\leq \max_{w\in\real^{\left|\samplespace\right|}} \sum_{x\in\samplespace}\pi_{Q_w}(x) \left(\expect{}{Q_w(x,:)}{{\alpha C(X)+w(X)}} - D_{KL}(Q_w(x,:)\Vert P(x,:)) -  w(x)\right)
	\end{align}
	where $\pi_{Q_w}$ is the stationary distribution of a Markov chain with transition matrix $Q_w\in\matspace(\samplespace)$. Simplifying the above, we have
	\begin{align}
		&= \max_{w\in\real^{\left|\samplespace\right|}} \expect{}{\pi_{Q_w}}{\alpha C(X)} - \expect{}{\pi_{Q_w}}{D_{KL}(Q_w(X,:)\Vert P(X,:))} \\
		&\leq \sup_{Q\in\matspace(X), Q\ll P} \expect{}{\pi_{Q}}{\alpha C(X)} - \expect{}{\pi_{Q}}{D_{KL}(Q(X,:)\Vert P(X,:))}.
	\end{align}
	The other inequality follows by \cref{eq_app:varineq_lowerbound}, which holds for any transition matrix $Q\in\matspace(\samplespace)$.
\end{proof}
The above variational formula gives a relation between coherent and entropic risk measures. This confirms our claim that the risk-sensitive cost is stable against model uncertainties.
\begin{corollary}
	For any $\alpha > 0$, suppose that $Q_\alpha$ maximizes the right-hand side of \cref{eq_app:varform}. We have
	\begin{align}
		e(\alpha) = \rho(\beta_\alpha) - \beta_\alpha/\alpha,
	\end{align}
	where $\beta_\alpha \coloneqq  \expect{}{\pi_{Q_\alpha}}{D_{KL}(Q_\alpha(X,:)\Vert P(X,:))}$.
\end{corollary}
\begin{proof}
	By \cref{prop_app:varform}, for any $Q\in\matspace(\samplespace)$ with $\expect{}{\pi_{Q}}{D_{KL}(Q(X,:)\Vert P(X,:))}\leq \beta_\alpha$, we have
	\begin{align}
		e(\alpha) \geq \expect{}{\pi_Q}{C(X)} - \frac{1}{\alpha} \expect{}{\pi_Q}{D_{KL}(Q(X,:)\Vert P(X,:))} \geq \expect{}{\pi_Q}{C(X)} - \frac{\beta_\alpha}{\alpha}.
	\end{align}
	Notice that equality holds for $Q = Q_\alpha$. The result follows by taking supremum over $Q$ from the right-hand side of the above inequality.
\end{proof}

%	\section{Limitations of Robustness Guarantees in the CVaR Formulation}\label{sec_app:cvar}
%	\input{Sections/Appendix/limitCvar}

	\section{Proof of \cref{thm:main}}\label{proof:thm:main}
	The idea behind the proof is motivated by the argument we presented for the deterministic ODE in \cref{subsec:heuristic_fg}, but with additional work building upon \cref{subsec:heuristic_approxnabla} to deal with the stochastic noise. We first present the sketch of the proof. We focus on \cref{thm:main}{\normalfont (ii)} as the proof of \cref{thm:main}{\normalfont (i)} follows by the exact same argument, replacing $\Lambda_\theta$ with $\Lambda_\theta^{(M)}.$

{\bf Step 0, Stochastic noise is negligible}: Let us rewrite the update equations as follows:
\begin{align}
	&\theta_{m+1} = \theta_m - \gamma_m f_m \left(\theta_m,\tilde{\Lambda}_m\right) - \gamma_m \left(F_m \left(\theta_m,\tilde{\Lambda}_m\right)-f_m \left(\theta_m,\tilde{\Lambda}_m\right)\right),\\
	&\tilde{\Lambda}_{m+1} = \tilde{\Lambda}_m + \eta \gamma_m \left(g_m\left(\theta_m,\tilde{\Lambda}_m\right) - 1\right)  + \eta \gamma_m \left(G_m \left(\theta_m,\tilde{\Lambda}_m\right) - g_m \left(\theta_m,\tilde{\Lambda}_m\right)\right).
\end{align}
Defining $r_m \coloneqq \left[\theta_m,\tilde{\Lambda}_m\right]^T$ and $k_m(r_m) \coloneqq \left[-f_m \left(\theta_m,\tilde{\Lambda}_m\right),\eta \left(g_m\left(\theta_m,\tilde{\Lambda}_m\right) - 1\right)\right]^T$, we may rewrite the above equations as $r_{m+1} = r_m + \gamma_m k_m(r_m) + \epsilon_m$, where $\epsilon_m$ is the stochastic noise. It then follows that $\sum_m \epsilon_m < \infty$ with probability $1$ and that $\{\tilde{\Lambda}_m\}$ is bounded almost surely. Notice that this approximation generalizes to any family of functions $\{\psi_m(\cdot)\}$ for which
\begin{align}
	\sum_{m=0}^\infty \sup_{r\in\real^l \times I} \psi_{m+1}(r) - \psi_m(r) < \infty,
\end{align}
and $\{\psi_m(\cdot)\}$, $\{\nabla \psi_m(\cdot)\}$, and $\{\nabla^2 \psi_m(\cdot)\}$ are uniformly bounded over $(\theta,\Lambda)\in \real^l \times I$ for any interval $I = [a,b]$. In particular, for any such family we have
\begin{align}
	\psi_{m+1}\left(r_{m+1}\right) = \psi_m\left(r_{m}\right) + \gamma_m \nabla_r \psi_m(r_m) \cdot k_m(r_m) + \epsilon_m(\psi_m) + \psi_{m+1}\left(r_{m+1}\right) - \psi_{m}\left(r_{m+1}\right),
\end{align}
where $\sum_m \epsilon_m(\psi) < \infty$.

{\bf Step I, Convergence of $\Lambda_{\theta_m}$:} This is the main step of the proof. To show the convergence of $\Lambda_{\theta_m}$ the first step is to show that $\tilde{\Lambda}_m - \Lambda_{\theta_m} \to 0$. We begin with showing that $\big|\tilde{\Lambda}_m - \Lambda_{\theta_m}\big|$ gets close to $0$ infinitely often.

The argument in \cref{subsec:heuristic_fg} suggests that $\tilde{\Lambda}_m - \Lambda_{\theta_m} \geq \epsilon$ for any fixed $\epsilon > 0$ after some iterations. Notice that the choice of $\epsilon$ cannot be $0$ due to the stochastic noise. In particular, we show that $\liminf_{m\to\infty} \tilde{\Lambda}_m - \Lambda_{\theta_m} \geq 0$. To prove this result, we will study the drift of a family of Lyapunov functions (one for each $M_m$) $\phi_m(\theta,\Lambda) = \Lambda - \Lambda^{(M_m)}_{\theta}$, which is the same function that we used in \cref{subsec:heuristic_fg}.

Next, we may try to show the convergence of $\tilde{\Lambda}_m - \Lambda_{\theta_m}$ using a monotonicity argument similar to \cref{subsec:heuristic_fg}. However, due to the stochastic noise, this approach fails. Instead, we study the convergence of $\Lambda_{\theta_m}$ and $\tilde{\Lambda}_m - \Lambda_{\theta_m}$ together. Notice that when $\big|\tilde{\Lambda}_m - \Lambda_{\theta_m}\big|$ is small enough, then $\tilde{\Lambda}_m $ is a good approximation of $\Lambda_{\theta_m}$ and one may expect that similar to the argument in \cref{subsec:heuristic_approxnabla}, the value of $\Lambda_{\theta_m}$ (and hence $\tilde{\Lambda}_m$) should not increase by much. On the other hand, if $\tilde{\Lambda}_m - \Lambda_{\theta_m}$ is bounded away from zero, the value of $\tilde{\Lambda}_m$ should decrease.

More precisely, using a contradiction based argument, we show that if $\limsup_{m\to\infty} \tilde{\Lambda}_m - \Lambda_{\theta_m} > 0$, then the value of $\tilde{\Lambda}_m$ goes to $-\infty$. This is done by breaking down the iterations of the algorithm into different cycles, where during an `even' cycle the value of $\big|\tilde{\Lambda}_m - \Lambda_{\theta_m}\big|$ remains small, and during an `odd' cycles, the value of $\tilde{\Lambda}_m - \Lambda_{\theta_m}$ is strictly bounded away from $0$. We show that the value of $\tilde{\Lambda}_m$ cannot increase by much during an even cycle, and it decreases by a constant during an odd cycle, implying that $\tilde{\Lambda}_m\to-\infty$.

The natural choice of Lyapunov function family to be applied during an odd cycle is $\phi_m(\theta,\Lambda) = \Lambda - \Lambda^{(M_m)}_{\theta}$. Based on the argument in \cref{subsec:heuristic_approxnabla}, during an odd cycle one may attempt to use $\psi_m (\theta,\Lambda)= \Lambda^{(M_m)}_{\theta}$ as the Lyapunov function family. However, we need to incorporate the fact that $\big|\tilde{\Lambda}_m - \Lambda^{(M_m)}_{\theta_m}\big|$ is small and thus, a the natural choice of Lyapunov function family to be applied during an even cycle is $\psi_m (\theta,\Lambda)= \Lambda^{(M_m)}_{\theta} + a\left(\tilde{\Lambda}_m - \Lambda^{(M_m)}_{\theta_m}\right)^2$ where $a > 0$ is a fixed constant to be determined.

{\bf Step II, Convergence of $\nabla_\theta{\Lambda}_{\theta_m}$:} The argument in here is almost identical to the one in \cref{subsec:heuristic_main}. Since $\tilde{\Lambda}_{m} - {\Lambda}_{\theta_m}\to 0$, it follows that $f_m(r_m) - f_m(\theta_m,\Lambda_{\theta_m}) \to 0$. Notice that $f_m(\theta_m,\Lambda_{\theta_m})\ \propto \ \nabla_\theta{\Lambda}_\theta$.

The outline provided above primarily presents the intuition behind the proof in \cite{Marbach2001}. The additional work in here is to deal with the nature of the stochastic noise in the risk-sensitive problem.

We start with showing some preliminary results and then fill in the details that we skipped in the sketch of the proof. Some of the steps that we take here are similar to the work of \cite{Marbach2001}. For the sake of completeness and self-sufficiency, we present all the details. %The proof of \cref{lem:prelim} is postponed to \cref{proof:lem:prelim}, and the proof of \cref{cor:prelim}

\begin{lemma}\label[lemma]{lem:prelim}
	Suppose that \cref{ass:1} holds. Then, there exits a constant $\overline{R}_{x^*} > 1$ such that for any $R\in \left(1,\overline{R}_{x^*}\right)$,  there exists $M(R),m(R)\in (1,\infty)$ that satisfies
	\begin{align}
		m(R) < \expect{\Phi_0=x}{P}{R^{\tau_{x^*}}} < M(R)
	\end{align}
	for all $P\in\xbar{\probspace}$ and $x\in\samplespace$. Moreover, $M(R)\downarrow 1$ as $R\downarrow 1$.
\end{lemma}
\begin{proof}
		Let us first show the existence of $\overline{R}_{x^*}$ and a uniform upper bound $M(R)$. Fix a constant $\hat{R} > 1$. We claim that $\exists \hat{N}\in\nplus$ such that $P({\tau_{x^*}} > \hat{N} | \Phi_0 = x) \leq {\hat{R}}^{-1}$ for any $x\in\samplespace$ and any $P\in\xbar{\probspace}$. Suppose the contrary, i.e., there exists a sequence $\{P_i\}\subset\xbar{\probspace}$ such that $N_i\coloneqq \inf\{N: P_i(\tau_{x^*} > N| \Phi_0 = x)\leq \hat{R}^{-1},~\forall x\in\samplespace  \}\uparrow\infty$ as $i\to\infty$. Using a compactness argument, there exists a subsequence $\{P_j\}$ such that $P_j \to P \in\xbar{\probspace}$. Contradiction follows by \cref{ass:1} and the fact that $P_j(\tau_{x^*} > N| \Phi_0 = x) \to P(\tau_{x^*} > N| \Phi_0 = x)$.

Given the choice of $\hat{N}$, for any $k\in\nplus$ and all $x\in\samplespace$ we have
\begin{align}
	&P({\tau_{x^*}} > \hat{N} + k| \Phi_0 = x)\\
	&\myquad[1]=\sum_{(x_1,\cdots,x_{k-1},y)\in\left(\samplespace\setminus \{x^*\}\right)^{k}} P(\Phi_1=x_1,\cdots,\Phi_{k-1}=x_{k-1},\Phi_{k}=y| \Phi_0 = x) P({\tau_{x^*}} > \hat{N}| \Phi_0 = y) \\
	&\myquad[1] \leq {\hat{R}}^{-1} P({\tau_{x^*}} > k| \Phi_0 = x).
\end{align}
In particular, $P({\tau_{x^*}} > k \hat{N}| \Phi_0 = x) \leq {\hat{R}}^{-k}$ for any $k\in\nplus$ and all $x\in\samplespace$. Hence, for any $x\in\samplespace$, $P\in\xbar{\probspace}$, and $1< R < \sqrt[{\hat{N}}]{\hat{R}}$ we have
\begin{align}
	\expect{\Phi_0 = x}{P}{R^{\tau_{x^*}}} &\le R\sum_{k=0}^{\infty} R^{k} P(\tau_{x^*} > k) \\
	&\leq R \sum_{k=0}^{\infty} \sum_{l=0}^{\hat{N}-1}  R^{k \hat{N} + l} P(\tau_{x^*} > k \hat{N}),\\
	&\leq \frac{\hat{R}^{\hat{N}+1}}{\hat{R}-1} \sum_{k=0}^{\infty} R^{k\hat{N}} {\hat{R}}^{-k} < \infty,
\end{align}
which proves the existence of $\overline{R}_{x^*}$ and a uniform upper bound $M(R)$ for $1<R<\overline{R}_{x^*}$.

Next, we prove the existence of a uniform lower bound $m(R) > 1$. For the sake of contradiction, suppose that there exists a sequence $\{P_i\} \subset \xbar{\probspace}$ such that $\expect{\Phi_0=x}{P_i}{R^{\tau_{x^*}}} \to 1$. Using a compactness argument, we can pick a subsequence $\{P_j\}$ such that $P_j \to P \in\xbar{\probspace}$. This together with the fact that $\expect{\Phi_0=x^*}{P}{R^{\tau_{x^*}}} < M(R)$ for all $P\in\xbar{\probspace}$ implies that $\expect{\Phi_0=x^*}{P_i}{R^{\tau_{x^*}}} \to \expect{\Phi_0=x^*}{P}{R^{\tau_{x^*}}}$. Contradiction follows by the fact that $\expect{\Phi_0=x^*}{P}{R^{\tau_{x^*}}} > 1$.

Finally, we show that we can choose $M(R)$ such that $M(R)\downarrow 1$ as $R\downarrow 1$. For the sake of contradiction, suppose there exist sequences $\{R_i\}$ and $\{P_i\}$ such that $R_i\downarrow 1$ and for all $i\in\nplus$, $\expect{\Phi_0 = x}{P_i}{{R_i}^{\tau_{x^*}}} > 1+\epsilon$ for some constant $\epsilon > 0$. Following the same argument as before, pick a subsequence $\{P_j\}$ such that $P_j\to P\in\xbar{\probspace}$. We have $\expect{\Phi_0 = x}{P_j}{{R_k}^{\tau_{x^*}}} \to \expect{\Phi_0 = x}{P}{{R_k}^{\tau_{x^*}}}$ for any $k\in\nplus$, as $j\to\infty$. Moreover, by the monotone convergence theorem $\expect{\Phi_0 = x}{P}{{R_k}^{\tau_{x^*}}} \downarrow 1$ as $k\to\infty$. Hence, we can pick $k_0$ large enough such that $\expect{\Phi_0 = x}{P}{{R_{k_0}}^{\tau_{x^*}}} < 1 + \epsilon/4$, and then $j_0$ large enough such that $\left|\expect{\Phi_0 = x}{P_j}{{R_{k_0}}^{\tau_{x^*}}} - \expect{\Phi_0 = x}{P}{{R_{k_0}}^{\tau_{x^*}}}\right| < \epsilon/4$ for all $j \geq j_0$. Hence, $\expect{\Phi_0 = x}{P_{j_0\vee k_0}}{{R_{k_0}}^{\tau_{x^*}}} < 1+\epsilon/2$. Contradiction follows by the fact that $R_{k_0} \geq R_{j_0\vee k_0}$.
\end{proof}
\begin{corollary}\label[corollary]{cor:prelim}
	Let \cref{ass:1,ass:2} hold. Then, \cref{lem:prelim} holds for $\xbarcheck{\probspace}$.
\end{corollary}
\begin{proof}
	The proof follows by \cref{cor:Phat_commonrec}.
\end{proof}

\subsection*{Step 0, Stochastic noise is negligible}
Let us rewrite \cref{eq:algorithm} as $r_{m+1} = r_m + \gamma_m k_m(r_m) + \epsilon_m$ where
$r_m \coloneqq
\begin{bmatrix}
	\theta_m\\
	\tilde{\Lambda}_m
\end{bmatrix}\in\real^{l+1}$,
\begin{align}
	&k_m(r_m) =
	\begin{bmatrix}
		-f_m \left(\theta_m,\tilde{\Lambda}_m\right)\\
		\eta \left(g_m\left(\theta_m,\tilde{\Lambda}_m\right) - 1\right)
	\end{bmatrix}
	,\allowdisplaybreaks\\
	&\epsilon_m =
	\begin{bmatrix}
		\epsilon_m(\theta)\\
		\epsilon_m(\Lambda)
	\end{bmatrix} = \gamma_m
	\begin{bmatrix}
		-F_m \left(\theta_m,\tilde{\Lambda}_m\right)+f_m \left(\theta_m,\tilde{\Lambda}_m\right)\\
		\eta \left(G_m \left(\theta_m,\tilde{\Lambda}_m\right) - g_m \left(\theta_m,\tilde{\Lambda}_m\right)\right)
	\end{bmatrix},
\end{align}
$g_m \left(\theta_m,\tilde{\Lambda}_m\right) = \expect{\Phi_{t_m}=x^*}{\theta_m}{G_m \left(\theta_m,\tilde{\Lambda}_m\right)}$, and $f_m \left(\theta_m,\tilde{\Lambda}_m\right) = \expect{\Phi_{t_m}=x^*}{\theta_m}{F_m \left(\theta_m,\tilde{\Lambda}_m\right)}$.
Notice that for any $m$, $\norm{k_m(r_m)} <\infty$ and $\expect{}{\theta_m}{\epsilon_m} = 0$; however, the bound over $\norm{k_m(r_m)}$ is not uniform since $\Lambda_{\theta_m} - \tilde{\Lambda}_m$ might be large.
\begin{lemma}\label[lemma]{lem:lambtild_lowerbound}
	The sequence $\{\tilde{\Lambda}_m\}$ is bounded from below by $\underline{\Lambda} \coloneqq \min(\tilde{\Lambda}_{0},\alpha \underline{C} -\eta \gamma_{\max})$ where $\underline{C}$ is a uniform lower bound on the cost function $C_\theta(\cdot)$.
\end{lemma}
\begin{proof}
	Recall that the update equation for $\tilde{\Lambda}_m$ is given as follows:
	\begin{align}
		\tilde{\Lambda}_{m+1} = \tilde{\Lambda}_m + \eta \gamma_m \left(G_m\left(\theta_m,\tilde{\Lambda}_m\right) - 1\right).
	\end{align}
	If $\tilde{\Lambda}_{m} \leq \alpha \underline{C}$ for some $m$, then $\alpha C_{\theta}(x) - \tilde{\Lambda}_{m} \geq 0$ for all $x\in\samplespace$ and $\theta\in\real^l$, which implies that $G_m\left(\theta_m,\tilde{\Lambda}_m\right) \geq 1$. Hence, $\tilde{\Lambda}_{m+1} \geq \tilde{\Lambda}_{m}$. The result follows by the fact that $\tilde{\Lambda}_{m+1} - \tilde{\Lambda}_{m} \geq -\eta \gamma_m$ for all $m$.
\end{proof}
\begin{corollary}\label[corollary]{cor:lambtild_lowerbound}
	$\{\|k_m(\theta_m,\tilde{\Lambda}_m)\|\}$ is uniformly bounded by $\textrm{Const}\times \gamma_m^{-\beta}$ for some $\textrm{Const} > 0$ independent of $m$.
\end{corollary}
\begin{proof}
	By \cref{lem:lambtild_lowerbound} and \cref{ass:2}, we have
	\begin{align}
		0 < G_m\left(\theta_m,\tilde{\Lambda}_m\right) &< \gamma_m^{-\beta}\left(1+\left[\sum_{i=0}^4 \frac{1}{i!}\left(\ln\left(\gamma_m^{\beta}H_m(\theta,\tilde{\Lambda}_m)\right)\right)^i\right]\mathbbm{1}\left\{H_m(\theta,\tilde{\Lambda}_m) > \gamma_m^{-\beta}\right\}\right)\\
		&< \gamma_m^{-\beta}\left(1+\left[\sum_{i=0}^4 \frac{1}{i!}\left(\left(\alpha \overline{C} - \underline{\Lambda}\right)\tau_{x^*}\right)^i\right]\right)
	\end{align}
	where $\overline{C}$ is a uniform upper bound on the cost function $C_\theta(\cdot)$. By \cref{lem:prelim}, $\expect{}{\theta}{(\tau_{x^*})^i}$ is uniformly bounded over $\theta\in\real^l$ for any fixed $i$. Hence, $\expect{\Phi_{t_m}=x^*}{\theta_m}{G_m\left(\theta_m,\tilde{\Lambda}_m\right)} < \textrm{Const}\times \gamma_m^{-\beta}$. Similar argument together with \cref{ass:3}, implies that $\norm{\expect{\Phi_{t_m}=x^*}{\theta_m}{F_m\left(\theta_m,\tilde{\Lambda}_m\right)}}\allowbreak < \textrm{Const}\times \gamma_m^{-\beta}$.
\end{proof}
Next, using a standard martingale argument, we show that the stochastic noise is negligible.
\begin{lemma}\label[lemma]{lem:sum_epsm}
	$\sum_m \epsilon_m$ converges with probability $1$.
\end{lemma}
\begin{proof}
	Let us write $\epsilon_m = \gamma_m(K_m(r_m)-k_m(r_m))$ where $K_m(r_m)$ is the stochastic version of $k_m(r_m)$. Notice that $\expect{}{\theta_m}{\epsilon_m|\mathcal{F}_m} = 0$ where $\mathcal{F}_m = \{\theta_0,\tilde{\Lambda}_0,\Phi_0,\Phi_1,\cdots,\Phi_{t_m}\}$ is the history of the Markov chain up to time-step $t_m$. Also, notice that by a similar argument as in \cref{cor:lambtild_lowerbound}, we have $\expect{}{\theta_m}{\norm{K_m(r_m)-k_m(r_m)}_2^2 | \mathcal{F}_m} \leq \textrm{Const}\times \gamma_m^{-2\beta}$, for some fixed $\textrm{Const} > 0$ independent of $m$.

	Let $w_n \coloneqq \sum_{m=1}^{n} \epsilon_m$. Since $\{\epsilon_m\}$ is a martingale difference sequence, $\{w_n\}_{n\geq0}$ is a martingale. Notice that
	\begin{align}
		\mathbb{E}\left[ \norm{w_n}^2 \right] = \mathbb{E}\left[\sum_{m=1}^{n} \norm{\epsilon_m}^2 \right] \leq \textrm{Const}\times \sum_{m=1}^{\infty} \gamma_m^{2(1-\beta)}.
	\end{align}
	which is bounded. Using the martingale convergence theorem, $w_n\to w_\infty$ with probability $1$, where $\mathbb{E}[w_\infty] < \infty$.
\end{proof}
\begin{corollary}\label[corollary]{cor:lambtild_upperbound}
	The sequence $\{\tilde{\Lambda}_m\}$ is bounded from above almost surely.
\end{corollary}
\begin{proof}
	By \cref{lem:sum_epsm}, $\sum\epsilon_m(\Lambda)$ converges; hence, there exists $m_0 > 0$ such that $\left|\epsilon_m(\Lambda)\right| < 1$ for any $m>m_0$. Also, notice that by \cref{cor:lambtild_lowerbound}, $\gamma_m \left(g_m\left(\theta_m,\tilde{\Lambda}_m\right) - 1\right) \to 0$; hence, there exists $m_1 > m_0$ such that $\gamma_m \left(g_m\left(\theta_m,\tilde{\Lambda}_m\right) - 1\right) < \eta^{-1}$ for all $m>m_1$.

	Notice that if $\tilde{\Lambda}_{m} \geq \alpha \overline{C}$ then $\alpha C_\theta(x) - \tilde{\Lambda}_{m} \leq 0$ for all $x\in\samplespace$ and $G_m\left(\theta_m,\tilde{\Lambda}_m\right) \leq 1$, and hence $\tilde{\Lambda}_{m+1} \leq  \tilde{\Lambda}_{m}	$. In particular, if $\tilde{\Lambda}_{m} \geq 2\alpha \overline{C}$, then $\tilde{\Lambda}_{m} - \tilde{\Lambda}_{m+1} > \gamma_m \times \textrm{Const}$ for some fixed and deterministic constant $\textrm{Const} > 0$ independent of $m$. Since $\sum_i \gamma_i\to\infty$, there exists $m_2 \geq m_1$ such that $\tilde{\Lambda}_{m_2} < 2\alpha \overline{C}$. Notice that,
	\begin{align}
		\tilde{\Lambda}_{m+1} - \tilde{\Lambda}_m &= \eta \gamma_m \left(g_m\left(\theta_m,\tilde{\Lambda}_m\right) - 1\right) + \epsilon_m(\Lambda)
	\end{align}
	which is bounded by $2$ for all $m\geq m_2$. Hence, $\tilde{\Lambda}_{m} < 2\alpha \overline{C}+2$ for any $m>m_2$ and the result follows.
\end{proof}
\begin{remark}
	In case we use the update rule given by \cref{eq:algorithm_project}, the value of $\{\tilde{\Lambda}_m\}$ stays bounded. Since the error is negligible and $\gamma_m k_m(r_m) \to 0$, we can drop the projection after finitely many (random) iterations and the rest of the analysis is similar to the case when there is no projection.
\end{remark}
Next, we generalize the above argument to functions of $r_m$. Let $\psi_m(\cdot):\real^{l+1}\to \real$ be a family of functions indexed by $m\in\nplus$. We want to study the changes in $\psi_m(\cdot)$ based on the update rule of $r_m$. We call $\{\psi_m(\cdot)\}$ a ``family of Lyapunov functions'' if
$\{\psi_m(\cdot)\}$, $\{\nabla \psi_m(\cdot)\}$, and $\{\nabla^2 \psi_m(\cdot)\}$ are uniformly bounded over $(\theta,\Lambda)\in \real^l \times I$ for any interval $I = [a,b]$, and moreover,
\begin{align}
	\sum_{m=0}^\infty \sup_{r\in\real^l \times I} \psi_{m+1}(r) - \psi_m(r) < \infty
\end{align}

Given a family of Lyapunov functions $\{\psi_m(\cdot)\}$, let us rewrite $\psi_{m+1}(r_{m+1})$ as follows
\begin{align}
	\psi_{m+1}(r_{m+1}) &= \psi_{m+1}(r_{m+1}) - \psi_{m}(r_{m+1}) + \psi_{m}(r_{m+1})\\
	&=  \psi_{m+1}(r_{m+1}) - \psi_{m}(r_{m+1}) + \psi_m(r_m + \gamma_m k_m(r_m) + \epsilon_m)\\
	&= \psi_m(r_m) + \gamma_m \nabla\psi_m(r_m)\cdot k_m(r_m) + \Delta_m(\psi_m) + (\psi_{m+1}(r_{m+1}) - \psi_{m}(r_{m+1}))\label{eq:Lyapunovupdate},
\end{align}
where $\Delta_m(\psi_m) \coloneqq \psi_m(r_{m+1}) - \psi_m(r_m) - \gamma_m \nabla\psi_m(r_m)\cdot k_m(r_m)$. Notice that $\Delta_m(\psi_m)$ is related to the second term in Taylor expansion of $\psi_m(r_{m+1})$. We will show that for any family of Lyapunov functions $\{\psi_m(\cdot)\}$, the sequence $\{\Delta_m(\psi_m)\}$ is summable.
\begin{lemma}\label[lemma]{lem:epsneg_gen}
	For any family of Lyapunov functions $\{\psi_m(\cdot)\}$, $\sum_m \Delta_m(\psi_m)$ converges with probability $1$.
\end{lemma}
\begin{proof}
	Consider a sample path of $\{r_m\}$. By \cref{lem:lambtild_lowerbound} and \cref{cor:lambtild_upperbound}, with probability $1$ there exists $I\subset\real$ such that $\{r_m\} \subset \real^l \times I$. By Taylor expansion of $\psi_m(r_{m+1})$, we have
	\begin{align}
		\nabla\psi_m(r_m)\cdot \epsilon_m - D_m(I) \norm{r_{m+1} - r_m}^2 \leq \Delta_m(\psi_m) \leq \nabla\psi_m(r_m)\cdot \epsilon_m + D_m(I) \norm{r_{m+1} - r_m}^2
	\end{align}
	where $D_m(I)$ is the upper bound of $\{\nabla^2\psi_m(\cdot)\}$ over $\real^l\times I$. Notice that
	\begin{align}
		\norm{r_{m+1} - r_m}^2 \leq 2\gamma_m^2 \norm{k_m(r_m)}^2 + 2\norm{\epsilon_m}^2< \textrm{Const}\times \gamma_m^{2-2\beta}+2\norm{\epsilon_m}^2,
	\end{align}
	which is summable. Next, we show that $\{\nabla\psi_m(r_m)\cdot \epsilon_m\}$ is summable with probability $1$.

	For any fixed $\overline{c} \in\real_{>0}$, let us define $w_n(\overline{c}) \coloneqq \sum_{m=1}^{n\wedge D(\overline{c})} \nabla\psi_m(r_m)\cdot \epsilon_m$, where
	\begin{align}
		D(\overline{c}) \coloneq \inf\left\{m:\expect{}{}{\norm{\nabla\psi_m(r_m)\cdot \epsilon_m}^2|\mathcal{F}_m}> \overline{c}\gamma_m^{2(1-\beta)} \right\} - 1
	\end{align}
	is a constant that depends on the sequence $\{r_m\}$. Notice that $\{w_n(\overline{c})\}$ is a martingale difference sequence with bounded second moments; hence, it converges. Since $\{r_m\} \subset \real^l \times I$ for some $I\in\real$ with probability $1$, and $\expect{}{}{\norm{\epsilon_m}^2|\mathcal{F}_m}\leq \text{Const}\times \gamma_m^{2(1-\beta)}$, for almost every sample path there exists a large enough $\overline{c}\in\real_{>0}$ such that
	$D(\overline{c}) = \infty$. Hence, $\nabla\psi_m(r_m)\cdot \epsilon_m$ is summable and the result follows.
\end{proof}

For the rest of the proof, we consider a fixed sample path for which the stochastic noise is summable, including $\epsilon_m$ as well as $\Delta_m(\psi_m)$ for any family of Lyapunov functions $\{\psi_m\}$ that will be considered below. Notice that with probability $1$ all these errors are summable, hence, almost every sample path is considered.

Our analysis in the rest of the proof is based on the Taylor expansion of function $g_m(\theta,\cdot)$ and $f_m(\theta,\cdot)$ around the corresponding fixed point $\Lambda_\theta^{(M_m)}$. To ensure the error term of the Taylor expansion is bounded, we need the following result. This is the reason why we needed that many terms in the definition of $g^{M}(\theta,\Lambda)$.
\begin{corollary}\label[corollary]{cor:deltagmfm}
	There exists $\delta_1 < \delta_0$ for which $\frac{\dertwo{f^{(M)}}(\theta,\Lambda)}{\der{\Lambda^2}}$ and $\frac{\derthree{f^{(M)}}(\theta,\Lambda)}{\der{\Lambda^3}}$ are uniformly bounded over $\theta\in\real^l$, $\Lambda_\theta - \Lambda < \delta_1$, and $M>1$.
\end{corollary}
\begin{proof}
	The proof follows by \cref{lem:gradgfgmfm} and simple algebra. In particular, one can calculate $\frac{\dertwo{f^{(M)}}(\theta,\Lambda)}{\der{\Lambda^2}}$ and $\frac{\derthree{f^{(M)}}(\theta,\Lambda)}{\der{\Lambda^3}}$ interchanging expectation and derivative (governed by \cref{ass:2,ass:3}, and \cref{lem:prelim}) and follow a similar argument as in the proof of \cref{lem:gradgfgmfm}.
\end{proof}

Notice that by \cref{lem:gradgfgmfm} and \cref{cor:deltagmfm}, $\frac{\der{g_m}(r)}{\der{\Lambda}}$, $\frac{\dertwo{g_m}(r)}{\der{\Lambda^2}}$, and $\frac{\der{f_m}(r)}{\der{\Lambda}}$, $\frac{\dertwo{f_m}(r)}{\der{\Lambda^2}}$, $\frac{\derthree{f_m}(r)}{\der{\Lambda^3}}$ are all uniformly bounded for any $r=(\theta,\Lambda)$ such that $\Lambda_\theta - \Lambda < \delta_0$. Hence, for $\Lambda_\theta - \Lambda < \delta_0$ we can write
\begin{align}
	\begin{aligned}
		g_m(\theta,\Lambda) &= 1+\frac{\der{g_m}(\theta,\Lambda_\theta^{(M_m)})}{\der{\Lambda}}\left(\Lambda - \Lambda^{(M_m)}_\theta\right) + O\left(\left(\Lambda - \Lambda^{(M_m)}_\theta\right)^2\right)
		\\
		f_m(\theta,\Lambda) &= -\frac{\der{g_m(\theta,\Lambda_\theta^{(M_m)})}}{\der{\Lambda}} \nabla_\theta \Lambda_\theta^{(M_m)} +  \frac{\der{f_m}(\theta,\Lambda_\theta^{(M_m)})}{\der{\Lambda}}\left(\Lambda - \Lambda^{(M_m)}_\theta\right)\\
		&\qquad +  \frac{\dertwo{f_m}(\theta,\Lambda_\theta^{(M_m)})}{\der{\Lambda^2}}\left(\Lambda - \Lambda^{(M_m)}_\theta\right)^2 + O\left(\left(\Lambda - \Lambda^{(M_m)}_\theta\right)^3\right) \\
	\end{aligned}\label{eq:taylorexp}
\end{align}
Finally, recall that by \cref{cor:unfbound_lambdagrad}, $\nabla_\theta \Lambda^{(M_m)}_\theta$ and $\nabla^2_\theta \Lambda^{(M_m)}_\theta$ are uniformly bounded.

\subsection*{Step I, Convergence of $\Lambda_{\theta_m}$:} Let us first introduce the two family of Lyapunov functions that will be used to establish the convergence of $\Lambda_{\theta_m}$. Each family of Lyapunov functions is applied to a different region characterized by $(\tilde{\Lambda}_m,\Lambda^{(M_m)}_{\theta_m})$. The first family of Lyapunov functions is $\phi_m(\theta,\Lambda) = \Lambda-\Lambda^{(M_m)}_\theta$ which gives positive `drift' in the region where $0 < \Lambda^{(M_m)}_\theta - \Lambda < \zeta_1$ for sufficiently small $\zeta_1 > 0$. The second family of Lyapunov functions is $\psi_m(\theta,\Lambda) = \Lambda^{(M_m)}_\theta + (6\zeta_2)^{-1}(\Lambda - \Lambda^{(M_m)}_\theta)^2$ which gives negative drift in the region where $|\Lambda - \Lambda_\theta| \leq \zeta_2$ for sufficiently small $\zeta_2 > 0$ and all large enough $m$.

As the first step, we show that the value of $\Lambda_{\theta_m} - \tilde{\Lambda}_m$ gets close to $0$ infinitely many times.
\begin{lemma}\label[lemma]{lem:liminf_abs}
	$\liminf_{m\to\infty} \left|\Lambda_{\theta_m} - \tilde{\Lambda}_m\right| = 0$
\end{lemma}
\begin{proof}
	For the sake of contradiction, suppose that $\liminf \left|\Lambda_{\theta_m} - \tilde{\Lambda}_m\right| > 0$. Hence, there exists $\varepsilon> 0$ such that $\left|\Lambda_{\theta_m} - \tilde{\Lambda}_m\right| > \varepsilon$ for all large $m$. Since $\left|\theta_{m+1} - \theta_m\right|\to0$ and $\nabla_\theta\Lambda_{\theta}$ is uniformly bounded, we have $\left|\Lambda_{\theta_{m+1}} - \Lambda_{\theta_m}\right| \to 0$. We also have $\left|\tilde{\Lambda}_{m+1} - \tilde{\Lambda}_m\right| \to 0$. Hence, either $\Lambda_{\theta_m} - \tilde{\Lambda}_m > \varepsilon$ for all large enough $m$, or $\Lambda_{\theta_m} - \tilde{\Lambda}_m < -\varepsilon$ for all large enough $m$. Without loss of generality suppose that $\Lambda_{\theta_m} - \tilde{\Lambda}_m > \varepsilon$ for all large enough $m$. By \cref{lem:unifapprox_truncost}, for all large $m$ we have $\Lambda_{\theta_m} - \Lambda^{(M_m)}_{\theta_m} < \varepsilon/2$. Hence, for all large enough $m$, we have $\Lambda^{(M_m)}_{\theta_m} - \tilde{\Lambda}_m > \varepsilon/2$.

	Notice that there exists $\Lambda \in [\tilde{\Lambda}_m,\Lambda^{(M_m)}_{\theta_m}]$ such that
	\begin{align}
		0 < g_m(\theta_m,\tilde{\Lambda}_m) - g_m(\theta_m,\Lambda^{(M_m)}_{\theta_m}) = \left(\Lambda^{(M_m)}_{\theta_m} - \tilde{\Lambda}_m\right) \times\left|\frac{\der{g_m(\theta_m,\Lambda)}}{\der{\Lambda}}\right|.
	\end{align}
	By \cref{lem:uniflowerbd_dg}, $\left|\frac{\der{g^{(M)}(\theta,\Lambda)}}{\der{\Lambda}}\right|$ is uniformly bounded away from $0$ over $\theta\in\real^l$, $M\in(1,\infty)$ and $\Lambda\in [\underline{\Lambda},2\alpha\overline{C}+2]$. Hence, $g_m(\theta_m,\tilde{\Lambda}_m)$ is uniformly bounded away from $1$ for all large values of $m$, and we have
	\begin{align}
		\liminf_{m\to\infty} g_m(\theta_m,\tilde{\Lambda}_m) - 1 > 0.
	\end{align}

	Contradiction follows by the fact that given the above inequality, the right-hand side of the following equality goes to $+\infty$ as $k\to\infty$ for all large values of $m$:
	\begin{align}
		\tilde{\Lambda}_{m+k} = \tilde{\Lambda}_m + \eta \sum_{i=m}^{m+k-1} \gamma_i \left(g_i\left(\theta_i,\tilde{\Lambda}_i\right) - 1\right) + \sum_{i=m}^{m+k-1} \epsilon_i(\Lambda).
	\end{align}
\end{proof}
Next, we show that $\liminf_{m\to\infty} \tilde{\Lambda}_m - \Lambda_{\theta_m} \geq 0$, using $\{\phi_m(\cdot)\}$ as the family of Lyapunov functions. We first show that $\phi_m(\theta,\Lambda)$ gives positive drift as long as $0 < \Lambda^{(M_m)}_\theta - \Lambda  < \zeta_1$ for sufficiently small $\zeta_1 > 0$. We then show that this drift will keep $\phi_m(r_m)$ above zero most often, \ie, $\liminf_{m\to\infty} \phi_m(r_m) \geq 0$.

\begin{lemma}\label[lemma]{lem:posdrift_phi}
	$\phi_m(\theta,\Lambda) =\Lambda - \Lambda^{(M_m)}_\theta$ is a family of Lyapunov functions. Moreover, there exists small enough $\zeta_1>0$ independent of $m$ such that for all  $0 < \Lambda^{(M_m)}_\theta - \Lambda  < \zeta_1$ we have
	\begin{align}
		\nabla \phi_m(\theta,\Lambda) \cdot k_m(\theta,\Lambda) > 0.
	\end{align}
\end{lemma}
\begin{proof}
	By \cref{lem:gradgfgmfm} and the assumption on$\{M_m\}_{m\geq0}$, $\{\phi_m(\cdot)\}$ is a family of Lyapunov functions. Using the Taylor expansion given by \cref{eq:taylorexp}, we have
	\begin{align}
		&\nabla \phi_m(\theta,\Lambda) \cdot k_m(\theta,\Lambda)\allowdisplaybreaks\\
		&\myquad[3]= -\frac{\der{g_m(\theta,\Lambda_\theta^{(M_m)})}}{\der{\Lambda}} \norm{\nabla_\theta \Lambda_\theta^{(M_m)}}^2 \\
		&\myquad[5]+ \nabla_\theta \Lambda_\theta^{(M_m)} \cdot \frac{\der{f_m}(\theta,\Lambda_\theta^{(M_m)})}{\der{\Lambda}}\left(\Lambda - \Lambda^{(M_m)}_\theta\right) + O\left(\left(\Lambda - \Lambda^{(M_m)}_\theta\right)^2\right)\\
		&\myquad[5]+\eta\frac{\der{g_m}(\theta,\Lambda_\theta^{(M_m)})}{\der{\Lambda}}\left(\Lambda - \Lambda^{(M_m)}_\theta\right) + O\left(\left(\Lambda - \Lambda^{(M_m)}_\theta\right)^2\right)\allowdisplaybreaks\\
		&\myquad[3]\geq -\frac{\der{g_m(\theta,\Lambda_\theta^{(M_m)})}}{\der{\Lambda}} \norm{\nabla_\theta \Lambda_\theta^{(M_m)}}^2 \allowdisplaybreaks\\
		&\myquad[5]- \left(\norm{\nabla_\theta \Lambda_\theta^{(M_m)}}^2\left(\Lambda^{(M_m)}_\theta - \Lambda \right)^{0.5} + \norm{\frac{\der{f_m}(\theta,\Lambda_\theta^{(M_m)})}{\der{\Lambda}}}^2\left(\Lambda^{(M_m)}_\theta - \Lambda \right)^{1.5}\right) \\
		&\myquad[5]-\eta\frac{\der{g_m}(\theta,\Lambda_\theta^{(M_m)})}{\der{\Lambda}}\left(\Lambda^{(M_m)}_\theta - \Lambda \right) + O\left(\left(\Lambda - \Lambda^{(M_m)}_\theta\right)^2\right)\allowdisplaybreaks\\
		&\myquad[3]= \norm{\nabla_\theta \Lambda_\theta^{(M_m)}}^2 \left(-\left(\Lambda^{(M_m)}_\theta - \Lambda \right)^{0.5}-\frac{\der{g_m(\theta,\Lambda_\theta^{(M_m)})}}{\der{\Lambda}} \right) + \left(\Lambda^{(M_m)}_\theta - \Lambda \right)\times \\
		&\myquad[5]\left(-\eta\frac{\der{g_m}(\theta,\Lambda_\theta^{(M_m)})}{\der{\Lambda}} - \norm{\frac{\der{f_m}(\theta,\Lambda_\theta^{(M_m)})}{\der{\Lambda}}}^2\left(\Lambda^{(M_m)}_\theta - \Lambda \right)^{0.5} + O\left(\Lambda - \Lambda^{(M_m)}_\theta\right)\right).
	\end{align}
	Notice that by \cref{lem:uniflowerbd_dg}, $- \frac{\der{g_m}(\theta,\Lambda_\theta^{(M_m)})}{\der{\Lambda}}$ is positive and is bounded away from zero, uniformly over $\theta\in\real^l$ and $m\in\nplus$. Hence, by \cref{lem:gradgfgmfm}, we can pick $\zeta_1$ small enough so that $\nabla \phi_m(\theta,\Lambda) \cdot k_m(\theta,\Lambda) > 0$ given $0 < \Lambda^{(M_m)}_\theta - \Lambda  < \zeta_1$.
\end{proof}
\begin{corollary}\label[corollary]{cor:posdrift_phi}
	Suppose that for all $i\in\{m,m+1,\cdots,n\}$, we have $0< \Lambda^{(M_i)}_{\theta_i} - \tilde{\Lambda}_i < \zeta_1$. Then, $\phi_{n}(r_{n}) \leq \phi_{m}(r_{m}) + \sum_{i=m}^{n-1}\Delta_{i}(\phi_i) + \sum_{i=m}^{n-1}\Lambda_{\theta_{i+1}}^{(M_{i+1})} - \Lambda_{\theta_{i+1}}^{(M_{i})}$.
\end{corollary}
\begin{proof}
	The proof follows by repetitive use of \cref{eq:Lyapunovupdate} and \cref{lem:posdrift_phi}.
\end{proof}
\begin{lemma}\label[lemma]{lem:liminf_phi}
	$\liminf_{m\to\infty} \tilde{\Lambda}_m - \Lambda^{(M_m)}_{\theta_m} \geq 0$.
\end{lemma}
\begin{proof}
	For the sake of contradiction, suppose that $\liminf_{m\to\infty} \tilde{\Lambda}_m - \Lambda^{(M_m)}_{\theta_m} = -\varepsilon < 0$. By \cref{lem:liminf_abs} and
	the contradiction assumption, the value of $\tilde{\Lambda}_m - \Lambda^{(M_m)}_{\theta_m}$ bounces between values close to $0$ and $-\varepsilon$ infinitely many times.
	Hence, there are infinitely many disjoint subsequences $\{m,m+1,\cdots,n\}$ such that
	\begin{align}
		& -(\zeta_1\wedge\varepsilon) < \tilde{\Lambda}_{m} - \Lambda^{(M_{m})}_{\theta_{m}}  < - 2(\zeta_1\wedge\varepsilon)/3,\\
		&- (\zeta_1\wedge\varepsilon)/3 < \tilde{\Lambda}_{n} - \Lambda^{(M_{n})}_{\theta_{n}}  < 0, \text{ and }   \\
		& - 2(\zeta_1\wedge\varepsilon)/3 < \tilde{\Lambda}_{i} - \Lambda^{(M_{i})}_{\theta_{i}}  < - (\zeta_1\wedge\varepsilon)/3,~\forall i\in\{m+1,\cdots, n-1\}.
	\end{align}
	Notice that for any such subsequence, we have $\phi_n(r_n) - \phi_m(r_m) > (\zeta_1\wedge\varepsilon)/3 $. On the other hand, by \cref{cor:posdrift_phi}, for any such subsequence we have
	\begin{align}
		\phi_{n}(r_{n}) \leq \phi_{m}(r_{m}) + \sum_{i=m}^{n-1}\Delta_{i}(\phi_i) + \sum_{i=m}^{n-1}\Lambda_{\theta_{i+1}}^{(M_{i+1})} - \Lambda_{\theta_{i+1}}^{(M_{i})}.
	\end{align}
	Notice that the sequence $\{\Delta_{i}(\phi_i)\}$ is summable, hence, for all $m$ large enough we have $\sum_{i=m}^{n-1}\Delta_{i}(\phi_i) < (\zeta_1\wedge\varepsilon)/6$. Similarly, by the choice of $\{M_i\}$, the sequence $\Lambda_{\theta_{i+1}}^{(M_{i+1})} - \Lambda_{\theta_{i+1}}^{(M_{i})}$ is summable and for all large enough $m$ we have $ \sum_{i=m}^{n-1}\Lambda_{\theta_{i+1}}^{(M_{i+1})} - \Lambda_{\theta_{i+1}}^{(M_{i})} < (\zeta_1\wedge\varepsilon)/6$. Contradiction follows.
\end{proof}
\begin{corollary}\label[corollary]{cor:liminf}
	$\liminf_{m\to\infty} \tilde{\Lambda}_m - \Lambda_{\theta_m} \geq 0$.
\end{corollary}
\begin{proof}
	Follows by \cref{cor:unifapprox_truncost} and \cref{lem:liminf_phi}.
\end{proof}

Next, we show that $\limsup_{m\to\infty} \tilde{\Lambda}_m - \Lambda_{\theta_m} \leq 0$, using both families $\{\phi_m(\theta,\Lambda)\}$ and $\{\psi_m(\theta,\Lambda)\}$. We first show that $\{\psi_m(\theta,\Lambda)\}$ has negative drift as long as $|\Lambda - \Lambda_\theta|$ is small enough. Hence, in a time interval in which the value of  $|\widetilde{\Lambda}_m - \Lambda_{\theta_m}|$ remains small, the value of $\psi_m(r_m)$ cannot increase by much. This implies that the value of $\widetilde{\Lambda}_m$ and $\Lambda_{\theta_m}$ cannot increase by much either. Using these observations, we then show that if $\limsup_{m\to\infty} \tilde{\Lambda}_m - \Lambda_{\theta_m} > 0$ then the value of $\tilde{\Lambda}_m$ diverges to $-\infty$. This is done by breaking down the time-steps to different cycles, where at `even' cycles the value of $|\widetilde{\Lambda}_m - \Lambda_{\theta_m}|$ remains small, but at `odd' cycles $\widetilde{\Lambda}_m - \Lambda_{\theta_m}$ is bounded away from zero and is positive. Applying $\{\psi_m(\cdot)\}$ as the family of Lyapunov functions to `even' cycles, we know that during any such cycle, the value of $\widetilde{\Lambda}_m$ cannot increase by much. On the contrary, applying $\phi(\cdot)$ as the Lyapunov function to `odd' cycles shows that during any such cycle, the value of $\widetilde{\Lambda}_m$ decreases by a constant. Hence, $\widetilde{\Lambda}_m$ cannot remain bounded and contradiction follows.
\begin{lemma}\label[lemma]{lem:negdrift_psi}
	For any fixed $\zeta_2 > 0$, $\psi_m(\theta,\Lambda) = \Lambda^{(M_m)}_\theta + (6\zeta_2)^{-1}\big(\Lambda - \Lambda^{(M_m)}_\theta\big)^2$ is a family of Lyapunov functions. Moreover, there exists a small enough $\zeta_2 > 0$ independent of $m$ such that for all $\big|\Lambda - \Lambda^{(M_m)}_\theta\big| < \zeta_2$,
	we have
	\begin{align}
		\nabla \psi_m(\theta,\Lambda) \cdot k_m(\theta,\Lambda) \leq 0.
	\end{align}
\end{lemma}
\begin{proof}
	Notice that
	\begin{align}
		\left|\psi_{m+1}(r) - \psi_{m}(r)\right| \leq \left(\Lambda_\theta^{(M_{m+1})} - \Lambda_\theta^{(M_{m})}\right)\left(1 + \frac{1}{6\zeta_2}\left(\left|\Lambda - \Lambda_\theta^{(M_{m+1})}\right| + \left|\Lambda - \Lambda_\theta^{(M_{m})}\right|\right)\right).
	\end{align}
	Therefore, by \cref{ass:2}, \cref{lem:gradgfgmfm} and the assumption on$\{M_m\}_{m\geq0}$, $\{\psi_m(\cdot)\}$ is a family of Lyapunov functions.
	Using the Taylor expansion given by \cref{eq:taylorexp}, we have
	\begin{align}
		&\nabla \psi_m(\theta,\Lambda) \cdot k_m(\theta,\Lambda) \\
		&\myquad[4]=\frac{\der{g_m(\theta,\Lambda_\theta^{(M_m)})}}{\der{\Lambda}} \norm{\nabla_\theta \Lambda_\theta^{(M_m)}}^2 -   \left(\Lambda - \Lambda^{(M_m)}_\theta\right) \frac{\der{f_m}(\theta,\Lambda_\theta^{(M_m)})}{\der{\Lambda}}\cdot \nabla_\theta \Lambda_\theta^{(M_m)} \\
		&\myquad[6]- \left(\Lambda - \Lambda^{(M_m)}_\theta\right)^2 \frac{\dertwo{f_m}(\theta,\Lambda_\theta^{(M_m)})}{\der{\Lambda^2}}\cdot \nabla_\theta \Lambda_\theta^{(M_m)} +  O\left(\left(\Lambda - \Lambda^{(M_m)}_\theta\right)^3\right) \\%afadf
		&\myquad[6] + \frac{1}{3\zeta_2} \left(\Lambda - \Lambda^{(M_m)}_\theta\right) \frac{\der{g_m(\theta,\Lambda_\theta^{(M_m)})}}{\der{\Lambda}} \norm{\nabla_\theta \Lambda_\theta^{(M_m)}}^2 \\
		&\myquad[6] + \frac{1}{3\zeta_2}\left(\Lambda - \Lambda^{(M_m)}_\theta\right)^2 \frac{\der{f_m}(\theta,\Lambda_\theta^{(M_m)})}{\der{\Lambda}}\cdot \nabla_\theta \Lambda_\theta^{(M_m)} + O\left(\left(\Lambda - \Lambda^{(M_m)}_\theta\right)^3\right) \\
		&\myquad[6] + \frac{\eta}{3\zeta_2}\left(\Lambda - \Lambda^{(M_m)}_\theta\right)^2 \frac{\der{g_m}(\theta,\Lambda_\theta^{(M_m)})}{\der{\Lambda}} + O\left(\left(\Lambda - \Lambda^{(M_m)}_\theta\right)^3\right).
	\end{align}
	Using Cauchy–Schwarz inequality, we get
	\begin{align}
		&\nabla \psi_m(\theta,\Lambda) \cdot k_m(\theta,\Lambda)
		\\
		&\qquad\leq \frac{\der{g_m(\theta,\Lambda_\theta^{(M_m)})}}{\der{\Lambda}} \norm{\nabla_\theta \Lambda_\theta^{(M_m)}}^2 \left(1 - \frac{1}{3\zeta_2} \left|\Lambda - \Lambda^{(M_m)}_\theta\right|\right)\\
		&\qquad + \left|\Lambda - \Lambda^{(M_m)}_\theta\right| \norm{\nabla_\theta \Lambda_\theta^{(M_m)}}\norm{\frac{\der{f_m}(\theta,\Lambda_\theta^{(M_m)})}{\der{\Lambda}}}\left( 1 + \frac{1}{3\zeta_2}\left|\Lambda - \Lambda^{(M_m)}_\theta\right|\right)\\
		&\qquad + \left(\Lambda - \Lambda^{(M_m)}_\theta\right)^2 \norm{\nabla_\theta \Lambda_\theta^{(M_m)}}\norm{\frac{\dertwo{f_m}(\theta,\Lambda_\theta^{(M_m)})}{\der{\Lambda^2}}}+ \frac{\eta}{3\zeta_2}\left(\Lambda - \Lambda^{(M_m)}_\theta\right)^2 \frac{\der{g_m}(\theta,\Lambda_\theta^{(M_m)})}{\der{\Lambda}} \\
		&\qquad+ O\left(\left(\Lambda - \Lambda^{(M_m)}_\theta\right)^3\right).
	\end{align}
    Notice that by \cref{lem:uniflowerbd_dg},  $\frac{\der{g_m(\theta,\Lambda_\theta^{(M_m)})}}{\der{\Lambda}}$ is uniformly negative. Moreover, all the other derivative terms are uniformly bounded. Next, we show that by choosing $\zeta_2$ small enough, the first and the last terms dominate all the other terms. Specifically,
	assuming $\left|\Lambda - \Lambda^{(M_m)}_\theta\right| < \zeta_2$, where
	\begin{align}
		0< \zeta_2 < \frac{\eta}{4} \times \left[\inf_{\theta\in\real^l,M>1} \left|\frac{\der{g^{(M)}(\theta,\alpha\overline{C})}}{\der{\Lambda}}\right|^2\right] \bigg/ \left[\sup_{\theta\in\real^l,M>1}\norm{\frac{\der{f^{(M)}}(\theta,\Lambda_\theta^{(M)})}{\der{\Lambda}}}^2\right],
	\end{align}
	we have
	\begin{align}
		&\nabla \psi_m(\theta,\Lambda) \cdot k_m(\theta,\Lambda)	\\
		&\qquad\leq \frac{\der{g_m(\theta,\Lambda_\theta^{(M_m)})}}{\der{\Lambda}} \left( \sqrt{\frac{2}{3}} \norm{\nabla_\theta \Lambda_\theta^{(M_m)}} - \sqrt{\frac{\eta}{6\zeta_2}}\left|\Lambda - \Lambda^{(M_m)}_\theta\right| \right)^2\\
		&\qquad - \frac{1}{6\zeta_2}\left(\Lambda - \Lambda^{(M_m)}_\theta\right)^2 \Bigg(-\eta\frac{\der{g_m(\theta,\Lambda_\theta^{(M_m)})}}{\der{\Lambda}}\\
		&\myquad[10] - 6\zeta_2 \norm{\nabla_\theta \Lambda_\theta^{(M_m)}}\norm{\frac{\dertwo{f_m}(\theta,\Lambda_\theta^{(M_m)})}{\der{\Lambda^2}}} + O\left(\Lambda - \Lambda^{(M_m)}_\theta\right) \Bigg)
		&\leq0
	\end{align}
	where the last inequality follows by the small choice of $\zeta_2$. Notice that by \cref{lem:gradgfgmfm} and \cref{cor:unfbound_lambdagrad}, $\norm{\nabla_\theta \Lambda_\theta^{(M_m)}}$, $\norm{\frac{\der{f_m}(\theta,\Lambda_\theta^{(M_m)})}{\der{\Lambda}}}$, $\norm{\frac{\dertwo{f_m}(\theta,\Lambda_\theta^{(M_m)})}{\der{\Lambda^2}}}$, and the error term of the Taylor expansion over $\left|\Lambda-\Lambda_\theta^{(M_m)}\right|<\zeta_2$ are all uniformly bounded for small enough $\zeta_2$. Moreover, by \cref{lem:uniflowerbd_dg}, $ \frac{\der{g_m(\theta,\Lambda_\theta^{(M_m)})}}{\der{\Lambda}}<0$  and it is uniformly bounded away from zero over $\theta\in\real^l$ and $m\in\nplus$.
\end{proof}
\begin{corollary}\label[corollary]{cor:negdrift_psi}
	Fix a constant $\nu < \zeta_2$. Suppose that for some $m$ and $n$ ($m < n$), we have
	\begin{align}
		&|\widetilde{\Lambda}_i - \Lambda^{(M_i)}_{\theta_i}| < \zeta_2,\quad&&\forall i\in\{m+1,m+2,\cdots,n-1\},\\
		&|\widetilde{\Lambda}_i - \Lambda^{(M_i)}_{\theta_i}| \leq \nu,\quad&& i\in \{m,n\}.
	\end{align}
	Then,
	\begin{align}
		&\Lambda^{(M_n)}_{\theta_{n}} \leq \Lambda^{(M_m)}_{\theta_{m}} +\frac{\nu^2}{3\zeta_2}+ \sum_{i=m}^{n-1}\Delta_i(\psi_i) + \sum_{i=m}^{n-1}\left|\psi_{i+1}(r_{i+1}) - \psi_{i}(r_{i+1})\right|,\\
		&\tilde{\Lambda}_{n} \leq \tilde{\Lambda}_{m} +\frac{\nu^2}{3\zeta_2} +2\nu + \sum_{i=m}^{n-1}\Delta_i(\psi_i) + \sum_{i=m}^{n-1}\left|\psi_{i+1}(r_{i+1}) - \psi_{i}(r_{i+1})\right|.
	\end{align}
\end{corollary}
\begin{proof}
	By \cref{lem:negdrift_psi}, using telescoping sum we have
	\begin{align}
		\psi_{n}(r_{n}) &\leq \psi_{m}(r_{m}) + \sum_{i=m}^{n-1}\Delta_i(\psi_i) + \sum_{i=m}^{n-1}\left|\psi_{i+1}(r_{i+1}) - \psi_{i}(r_{i+1})\right|.
	\end{align}
	The result follows by the fact that $\left|\psi_{i}(r_{i}) - \Lambda^{(M_i)}_{\theta_{i}}\right| < \frac{\nu^2}{6\zeta_2}$ and $\left|\tilde{\Lambda}_i - \Lambda^{(M_i)}_{\theta_{i}}\right| < \nu$ for $i\in\{m,n\}$.
\end{proof}

\begin{lemma}\label[lemma]{lem:limsup_psi}
	$\limsup_{m\to\infty} \tilde{\Lambda}_m - \Lambda^{(M_m)}_{\theta_m} \leq 0$
\end{lemma}
\begin{proof}
	For the sake of contradiction, suppose that $\limsup_{m\to\infty} \tilde{\Lambda}_m - \Lambda^{(M_m)}_{\theta_m} > 0$. Let $\zeta = \zeta_1\wedge\zeta_2\wedge\delta_1$, where the value of $\delta_1$ is given by \cref{cor:deltagmfm}. Let $A\in(0,\zeta)$ to be small enough such that
	\begin{align}
		&\limsup_{m\to\infty} \tilde{\Lambda}_m - \Lambda^{(M_m)}_{\theta_m} > A,
	\end{align}
	and suppose that $\limsup_{m\to\infty}\sup_{\theta\in\real^l} g_m(\theta,\Lambda^{(M_m)}_{\theta} + A/2) = 1-\varepsilon$ for some $\varepsilon>0$. The fact that $\varepsilon > 0$ follows by \cref{lem:uniflowerbd_dg}, using a similar argument as in the proof of \cref{lem:liminf_abs}. Fix $\nu \in (0,A/2)$.

	Let $\{m_0,m_1,m_2,\cdots\}$ be a sequence of integers such that for any $k\in\nzero$, during the $k^\text{th}$ even cycle (from time-step $m_{2k}$ to $m_{2k+1}-1$) we have
	\begin{align}
		&\left|\tilde{\Lambda}_{m_{2k}} - \Lambda^{(M_{m_{2k}})}_{\theta_{m_{2k}}}\right| \leq \nu ,\qquad 0 < \tilde{\Lambda}_{m_{2k+1}-1} - \Lambda^{(M_{m_{2k+1}-1})}_{\theta_{m_{2k+1}-1}} \leq \nu\\
		&\left|\tilde{\Lambda}_m - \Lambda^{(M_m)}_{\theta_{m}}\right| \leq A,\qquad\forall m\in\{m_{2k}+1,\cdots,m_{2k+1}-2\},
	\end{align}
	and during the $k^\text{th}$ odd cycle (from time-step $m_{2k+1}$ to $m_{2(k+1)}-1$), we have,
	\begin{align}
		&
		\begin{aligned}
			&\exists a_{2k+1},b_{2k+1}:&& \nu  \leq \tilde{\Lambda}_{a_{2k+1}} - \Lambda^{(M_{a_{2k+1}})}_{\theta_{a_{2k+1}}} < \frac{A}{2}, \qquad A < \tilde{\Lambda}_{b_{2k+1}} - \Lambda^{(M_{b_{2k+1}})}_{\theta_{b_{2k+1}}} ,\\
			&~&&\text{ and }\frac{A}{2} \leq \tilde{\Lambda}_m - \Lambda^{(M_m)}_{\theta_{m}} \leq A ,\qquad\forall m: a_{2k+1}<m<b_{2k+1},
		\end{aligned}\\
		& \tilde{\Lambda}_m - \Lambda^{(M_m)}_{\theta_{m}} \geq \nu,\qquad\forall m\in\{m_{2k+1},\cdots,m_{2(k+1)}-1\}.
	\end{align}
	In particular, at the beginning and the end of each even cycle, the value of $\left|\tilde{\Lambda}_m - \Lambda^{(M_m)}_{\theta_{m}}\right|$ is smaller than $\nu$. Moreover, $\left|\tilde{\Lambda}_m - \Lambda^{(M_m)}_{\theta_{m}}\right| < A$ during any such cycle. On contrary, $\left|\tilde{\Lambda}_m - \Lambda^{(M_m)}_{\theta_{m}}\right| \geq \nu$ during any even cycle, and the value of $\left|\tilde{\Lambda}_m - \Lambda^{(M_m)}_{\theta_{m}}\right|$ crosses $A$ at some point.

	Notice that by the contradiction assumption and \cref{lem:liminf_abs}, such a sequence $\{m_i\}\subset \nplus$ exists. We assume $m_0$ is large enough so that for all $m\geq m_0$, if $\Lambda - \Lambda^{(M_m)}_{\theta} \geq A/2$, then $g_m(\theta,{\Lambda}) \leq 1-\varepsilon/2$ for any $\theta\in\real^l$.

	Consider the $k^\text{th}$ even cycle. By the choice of $A$, $\nu$ and \cref{cor:negdrift_psi}, we have
	\begin{align}
		\tilde{\Lambda}_{m_{2k+1}-1} \leq \tilde{\Lambda}_{m_{2k}} +\frac{\nu^2}{3\zeta} +2\nu + \sum_{m=m_{2k}}^{m_{2k+1}-1}\Delta_m(\psi_m) + \sum_{m=m_{2k}}^{m_{2k+1}-1}\left|\psi_{m+1}(r_{m+1}) - \psi_{m}(r_{m+1})\right|.
	\end{align}
	Hence, by choosing $\nu$ to be small enough and $m_0$ to be large enough, we can make $\tilde{\Lambda}_{m_{2k+1}-1} - \tilde{\Lambda}_{m_{2k}}$ to be as small as we like. That is to say, the value of $\tilde{\Lambda}_{m}$ cannot increase that much at any even cycle.

	Next, let us consider the $k^\text{th}$ odd cycle. Using $\{\phi_m(\cdot)\}$ as the family of Lyapunov functions, we have
	\begin{align}
		\phi_{b_{2k+1}}(r_{b_{2k+1}}) &= \phi_{a_{2k+1}}(r_{a_{2k+1}}) + \sum_{m=a_{2k+1}}^{b_{2k+1}-1}\gamma_{m} \nabla \phi_m(r_m)\cdot k_m(r_m) \\
		&\myquad[4]+ \sum_{m=a_{2k+1}}^{b_{2k+1}-1}\Delta_{m}(\phi_m) + \sum_{m=a_{2k+1}}^{b_{2k+1}-1}\Lambda_{\theta_{m+1}}^{(M_{m+1})} - \Lambda_{\theta_{m+1}}^{(M_{m})}.
	\end{align}
	Let $C_1 > 0$ be a uniform upper bound on $\left|\nabla \phi_m(r_m)\cdot k_m(r_m)\right|$  independent of $m.$ The existence of such an upper bound is guaranteed by \cref{lem:gradgfgmfm} and the choice of $\zeta$. We have
	\begin{align}
		\frac{A}{2} &< \phi_{b_{2k+1}}(r_{b_{2k+1}}) - \phi_{a_{2k+1}}(r_{a_{2k+1}}) \\
		&\leq \sum_{m=a_{2k+1}}^{b_{2k+1}-1}\gamma_{m}C_1 + \sum_{m=a_{2k+1}}^{b_{2k+1}-1}\Delta_{m}(\phi_m) + \sum_{m=a_{2k+1}}^{b_{2k+1}-1}\Lambda_{\theta_{m+1}}^{(M_{m+1})} - \Lambda_{\theta_{m+1}}^{(M_{m})}.
	\end{align}
	Assuming $m_0$ is large enough (so that $\gamma_{a_{2k+1}}$ as well as the error term is small), we have
	\begin{align}
		\frac{A}{4C_1} < \sum_{m=a_{2k+1}+1}^{b_{2k+1}-1}\gamma_{m}. \label{eq:gamma_bound}
	\end{align}
	Now, using \cref{eq:gamma_bound}, the fact that $g_m(r_{m}) < 1-\varepsilon/2$ for all $a_{2k+1}<m<b_{2k+1}$, the fact that $g(r_{m}) \leq 1$ for all $m_{2k+1}-1\leq m<m_{2(k+1)}-1$, and assuming $m_0$ is large enough, we have
	\begin{align}
		\tilde{\Lambda}_{m_{2(k+1)}} &= \tilde{\Lambda}_{m_{2k+1}-1} + \sum_{m=m_{2k+1}-1}^{m_{2(k+1)}-1}\gamma_m \eta (g(r_m)-1) + \sum_{m=m_{2k+1}-1}^{m_{2(k+1)}-1}\epsilon_m(\Lambda) \\
		&\leq  \tilde{\Lambda}_{m_{2k+1}-1} + \sum_{m=a_{2k+1}+1}^{b_{2k+1}-1}\gamma_m \eta (g(r_m)-1)  + \sum_{m=m_{2k+1}-1}^{m_{2(k+1)}-1}\epsilon_m(\Lambda) \\
		&\leq  \tilde{\Lambda}_{m_{2k+1}-1} - \frac{A}{4C_1} \times \frac{\varepsilon}{2} \eta  + \sum_{m=m_{2k+1}-1}^{m_{2(k+1)}-1}\epsilon_m(\Lambda) \leq \tilde{\Lambda}_{m_{2k+1}-1} - \frac{A\epsilon}{16 C_1} \eta.
	\end{align}
	Notice that we can pick $m_0$ to be large enough such that the error terms are all sufficiently small. This follows by the assumption on$\{M_m\}_{m\geq0}$ and the fact that the error in the approximation of the family of Lyapunov functions is summable. Hence, at any odd cycle, the value of $\tilde{\Lambda}_{m}$ decreases by a constant independent of $\nu$ (for all sufficiently small choices of $\nu>0$). This implies that $\tilde{\Lambda}_m\to-\infty$ which is a contradiction.
\end{proof}

\begin{corollary}\label[corollary]{cor:limsup}
	$\limsup_{m\to\infty} \tilde{\Lambda}_m - \Lambda_\theta \leq 0$.
\end{corollary}
\begin{proof}
	Follows by \cref{cor:unifapprox_truncost} and \cref{lem:limsup_psi}.
\end{proof}

\begin{corollary}\label[corollary]{cor:convergence}
	$\lim_{m\to\infty} \tilde{\Lambda}_m - \Lambda_{\theta_m} = 0$. Moreover, $\lim_{m\to\infty}\Lambda_{\theta_m}  = \Lambda^*$ for some constant $\Lambda^* \in[-\alpha\underline{C},\alpha\overline{C}]$.
\end{corollary}
\begin{proof}
	The first part follows by \cref{cor:liminf} and \cref{cor:limsup}. Hence, for any $0< \nu < \zeta_2$ there exists $m_0(\nu)$ large enough so that for all $m>m_0(\nu)$, we have $\left|\tilde{\Lambda}_m - \Lambda_{\theta_m}\right| < \nu$. By \cref{cor:negdrift_psi}, for all $n>m>m_0(\nu)$, we have
	\begin{align}
		\tilde{\Lambda}_{n} &\leq \tilde{\Lambda}_{m} +\frac{\nu^2}{3\zeta_2} +2\nu + \sum_{i=m}^{n-1}\Delta_i(\psi_i) + \sum_{i=m}^{n-1}\left|\psi_{i+1}(r_{i+1}) - \psi_{i}(r_{i+1})\right| \\
		&\leq \tilde{\Lambda}_{m} +\frac{\nu^2}{3\zeta_2} +3\nu,
	\end{align}
	where the last inequality follows by assuming $m_0(\nu)$ is large enough. Hence, taking $\limsup$ from the left-hand side and then $\liminf$ from the right-hand side, we get
	\begin{align}
		\limsup_{n\to\infty}\tilde{\Lambda}_{n} \leq \liminf_{m\to\infty}\tilde{\Lambda}_{m} +\frac{\nu^2}{3\zeta_2} +3\nu.
	\end{align}
	Since the choice of $\nu > 0$ was arbitrary, $\lim_{n\to\infty}\tilde{\Lambda}_{n} = \Lambda^*$ for some constant $\Lambda^* \in[-\alpha\underline{C},\alpha\overline{C}]$ and the result follows.
\end{proof}

\subsection*{Step II, Convergence of $\nabla_\theta\Lambda_{\theta_m}$:}
Using a similar argument as the previous step, we first show that $\liminf_{m\to\infty}\norm{\nabla_\theta{\Lambda}_{\theta_m}} = 0$, and then $\limsup_{m\to\infty}\norm{\nabla_\theta{\Lambda}_{\theta_m}} = 0$. The idea is to use the fact that $\Lambda_{\theta_m}$ converges. Naturally, we use the function $\kappa_m(\theta,\Lambda) = \Lambda_{\theta}$ as Lyapunov function.
\begin{lemma}\label[lemma]{lem:liminf_nabla}
	$\liminf_{m\to\infty}\norm{\nabla_\theta{\Lambda}_{\theta_m}} = 0.$
\end{lemma}
\begin{proof}
	For the sake of contradiction, suppose that $\liminf_{m\to\infty}\norm{\nabla_\theta{\Lambda}_{\theta_m}} > \varepsilon$. Notice that
	\begin{align}
		\Lambda_{\theta_{m+1}} = \Lambda_{\theta_m} + \gamma_m \nabla_\theta \Lambda_{\theta_m} \cdot k_m(r_m) + \Delta_{m}(\kappa_m).
	\end{align}
	Since $\Lambda_{\theta_m}$ and $\tilde{\Lambda}_m$ both converges to the same constant, using \cref{lem:gradgfgmfm}, \cref{cor:unifapprox_truncost}, and \cref{cor:unifapprox_nablatheta}, for all large enough $m$, we have $\nabla_\theta \Lambda_{\theta_m} \cdot k_m(r_m) = -\norm{\nabla_\theta \Lambda_{\theta_m}}^2 + \nu_m$ where $|\nu_m| < \varepsilon^2/2$. Hence,
	\begin{align}
		\Lambda_{\theta_{m+1}} \leq \Lambda_{\theta_m} - \gamma_m (\norm{\nabla_\theta \Lambda_{\theta_m}}^2 - |\nu_m|)+ \Delta_{m}(\kappa_m).
	\end{align}
	Using a telescoping sum, we have $\Lambda_{\theta_n} \to -\infty$, which is a contradiction.
\end{proof}
\begin{lemma}\label[lemma]{lem:limsup_nabla}
	$\limsup_{m\to\infty}\norm{\nabla_\theta{\Lambda}_{\theta_m}} = 0$.
\end{lemma}
\begin{proof}
	For the sake of contradiction, suppose that $\limsup_{m\to\infty}\norm{\nabla_\theta{\Lambda}_{\theta_m}} =\varepsilon> 0$. By \cref{lem:liminf_nabla}, there are infinitely many sequences $\{m,m+1,\cdots,n\}$ such that $ \norm{\nabla_\theta{\Lambda}_{\theta_m}} < \varepsilon/3$, $\norm{\nabla_\theta{\Lambda}_{\theta_n}} > 2\varepsilon/3$ and $\norm{\nabla_\theta{\Lambda}_{\theta_i}} \in (\varepsilon/3,2\varepsilon/3)$ for all $i\in\{m+1,m+2,\cdots,n-1\}$. By \cref{lem:gradgfgmfm} and \cref{cor:convergence}, we have
	\begin{align}
		\varepsilon/3 &\leq \norm{\nabla_\theta{\Lambda}_{\theta_n}} - \norm{\nabla_\theta{\Lambda}_{\theta_m}} \leq \norm{\nabla_\theta{\Lambda}_{\theta_n} - \nabla_\theta{\Lambda}_{\theta_m}}\\
		&\leq \textrm{Const} \norm{\theta_n - \theta_m} = \textrm{Const} \norm{ - \sum_{i=m}^{n-1} \gamma_i f_i(r_i) + \sum_{i=m}^{n-1} \epsilon_i(\theta)}\\
		&\leq \textrm{Const} \sum_{i=m}^{n-1} \gamma_i \norm{f_i(r_i)} + \textrm{Const}\norm{\sum_{i=m}^{n-1} \epsilon_i(\theta)}
	\end{align}
	Assuming $m$ is large enough and using \cref{lem:gradgfgmfm}, we get $\sum_{i=m}^{n-1} \gamma_i > \varepsilon/(6 C_1  \textrm{Const})$, where $C_1$ is the uniform upperbound over $\norm{f_i(r_i)}$. Notice that by a similar argument as in the proof of \cref{lem:liminf_nabla}, for all large enough $i$ we have $\nabla_\theta \Lambda_{\theta_i} \cdot k_i(r_i) = -\norm{\nabla_\theta \Lambda_{\theta_i}}^2 + \nu_i$, where $|\nu_i| < \varepsilon^2/18$. Hence, we have
	\begin{align}
		\Lambda_{\theta_{n}} &\leq \Lambda_{\theta_{m+1}} - \sum_{i=m+1}^{n-1}\gamma_i (\norm{\nabla_\theta \Lambda_{\theta_i}}^2 - |\nu_i|)+ \sum_{i=m}^{n-1}\Delta_{i}(\kappa_i)\\
		&\leq \Lambda_{\theta_{m+1}} - \frac{\varepsilon}{6 C_1  \textrm{Const}} \times \frac{\varepsilon^2}{18} + \sum_{i=m}^{n-1}\Delta_{i}(\kappa_i)
	\end{align}
	which contradicts with the fact that $\Lambda_{\theta_{n}}$ converges, as the error term is arbitrary small for all $m$ large enough.
\end{proof}

	\section{Proof of \cref{prop:costest}}\label{proof:prop:costest}
	Let us rewrite the update equation \cref{eq:prop:costest_vartrunc} as follows:
\begin{align}
	\tilde{\Lambda}_{m+1} = \tilde{\Lambda}_{m} + \gamma_m\left(g^{(M_m)}(\theta,\tilde{\Lambda}_{m}) - 1\right) + \epsilon_m
\end{align}
where $\epsilon_m$ is the stochastic noise. Notice that $\left|\epsilon_m\right| \leq 2\gamma_m M_m = 2 \gamma_m ^ {1-\beta}$. Let $\mathcal{F}_m \coloneqq \{\Phi_0,\Phi_1,\cdots,\Phi_{t_m}\}$ denote the history of Markov chain up to the $m$th visit to the recurrent state $x^*$. It is easy to verify that $\{\epsilon_m \}_{m \geq 0}$ is a martingale difference sequence with respect to the filtration given by $\{\mathcal{F}_m\}\cref{ass:1,ass:2}$. In addition, we have $\expect{}{}{\left|\epsilon_m\right|^2} < 2 \gamma_m ^ {2(1-\beta)}$ and $\sum_{m} \gamma_m ^{2(1-\beta)} < \infty$. Hence, by the martingale convergence theorem $\sum_m \epsilon_m \to \epsilon$ with $\expect{}{}{|\epsilon|} < \infty$, and in particular, $\sum_m \epsilon_m < \infty$ almost surely.

Next, we show that for any constant $\delta > 0$, for almost every sample path, after some time, the value of $\tilde{\Lambda}_m$ remains in a $\delta$-neighborhood of $\Lambda_\theta$.
Fix a sample path $\omega$ for which $\sum_m \epsilon_m < \infty$, and a constant $\delta > 0$. Let $N_\omega\in\nplus $ be large enough so that for all $\ell_1,\ell_2 \geq N_\omega$, we have $\big|\sum_{m=\ell_1}^{\ell_2} \epsilon_m \big| < \delta/4$. Let $N_\delta$ be large enough so that for all $m \geq N_\delta$, we have $\gamma_m^{1-\beta} < {\delta}/{4}$ and $\Lambda_\theta - \Lambda^{(M_m)}_\theta < \delta/8$. Notice that by \cref{lem:unifapprox_truncost} and the fact that $M_m\uparrow \infty$, such an $N_\delta>0$ exists for any $\delta > 0$. Suppose that $N_{\delta,\omega} > \max(N_\delta,N_\omega)$ be the first index for which $\tilde{\Lambda}_{N_{\delta,\omega}} \in [\Lambda_\theta - \delta/4,\Lambda_\theta + \delta/4]$. We claim that $N_{\delta,\omega} < \infty$ , and that for all $m \geq N_{\delta,\omega}$ we have $\tilde{\Lambda}_m \in [\Lambda_\theta - \delta,\Lambda_\theta + \delta]$. The results follows by these two claims and the arbitrary choice of $\delta>0$.

\begin{claim}\label[claim]{claim_app:Nfinite}
	There exists $ N_{\delta,\omega} \in (\max(N_\delta,N_\omega),\infty)$ for which $\tilde{\Lambda}_{N_{\delta,\omega}} \in [\Lambda_\theta - \delta/4,\Lambda_\theta + \delta/4]$.
\end{claim}
\begin{proof}[proof of claim]
	For the sake of contradiction, suppose that for all $m>\max(N_\delta,N_\omega)$, we have $\big|\tilde{\Lambda}_m - \Lambda_\theta\big| > \delta/4$. Since $\big|\tilde{\Lambda}_{m+1} - \tilde{\Lambda}_{m}\big| \leq \gamma_m^{1-\beta}\to 0$, after some $m_0$, we have either $\tilde{\Lambda}_m - \Lambda_\theta > \delta/4$ for all $m\geq m_0$, or $\Lambda_\theta - \tilde{\Lambda}_m > \delta/4$ for all $m\geq m_0$. Without loss of generality, let us consider the latter case, i.e., $\tilde{\Lambda}_m - \Lambda_\theta < -\delta/4$ for all $m \geq m_0$. Notice that for all $\ell\in\nplus$, we have
	\begin{align}
		\tilde{\Lambda}_{m_0 + \ell} -\tilde{\Lambda}_{m_0} &=  \sum_{m = m_0}^{m_0+\ell-1} \gamma_m\left(g^{(M_m)}(\theta,\tilde{\Lambda}_{m}) - 1\right) +  \sum_{m = m_0}^{m_0+\ell-1} \epsilon_m  \\
		&\geq  \left(g^{(M_m)}(\theta,\Lambda^{(M_m)}_\theta - \delta/8) - 1\right) \sum_{m = m_0}^{m_0+\ell-1} \gamma_m - \delta/4
	\end{align}
	where the inequality follows by the fact that for any fixed $\theta\in\real^l$, $g^{(M_m)}(\theta,\cdot)$ is a decreasing function, and that $\tilde{\Lambda}_m - \Lambda^{(M_m)}_\theta < -\delta/8$ for all $m \geq m_0$. Notice that the right-hand side of the above inequality goes to $+\infty$ as $\ell\to\infty$, which contradicts with the assumption that $\tilde{\Lambda}_m < \Lambda_\theta-\delta/4$ for all $m \geq m_0$.
\end{proof}

\begin{claim}
	For all $m > N_{\delta,\omega}$, we have $\tilde{\Lambda}_{m} \in [\Lambda_\theta - \delta,\Lambda_\theta + \delta]$.
\end{claim}
\begin{proof}
	Let $\underline{\ell} > N_{\delta,\omega}$ denote the first iteration after $N_{\delta,\omega}$ at which $\tilde{\Lambda}_{\underline{\ell}} \notin [\Lambda_\theta - \delta/4,\Lambda_\theta + \delta/4]$. Without loss of generality, suppose that $\tilde{\Lambda}_{\underline{\ell}} < \Lambda_\theta - \delta/4$. Let $\overline{\ell} > \underline{\ell}$ denote the first iteration after $\underline{\ell}$ at which $\tilde{\Lambda}_{\overline{\ell}} \in [\Lambda_\theta - \delta/4,\Lambda_\theta + \delta/4]$. Notice that $\overline{\ell}<\infty$ by the same argument as in \cref{claim_app:Nfinite}. Also, notice that by the choice of $N_{\delta,\omega}$ and the fact that $\left|\tilde{\Lambda}_{m+1} - \tilde{\Lambda}_{m}\right| < \delta_m^{1-\beta}$ for any $m\in\nplus$, we have $\tilde{\Lambda}_{\ell} < \Lambda_\theta - \delta/4$ for all $\ell \in \{\underline{\ell},\underline{\ell}+1,\cdots,\overline{\ell}-1\}$. Hence, for any such $\ell$, we have
	\begin{align}
		\tilde{\Lambda}_{\ell} &= \tilde{\Lambda}_{\underline{\ell}} + \sum_{m=\underline{\ell}}^{\ell-1} \gamma_m\left(g^{(M_m)}(\theta,\tilde{\Lambda}_{m}) - 1\right) +  \sum_{m=\underline{\ell}}^{\ell-1} \epsilon_m \\
		&\geq \tilde{\Lambda}_{\underline{\ell}} + \sum_{m=\underline{\ell}}^{\ell-1} \epsilon_m \geq \Lambda_\theta - \delta/4 + \delta_{\underline{\ell}}^{1-\beta} - \delta/4 \geq \Lambda_\theta - 3\delta/4.
	\end{align}
	In particular, $\tilde{\Lambda}_{\ell} \in [\Lambda_\theta - \delta,\Lambda_\theta + \delta]$ for all $\ell \in \{\underline{\ell},\underline{\ell}+1,\cdots,\overline{\ell}-1\}$.
\end{proof}

	\section{Risk-Sensitive Markov Decision Processes}\label{sec:riskmdp}
	In the previous section, we considered an abstract problem where the state transition probabilities of a Markov chain are parameterized by a parameter $\theta\in\real^l$, and we derived the gradient of the cost with respect to $\theta.$ In this section, we show how to apply the analysis in the previous sections to an MDP, where action is chosen according to some paramertized (according to parameter vector $\theta\in\real^l$) probability distribution. We present an algorithm that updates the policy at each visit to the regeneration state $s^*$. The analysis of \cref{sec:algorithm} applies trivially to the new setting.

Consider a discrete-time Markov decision process $\{\Phi_i,\U_i\}_{i\geq 0}$ with finite state space $\statespace$ and finite action space $\actionspace$. For any state-action pair $(s,a)\in\statespace\times\actionspace$, the transition probability is given by
\begin{align}
	P(s,a,s') \coloneqq \prob{\Phi_{i+1} = s' \vert \Phi_i = s, \U_i = a} \qquad\forall s'\in \statespace,\text{ and }i\in \nzero.
\end{align}
Suppose that we have access to a set of policies, parameterized by $\theta\in\real^l$, such that
\begin{align}
	\mu_\theta(s,a) \coloneqq \prob{\U_{i} = a \vert \Phi_i = s} \qquad\forall i\in \nzero.
\end{align}
Let $C:\statespace\times \actionspace\to \real$ denote the one-step cost function. The objective is to find the policy $\theta$ that optimizes the risk-sensitive cost
\begin{align}\label{eq:deflambdamdp}
	\Lambda_\theta \coloneqq \lim_{n\to\infty}\frac{1}{n}\ln\expect{\Phi_0=s}{\theta}{\expon{\alpha \sum_{i=0}^{n-1}C(\Phi_i,\U_i)}},\qquad \forall s\in \statespace,
\end{align}
where $\alpha >0$ is the risk factor and $\expect{\Phi_0=s}{\theta}{\cdot}$ denotes the expectation with respect to the policy $\theta$. Assuming the resulting chain is aperiodic and recurrent, the convergence of the above limit follows by the multiplicative ergodic theorem~\cite[Theorem 1.2]{Balaji2000}. The same also follows by a similar argument as in \cref{sec_app:back_risk}. In particular, let $\hat{P}_\theta(s,s') = \sum_{a\in\actionspace}\expon{\alpha C(s,a)}\mu_\theta(s,a) P(s,a,s')$, for any $s,s'\in\statespace;$ assuming $\hat{P}_\theta(s,s')$ is primitive, $\lambda_\theta\coloneqq\expon{\Lambda_\theta}$ is the largest eigenvalue of $\hat{P}_\theta$ with multiplicity $1$. Moreover, it follows that the relative risk-sensitive value function is uniquely determined by the equation
\begin{align}
	h_\theta(s) &= \expect{\Phi_0=s}{\theta}{\expon{\sum_{i=0}^{\tau_{s^*}-1}\left(\alpha C(\Phi_i,\U_i)- \Lambda_\theta\right)}}=\checkexpect{\check{\Phi}_0=s}{\theta}{\expon{\tau_{s^*}(\Lambda_\theta-\Lambda)}}
\end{align}
up to a constant factor, where $\checkexpect{\check{\Phi}_0=s}{\theta}{\cdot}$ denotes the expectation with respect to the twisted kernel $\check{P}_\theta$, and $s^*$ is a recurrent state. All the analysis and assumptions of the previous sections extend naturally to the case of MDP with proper rewording.

Let $P_\theta(s,s')\coloneqq \sum_{a\in\actionspace}\mu_\theta(s,a) P(s,a,s')$, and define $\probspace= \left\{ P_\theta:\theta\in\mathbb{R}^l \right\}$. Let $\xbar{\probspace}$ denote the closure of $\probspace$ in the space of $|\statespace|\times|\statespace|$ matrices. \cref{ass:1,ass:2,ass:3} translate into the following assumption for MDPs.
\begin{assumption}\label[assumption]{ass:6}
	{\normalfont (i)} For each $P\in \xbar{\probspace}$, the Markov chain with transition probability $P$ is aperiodic and irreducible with a common recurrent state $s^*$. {\normalfont (ii)} For any $s,a\in\statespace\times\actionspace$, $\mu_\theta(s,a)$ is bounded, twice differentiable, and has bounded first and second derivatives. {\normalfont (iii)} For any $(s,a)\in\statespace\times\actionspace$, there exist bounded functions $L_\theta(s,a)$ and $L^{(2)}_\theta(s,a)$ such that $({\normalfont a})~\nabla_\theta \mu_\theta(s,a) = \mu_\theta(s,a) L_\theta(s,a)$ and $({\normalfont b})~ \nabla^2_\theta \mu_\theta(s,a) = \mu_\theta(s,a) L^{(2)}_\theta(s,a).$
\end{assumption}

Given \cref{ass:6} and following a similar argument as in \cref{sec:policygrad}, we get the following risk-sensitive formula for $\nabla_\theta \Lambda_\theta$ in terms of visits to the recurrent state $s^*$:
\begin{align}
	\checkexpect{\check{\Phi}_0=s^*}{\theta}{\tau_{s^*}} \nabla_\theta \Lambda_\theta &= \expect{\Phi_0 = s^*}{\theta}{\left(\sum_{i=0}^{\tau_{s^*}-1}L_\theta\left({\Phi}_i,{\U}_{i}\right)\right)\expon{\sum_{i=0}^{\tau_{s^*}-1}\left(\alpha C\left({\Phi}_i,{\U}_{i}\right)- \Lambda_\theta\right)}}\label{eq:gradlambda_regeneration_mdp_extended}
\end{align}
The same issues as we discussed in \cref{sec:costest} and \cref{sec:heuristic} arise in the case of risk-sensitive MDPs, and using the vanilla form of the above risk-sensitive formula to develop a trajectory-based algorithm using stochastic approximation will not work because the stochastic noise may not be summable. Hence, as before, we focus on a truncated and smooth approximation of the risk-sensitive cost.

Following the same steps as in \cref{sec:algorithm}, our trajectory-based gradient algorithm for the risk-sensitive MDP is same as the one given by \cref{eq:algorithm} (or \cref{eq:algorithm_project}), with the only difference being the definition of $H_m(\theta,\Lambda)$ which should be replaced with the following:
\begin{align}
	H_m(\theta,\Lambda) \coloneqq \expon{\sum_{i=t_m}^{t_{m+1}-1} \left(\alpha C_\theta({\Phi}_i,\U_i) - \Lambda\right)},
\end{align}
where $t_m$ is the $m$th visit to the recurrent state $s^*$ and $t_0 = 0$. \cref{thm:main} extends naturally to the case of risk-sensitive MDPs.

	\section{Supporting Results}\label{proof:supporting}
	\subsection{Supporting Results of \cref{sec:review}}
	%\subsection{Proof of \cref{lem:hbounded}}\label{proof:lem:hbounded}
	\begin{lemma}\label[lemma]{lem:hbounded}
		Let \cref{ass:1,ass:2} hold. For all $\theta\in\real^l$, consider the version of the risk-sensitive value function $h_\theta(\cdot)$ given by \eqref{eq:recversion}. Then, $\exists \underline{h}, \overline{h} \in \real_{++}$ such that $h_\theta(x) \in [\underline{h},\overline{h}]$ for all $x\in\samplespace$ and $\theta\in\real^l$.
	\end{lemma}
	\begin{proof}\label{proof:lem:hbounded}
		Notice that, by \cref{ass:2}, $\exists \underline{C},\overline{C} \in \real$ such that $C_\theta(x) \in [\underline{C},\overline{C}]$ for all $x\in\samplespace$ and $\theta\in\real^l$. This together with \eqref{eq:deflambda}, implies that $\Lambda_\theta \in [\alpha \underline{C},\alpha \overline{C}]$. Notice that by \eqref{eq:fixedpoint}, \eqref{eq:recversion} and \cref{ass:1}, we have $h_\theta(x^*) = 1$ for all $\theta \in \real^l$.

For the sake of contradiction, suppose that there exists $x\in\samplespace$ and a sequence $\{\theta_i\}$ such that $h_{\theta_i}(x) \to 0$. Using a compactness argument, we can choose a subsequence $\{\theta_j\}$ such that $P_{\theta_j} \to P\in\xbar{\probspace}$, $\Lambda_{\theta_j}\to\Lambda\in [\alpha \underline{C},\alpha \overline{C}]$, and $C_{\theta_j} \to C$ for some $C:\samplespace \to [\underline{C},\overline{C}]$. Notice that
\begin{align}
	h_{\theta_j}(x) = \frac{\exp(\alpha C_{\theta_j}(x))}{\lambda_{\theta_j}} 	\sum_{y\in\samplespace}P_{\theta_j}(x,y)h_{\theta_j}(y)\to 0,
\end{align}
which implies that for any $y\in\samplespace$ with $P(x,y) > 0$, we have $h_{\theta_j}(y)\to 0$. Repetitive use of the same argument together with \cref{ass:1}, implies that $h_{\theta_j}(y)\to 0$ for all $y\in\samplespace$. Contradiction follows by the fact that $h_{\theta}(x^*) = 1$ for all $\theta\in\real^l$. Hence, $\exists \underline{h} > 0$ such that $h_\theta(x) \geq \underline{h}$ for all $x\in\samplespace$ and $\theta\in\real^l$.

The proof of the existence of a uniform upper bound follows the same logic. For the sake of contradiction, suppose that there exists $x\in\samplespace$ and a sequence $\{\theta_i\}$ such that $h_{\theta_i}(x) \to \infty$. Consider a similar subsequence $\{\theta_j\}$ as above. Notice that
\begin{align}
	h_{\theta_j}(x^*) = \frac{\exp(\alpha C_{\theta_j}(x^*))}{\lambda_{\theta_j}} 	\sum_{y\in\samplespace}P_{\theta_j}(x^*,y)h_{\theta_j}(y) = 1,
\end{align}
which implies that for any $y\in\samplespace$ with $P(x^*,y) > 0$,  $h_{\theta_j}(y)$ stays bounded; in particular, for any such $y$, we have
\begin{align}
	\limsup_{j\to\infty} h_{\theta_j}(y) &\leq \exp\left(\alpha (\overline{C} - \underline{C})\right) \frac{1}{P(x^*,y)}\\
	&\leq \exp\left(\alpha (\overline{C} - \underline{C})\right) \max_{w,z:P(w,z)>0} \frac{1}{P(w,z)}
\end{align}
Notice that $\max_{w,z:P(w,z)>0} {\left(P(w,z)\right)}^{-1} < \infty$.
Repetitive use of the same argument together with \cref{ass:1}, implies that
\begin{align}
	\sup_{y\in\samplespace}\limsup_{j\to\infty} h_{\theta_j}(y) \leq \left(\exp\left(\alpha (\overline{C} - \underline{C})\right) \max_{w,z:P(w,z)>0} \frac{1}{P(w,z)}\right)^{|\samplespace|},
\end{align}
which contradicts with the assumption that $h_{\theta_j}(x)\to\infty$. Hence, $\exists \overline{h} > 0$ such that $h_\theta(x) \leq \overline{h}$ for all $x\in\samplespace$ and $\theta\in\real^l$.
	\end{proof}

	%\subsection{Proof of \cref{cor:Phat_commonrec}}\label{proof:cor:Phat_commonrec}
	\begin{corollary}\label[corollary]{cor:Phat_commonrec}
		Let \cref{ass:1,ass:2} hold. Then, \cref{ass:1} holds for $\check{\probspace}$, \ie, elements of $\,\xbarcheck{\probspace}$ are aperiodic and recurrent with the common recurrent state $x^*$.
	\end{corollary}
	\begin{proof}\label{proof:cor:Phat_commonrec}
		Let $\check{P} \in \xbarcheck{\probspace}$, and suppose that $\check{P}_{\theta_i} \to \check{P}$ where $\{\check{P}_{\theta_i}\}\subset \check{\probspace}$. Using a compactness argument, we can choose a subsequence $\{{\theta_j}\}$ such that $P_{\theta_j} \to P\in\xbar{\probspace}$, $\Lambda_{\theta_j}\to\Lambda \in [\alpha \underline{C},\alpha \overline{C}]$, $h_{\theta_j} \to h$ for some $h:\samplespace\to [\underline{h},\overline{h}]$, and $C_{\theta_j} \to C$ for some $C:\samplespace\to [\underline{C},\overline{C}]$.
It is easy to see that $\Lambda$ and $h(\cdot)$ are the risk-sensitive cost and the risk-sensitive relative value function of $P$, respectively, with one-step cost function $C(\cdot)$. Moreover, $\check{P}$ is the twisted kernel of $\hat{P} \coloneqq \expon{\alpha C} P$. Now, the result follows by \cref{ass:1} and the fact that $h(x) \in [\underline{h},\overline{h}]$ for all $x\in\samplespace$.
	\end{proof}

	\subsection{Supporting Results of \cref{sec:costest}}
	%\subsection{Proof of \cref{lem:unifapprox_truncost}}\label{proof:lem:unifapprox_truncost}
	\begin{lemma}\label[lemma]{lem:unifapprox_truncost}
		Let \cref{ass:1,ass:2} hold. For any $\delta > 0$ there exists $M_\delta$ large enough such that for all $\theta\in\real^l$, we have $0 \leq \Lambda_\theta  - \Lambda_\theta^{(M_\delta)} < \delta$.
	\end{lemma}
	\begin{proof}\label{proof:lem:unifapprox_truncost}
		Let us first show the existence of a unique $\Lambda_\theta^{(M)}$ for all $\theta\in\real^l$ and $M\in(1,\infty)$. This follows from the following two facts: $g^{(M)}(\theta,\cdot)$ is a decreasing continuous function; and, $\lim_{\Lambda\to-\infty}g^{(M)}(\theta,\Lambda) = M > 1$ and $\lim_{\Lambda\to+\infty}g^{(M)}(\theta,\Lambda) = 0$ which are both true due to the monotone convergence theorem. Notice that $\Lambda_\theta^{(M)} < \Lambda_\theta$ for all $M\in(1,\infty)$.

Next, we show that for any fixed $\theta\in\real^l$, $\Lambda_\theta^{(M)}\uparrow \Lambda_\theta$ as $M\uparrow \infty$. Consider a sequence $M_i\uparrow\infty$. Notice that for any $i<j$, we have $\Lambda_\theta^{(M_i)} < \Lambda_\theta^{(M_j)}  < \Lambda_\theta$. Hence, the sequence $\{\Lambda_\theta^{(M_i)}\}$ is an increasing and bounded sequence, which implies that $\Lambda_\theta^{(M_i)}\uparrow\Lambda_\theta^{\infty} \leq \Lambda_\theta$. We want to show that $\Lambda_\theta^{\infty} = \Lambda_\theta$. By the monotone convergence theorem, we have $g^{(M_i)}(\theta,\Lambda_\theta^{\infty}) \uparrow g(\theta,\Lambda_\theta^{\infty})$. Since $g^{(M_i)}(\theta,\Lambda_\theta^{\infty}) < 1$ for all $i$, we have $g(\theta,\Lambda_\theta^{\infty})\leq 1$ which together with $g(\theta,\cdot)$ being strictly decreasing implies that $\Lambda_\theta^{\infty} \geq \Lambda_\theta$. Hence, we have $\Lambda_\theta^{\infty} = \Lambda_\theta$. Notice that the exact same argument applies to the risk-sensitive cost associated with any transition probability $P\in\xbar{\probspace}$ and any cost function $C:\samplespace \to [\underline{C},\overline{C}]$.

Finally, we prove $\Lambda_\theta^{(M)}\uparrow \Lambda_\theta$ as $M\uparrow \infty$ uniformly over $\theta\in\real^l$. For the sake of contradiction, suppose that there exist $\delta > 0$ and sequences $M_i\uparrow\infty$ and $\{\theta_i\}$ for which $\Lambda_{\theta_i}  - \Lambda_{\theta_i}^{(M_i)} > \delta$ for all $i$. Using a compactness argument together with \cref{lem:hbounded}, we can choose a subsequence $\{{\theta_j}\}$ such that $P_{\theta_j} \to P\in\xbar{\probspace}$, $\Lambda_{\theta_j}\to\Lambda \in [\alpha \underline{C},\alpha \overline{C}]$, $h_{\theta_j} \to h$ for some $h:\samplespace\to [\underline{h},\overline{h}]$, and $C_{\theta_j} \to C$ for some $C:\samplespace\to [\underline{C},\overline{C}]$. Notice that $\Lambda$ is the risk-sensitive cost associated with the transition probability $P\in\xbar{\probspace}$ and cost function $C:\samplespace\to [\underline{C},\overline{C}]$.

Since $M_j\uparrow \infty$, we have $\Lambda^{(M_j)}\uparrow \Lambda$. Next, we claim that for any fixed $k\in\nplus$, $\Lambda^{(M_k)}_{\theta_j} \to \Lambda^{(M_k)}$.
For the sake of contradiction, suppose there exists a further subsequence $\{\theta_l\}$ such that $\Lambda^{(M_k)}_{\theta_l} \to \hat{\Lambda}\neq\Lambda^{(M_k)}$. Since $C_{\theta_l} \to C$ and $P_{\theta_l}\to P$, we have the following convergence in distribution:
\begin{align}
	\expon{\sum_{i=0}^{\tau_{x^*}-1}\left(\alpha C_{\theta_l}(\Phi_i)- \Lambda^{(M_k)}_{\theta_l}\right)}\wedge M_k \to \expon{\sum_{i=0}^{\tau_{x^*}-1}\left(\alpha C(\Phi_i)- \hat{\Lambda}\right)}\wedge M_k.
\end{align}
Notice that in the above relation, $\tau_{x^*}$ depends on the transition probability which is not explicitly written. The above convergence follows by the fact that for any $\delta>0$, there exists $N_\delta>0$ for which $\sup_{P\in\xbar{\probspace}} P(\tau_{x^*} > N_\delta) < \delta$ (see the proof of \cref{lem:prelim}). Since these random variables are all bounded by $M_k$, we have
\begin{align}
	&\expect{\Phi_0=x^*}{\theta}{\expon{\sum_{i=0}^{\tau_{x^*}-1}\left(\alpha C_{\theta_l}(\Phi_i)- \Lambda^{(M_k)}_{\theta_l}\right)}\wedge M_k}\\
	&\myquad[8]\to \expect{\Phi_0=x^*}{P}{\expon{\sum_{i=0}^{\tau_{x^*}-1}\left(\alpha C(\Phi_i)- \hat{\Lambda}\right)}\wedge M_k},
\end{align}
which in turn, implies that $\expect{\Phi_0=x^*}{P}{\expon{\sum_{i=0}^{\tau_{x^*}-1}\left(\alpha C(\Phi_i)- \hat{\Lambda}\right)}\wedge M_k} = 1$. Hence,  $\hat{\Lambda}=\Lambda^{(M_k)}$ and contradiction follows.

Now, based on the above argument, there exists $k_0$ large enough such that $0<\Lambda - \Lambda^{(M_{k_0})}<\delta/6$, and $j_0$ large enough such that for all $j\geq j_0$, we have $|\Lambda^{(M_{k_0})}_{\theta_j} - \Lambda^{(M_{k_0})}| < \delta/6$  and $|\Lambda - \Lambda_{\theta_{j}}|<\delta/6$. This implies that for all $j\geq j_0$, we have $|\Lambda^{(M_{k_0})}_{\theta_j} - \Lambda_{\theta_{j}}| =  \Lambda_{\theta_{j}} - \Lambda^{(M_{k_0})}_{\theta_j} < \delta/2$. Contradiction follows by the fact that
$ \Lambda_{\theta_{j}} - \Lambda^{(M_{k_0\vee j})}_{\theta_j} <  \Lambda_{\theta_{j}} - \Lambda^{(M_{k_0})}_{\theta_j}$.

	\end{proof}

	\subsection{Supporting Results of \cref{sec:policygrad}}
	%\subsection{Proof of \cref{lem:propgf}}\label{proof:lem:propgf}
	\begin{lemma}\label[lemma]{lem:propgf}
		If $g(\theta,\Lambda) < \infty$ for some $(\theta,\Lambda)\in\real^{l+1}$, then $\nabla g(\theta,\Lambda)$ exists and in particular, we have
		\begin{align}
			&\begin{aligned}
				\nabla_\theta g(\theta,\Lambda) &= \expect{\Phi_{0} = x^*}{\theta}{\sum_{i=0}^{\tau_{x^*}-1} \big(\alpha \nabla_\theta C_\theta({\Phi}_i) + L_\theta(\Phi_i,\Phi_{i+1}) \big)\expon{\sum_{i=0}^{\tau_{x^*}-1} \left(\alpha C_\theta({\Phi}_i) - \Lambda\right)}}\\
				&= \checkexpect{\check{\Phi}_{0} = x^*}{\theta}{\sum_{i=0}^{\tau_{x^*}-1} \big(\alpha \nabla_\theta C_\theta(\check{\Phi}_i) + L_\theta(\check{\Phi}_i,\check{\Phi}_{i+1}) \big)},
			\end{aligned}\\
			&\frac{\der{g(\theta,\Lambda)}}{\der{\Lambda}} = -\expect{\Phi_{0} = x^*}{\theta}{\tau_{x^*}\expon{\sum_{i=0}^{\tau_{x^*}-1} \left(\alpha C_\theta({\Phi}_i) - \Lambda\right)}} = -\checkexpect{\check{\Phi}_{0} = x^*}{\theta}{\tau_{x^*}}.
		\end{align}
	\end{lemma}
	\begin{proof}\label{proof:lem:propgf}
		Notice that if $g(\theta,\Lambda) < \infty$, then we have
\begin{align}
	g(\theta,\Lambda) &=\! \sum_{k=1}^\infty\,\, \sum_{\substack{x_1,x_2,\cdots,x_{k-1}\\ x_i \neq x^*}} \!\!\!P_\theta\left(\Phi_i=x_i,\,\forall i<k\text{ and } \Phi_{k}=x^*| \Phi_0=x^*\right) \, \e^{\sum_{i=0}^{k-1} \left(\alpha C_\theta(x_i) - \Lambda\right)}\label{eq_app:extendedg}\\
	&= \expect{\Phi_0 = x^*}{\theta}{\expon{\sum_{i=0}^{\tau_{x^*}-1} \left(\alpha C_\theta({\Phi}_i) - \Lambda\right)}} = \checkexpect{\check{\Phi}_0 = x^*}{\theta}{\e^{(\Lambda_\theta - \Lambda)\tau_{x^*}}} < \infty,
\end{align}
i.e., moment generating function of $\tau_{x^*}$, the first return time to state $x^*$ of the Markov chain with transition probability $\check{P}_\theta$, exists and is finite at $\Lambda_\theta - \Lambda$. Hence, there exists a small enough $\epsilon > 0$ such that $g(\theta,\Lambda-\epsilon) < \infty$. Notice that, by \cref{ass:1,ass:2,ass:3}, we have
\begin{align}
	&\expect{\Phi_{0} = x^*}{\theta}{\left\Vert \sum_{i=0}^{\tau_{x^*}-1} \big(\alpha \nabla_\theta C_\theta({\Phi}_i) + L_\theta(\Phi_i,\Phi_{i+1}) \big)\expon{\sum_{i=0}^{\tau_{x^*}-1} \left(\alpha C_\theta({\Phi}_i) - \Lambda\right)}\right\Vert_2} \\
	&\myquad[6]\leq \textrm{Const}\times \expect{\Phi_{0} = x^*}{\theta}{\tau_{x^*}\expon{\sum_{i=0}^{\tau_{x^*}-1} \left(\alpha C_\theta({\Phi}_i) - \Lambda\right)}} \\
	&\myquad[6] = \textrm{Const}\times \checkexpect{\check{\Phi}_0 = x^*}{\theta}{\tau_{x^*} \e^{(\Lambda_\theta - \Lambda)\tau_{x^*}}}\\
	&\myquad[6] \leq \textrm{Const}\times \checkexpect{\check{\Phi}_0 = x^*}{\theta}{\e^{(\Lambda_\theta - \Lambda+\epsilon)\tau_{x^*}}} < \infty.
\end{align}
Hence, the interchange of differentiation and infinite summation is allowed in \cref{eq_app:extendedg} by the dominant convergence theorem, and we have
\begin{align}
	\nabla_\theta g(\theta,\Lambda) &= \sum_{k=1}^\infty\,\,\sum_{\substack{x_1,\cdots,x_{k-1}\\ x_i \neq x^*}} \nabla_\theta\Bigg( P_\theta\left(\Phi_i=x_i,\,\forall i<k\text{ and } \Phi_{k}=x^*| \Phi_0=x^*\right) \e^{\sum_{i=0}^{k-1} \left(\alpha C_\theta(x_i) - \Lambda\right)}\Bigg)\\
	&=\expect{\Phi_{0} = x^*}{\theta}{ \sum_{i=0}^{\tau_{x^*}-1} \big(\alpha \nabla_\theta C_\theta({\Phi}_i) + L_\theta(\Phi_i,\Phi_{i+1}) \big)\expon{\sum_{i=0}^{\tau_{x^*}-1} \left(\alpha C_\theta({\Phi}_i) - \Lambda\right)}}.
\end{align}
The relation for ${\der{g(\theta,\Lambda)}}\mathbin{/}{\der{\Lambda}}$ follows by the same argument.
	\end{proof}

	\subsection{Supporting Results of \cref{sec:heuristic}}
	\begin{corollary}\label[corollary]{cor:unifapprox_truncost}
		$\Lambda_\theta^{(M)}\uparrow \Lambda_\theta$ as $M\uparrow \infty$ uniformly over $\theta\in\real^l$.
	\end{corollary}
	\begin{proof}
		The proof follows by \cref{lem:unifapprox_truncost} and the fact that $\Lambda_\theta^{(M)}$ given here is larger than the one given in \cref{lem:unifapprox_truncost} and smaller than $\Lambda_\theta$.
	\end{proof}

	%\subsection{Proof of \cref{lem:propgmfm}}\label{proof:lem:propgmfm}
	\begin{lemma}\label[lemma]{lem:propgmfm}
		$f^{(M)} (\theta,\Lambda) = \nabla_\theta g^{(M)} (\theta,\Lambda)$ for any $(\theta,\Lambda)\in\real^{l+1}$, and in particular,
		\begin{align}
			f^{(M)} (\theta,\Lambda_\theta^{(M)}) = -\frac{\der{g^{(M)}(\theta,\Lambda)}}{\der{\Lambda}}\bigg|_{\Lambda = \Lambda_\theta^{(M)}} \nabla_\theta \Lambda_\theta^{(M)}.
		\end{align}
	\end{lemma}
	\begin{proof}\label{proof:lem:propgmfm}
		Notice that
\begin{align}
	g^{(M)} (\theta,\Lambda) = \sum_{k=1}^\infty\,\, \sum_{\substack{x_1,x_2,\cdots,x_{k-1}\\ x_i \neq x^*}} P_\theta\left(\Phi_i=x_i,\,\forall i<k\text{ and } \Phi_{k}=x^*| \Phi_0=x^*\right) \,G^{(M)} (\theta,\Lambda),
\end{align}
and that
\begin{align}
	&\expect{\Phi_{0} = x^*}{\theta}{\Vert F^{(M)} (\theta,\Lambda)\Vert_2} = \\
	&\myquad[5]\sum_{k=1}^\infty \,\, \sum_{\substack{x_1,x_2,\cdots,x_{k-1}\\ x_i \neq x^*}} P_\theta\left(\Phi_i=x_i,\,\forall i<k\text{ and } \Phi_{k}=x^*| \Phi_0=x^*\right) \times \\
	&\myquad[11]\left\Vert \nabla_\theta G^{(M)} (\theta,\Lambda) + G^{(M)} (\theta,\Lambda) \sum_{i=0}^{\tau_{x^*}-1} L_\theta(\Phi_i,\Phi_{i+1}) \right\Vert_2 \\
	&\myquad[5]\sum_{k=1}^\infty \,\, \sum_{\substack{x_1,x_2,\cdots,x_{k-1}\\ x_i \neq x^*}}  \left\Vert \nabla_\theta\left( P_\theta\left(\Phi_i=x_i,\,\forall i<k\text{ and } \Phi_{k}=x^*| \Phi_0=x^*\right) \,G^{(M)} (\theta,\Lambda)\right) \right\Vert_2
\end{align}
which are uniformly bounded over $\theta\in\real^l$ by \cref{ass:2,ass:3}, \cref{lem:prelim} (\cref{lem:prelim} implies that for any fixed $i\in\nplus$, $\expect{\Phi_0 = x^*}{\theta}{(\tau_{x^*})^i}$ is uniformly bounded) and the definition of the function $G^{(M)} (\theta,\Lambda)$. Hence, the interchange of gradient and infinite sum is allowed by the dominant convergence theorem and we have $f^{(M)} (\theta,\Lambda) = \nabla_\theta g^{(M)} (\theta,\Lambda)$. Finally, the implicit function theorem, implies that
\begin{align}
	\nabla_\theta \Lambda_\theta^{(M)} = -\left(\frac{\der{g^{(M)}(\theta,\Lambda)}}{\der{\Lambda}}\bigg|_{\Lambda = \Lambda_\theta^{(M)}}\right)^{-1} f^{(M)} (\theta,\Lambda_\theta^{(M)}).
\end{align}
	\end{proof}

	%\subsection{Proof of \cref{lem:uniflowerbd_dg}}\label{proof:lem:uniflowerbd_dg}
	\begin{lemma}\label[lemma]{lem:uniflowerbd_dg}
		For any $I = [a,b] \subset \mathbb{R}$, the value of $-{\der{g^{(M)}(\theta,\Lambda)}}\mathbin{/}{\der{\Lambda}}$ is uniformly bounded away from $0$ over $\theta\in\real^l$, $\Lambda\in I$ and $M>1$.
	\end{lemma}
	\begin{proof}\label{proof:lem:uniflowerbd_dg}
		Following the same argument as in the proof of \cref{lem:propgmfm}, we have
\begin{align}
	\frac{\der{g^{(M)}(\theta,\Lambda)}}{\der{\Lambda}} &= \expect{\Phi_{0} = x^*}{\theta}{\frac{\der{G^{(M)}(\theta,\Lambda)}}{\der{\Lambda}}}\\
	&= \expect{\Phi_{0} = x^*}{\theta}{(-\tau_{x^*})\times H(\theta,\Lambda) \mathbbm{1}\left\{H(\theta,\Lambda) \leq M\right\}}\\
	&\myquad[2]+ \expect{\Phi_{0} = x^*}{\theta}{M\times (-\tau_{x^*})\times  \left[\sum_{i=0}^3 \frac{1}{i!}\left(\ln\left(M^{-1}H(\theta,\Lambda)\right)\right)^i\right]\mathbbm{1}\left\{H(\theta,\Lambda) > M\right\}}\\
	&= -\expect{\Phi_{0} = x^*}{\theta}{\tau_{x^*}\times W^{(M)}(\theta,\Lambda)},
\end{align}
where the random function $W^{(M)}(\theta,\Lambda)$ is defined as follows:
\begin{align}
	W^{(M)} (\theta,\Lambda) \coloneqq  H(\theta,\Lambda) \mathbbm{1}\left\{H(\theta,\Lambda) \leq M\right\} +M \left[\sum_{i=0}^3 \frac{1}{i!}\left(\ln\left(M^{-1}H(\theta,\Lambda)\right)\right)^i\right]\mathbbm{1}\left\{H(\theta,\Lambda) > M\right\}.
\end{align}
Notice that $0 < W^{(M)} (\theta,\Lambda) < G^{(M)}(\theta,\Lambda)$, and that $W^{(M)} (\theta,\Lambda)$ is strictly decreasing in $\Lambda$ for any fixed $\theta\in\real^l$ and $M>1$. Hence, we have
\begin{align}
	\left|\frac{\der{g^{(M)}(\theta,\Lambda)}}{\der{\Lambda}}\right| > \inf_{\theta\in\real^l,M>1} \expect{\Phi_{0} = x^*}{\theta}{ W^{(M)}(\theta,b)},
\end{align}
Let $d = \max(\alpha\overline{C},b)$. Notice that for all $M>1$, $W^{(M)} (\theta,d) = H(\theta,d)$. Hence, we have
\begin{align}
	\inf_{\theta\in\real^l,M>1} \expect{\Phi_{0} = x^*}{\theta}{ W^{(M)}(\theta,b)} &\geq \inf_{\theta\in\real^l,M>1} \expect{\Phi_{0} = x^*}{\theta}{ W^{(M)}(\theta,d)}\\
	&=\inf_{\theta\in\real^l} \expect{\Phi_{0} = x^*}{\theta}{ H(\theta,d)} \\
	&=\inf_{\theta\in\real^l} \checkexpect{\Phi_{0} = x^*}{\theta}{\expon{\left(\Lambda_\theta- d\right)\tau_{x^*}}}\\
	&\geq\inf_{\theta\in\real^l} \checkexpect{\Phi_{0} = x^*}{\theta}{\expon{\left(\alpha\underline{C}- d\right)\tau_{x^*}}}.
\end{align}
By \cref{cor:prelim} and following the same argument as in the proof of the lower bound in \cref{lem:prelim}, we have $\inf_{\theta\in\real^l} \checkexpect{\Phi_{0} = x^*}{\theta}{\expon{\left(\alpha\underline{C}- d\right)\tau_{x^*}}} > 0$.

	\end{proof}

	%\subsection{Proof of \cref{lem:gradgfgmfm}}\label{proof:lem:gradgfgmfm}
	\begin{lemma}\label[lemma]{lem:gradgfgmfm}
		There exists $\delta_0 > 0$ for which the functions $g(\theta,\Lambda)$, $\nabla g(\theta,\Lambda)$, $\nabla^2 g(\theta,\Lambda)$, $g^{(M)}(\theta,\Lambda)$, $\nabla g^{(M)}(\theta,\Lambda)$, $\nabla^2 g^{(M)}(\theta,\Lambda)$ are all uniformly bounded over $\theta\in\real^l$, $\Lambda_\theta - \Lambda < \delta_0$, and $M>1$. Moreover, $g^{(M)}(\theta,\Lambda) \uparrow g(\theta,\Lambda)$ and $\nabla g^{(M)}(\theta,\Lambda) \uparrow \nabla g(\theta,\Lambda)$ as $M\uparrow\infty$ uniformly over $\theta \in \real^l$ and $\Lambda_\theta - \Lambda < \delta_0$.
	\end{lemma}
	\begin{proof}\label{proof:lem:gradgfgmfm}
		Notice that $g(\theta,\Lambda) = \checkexpect{\check{\Phi}_0 = x^*}{\theta}{\e^{(\Lambda_\theta - \Lambda)\tau_{x^*}}}$. By \cref{cor:prelim}, there exists $\delta_0 > 0$ small enough such that for all $\theta\in\real^l$ and all $\Lambda$ with $\Lambda_\theta - \Lambda \leq 2 \delta_0$, we have $g(\theta,\Lambda) < \text{Const}$. Notice that $g^{(M)}(\theta,\Lambda) < g(\theta,\Lambda)$ for all $M$. Next, we show that all the remaining functions are uniformly bounded over $\theta\in\real$ and $\Lambda_\theta - \Lambda < \delta_0$.

Following the same argument as in the proof of \cref{lem:propgf}, for any $(\theta,\Lambda)$ with $\Lambda_\theta - \Lambda < \delta_0$ we have
\begin{align}
	\left\Vert \nabla g(\theta,\Lambda)\right\Vert &\leq \textrm{Const}\times \expect{\Phi_{0} = x^*}{\theta}{\tau_{x^*}\expon{\sum_{i=0}^{\tau_{x^*}-1} \left(\alpha C_\theta({\Phi}_i) - \Lambda\right)}}\\
	&\leq \textrm{Const}\times \expect{\Phi_{0} = x^*}{\theta}{\expon{\sum_{i=0}^{\tau_{x^*}-1} \left(\alpha C_\theta({\Phi}_i) - \Lambda + \delta_0/2\right)}}\\
	&= \textrm{Const} \times g(\theta,\Lambda-\delta_0/2) \leq \textrm{Const},
\end{align}
where the first inequality follows by \cref{ass:2,ass:3}, and the last inequality follows by the choice of $\delta_0$. It is easy to see that the first inequality holds even if we replace $\nabla g(\theta,\Lambda)$ with $\nabla g^{(M)}(\theta,\Lambda)$. Following a similar argument, we deduce that $\nabla^2 g(\theta,\Lambda)$ and $\nabla^2 g^{(M)}(\theta,\Lambda)$ are all uniformly bounded over $\theta\in\real^l$, $\Lambda_\theta - \Lambda < \delta_0$, and $M\in(1,\infty)$. Notice that the second part of \cref{ass:3} is required to ensure this.
Finally, to show the uniform convergence, notice that
\begin{align}
	&\left|g(\theta,\Lambda) - g^{(M)}(\theta,\Lambda)\right| \leq \expect{\Phi_0 = x^*}{\theta}{H(\theta,\Lambda)\mathbbm{1}\left\{H(\theta,\Lambda) > M\right\}}, \\
	&
	\begin{aligned}
		\left\Vert\nabla g(\theta,\Lambda) - \nabla g^{(M)}(\theta,\Lambda)\right\Vert &\leq \textrm{Const}\times \expect{\Phi_0 = x^*}{\theta}{H(\theta,\Lambda-\delta_0/2)\mathbbm{1}\left\{H(\theta,\Lambda) > M\right\}} \\
		&\leq \textrm{Const}\times \expect{\Phi_0 = x^*}{\theta}{H(\theta,\Lambda-\delta_0/2)\mathbbm{1}\left\{H(\theta,\Lambda-\delta_0/2) > M\right\}}.
	\end{aligned}
\end{align}
Hence, to prove the uniform convergence, it is sufficient to show that the class of random variables
\begin{align}
	\mathcal{Y} \coloneqq\left\{H(\theta,\Lambda): \theta\in\real^l,\Lambda_\theta - \Lambda < 3\delta_0/2 \right\}
\end{align}
is uniformly integrable. We claim that $\sup_{X\in\mathcal{Y}} \mathbb{E}\left[X^{1+\epsilon}\right] < \infty$ for some small enough $\epsilon > 0$. Notice that for any $\epsilon > 0$,
\begin{align}
	\expect{\Phi_0 = x^*}{\theta}{H(\theta,\Lambda)^{1+\epsilon}} &= \checkexpect{\check{\Phi}_0 = x^*}{\theta}{\e^{(\Lambda_\theta - \Lambda)\tau_{x^*} }\times \expon{\epsilon\sum_{i=0}^{\tau_{x^*}-1} \left(\alpha C_\theta(\check{\Phi}_i) - \Lambda\right)}}\\
	&\leq \checkexpect{\check{\Phi}_0 = x^*}{\theta}{\e^{3\delta_0\tau_{x^*}/2 } \times \e^{\tau_{x^*}\left(2\alpha\overline{C}+3\delta_0/2\right)\epsilon }},
\end{align}
which is uniformly bounded if $\epsilon < \left(2\alpha\overline{C}+3\delta_0/2\right)^{-1}\times \delta_0/2$. Hence, the family of random variables $\mathcal{Y}$ are uniformly integrable and the result follows.
	\end{proof}

	\begin{corollary}\label[corollary]{cor:unfbound_lambdagrad}
		$\nabla_\theta\Lambda_\theta^{(M)}$, $\nabla^2_\theta\Lambda_\theta^{(M)},$ $\nabla_\theta\Lambda_\theta$, and $\nabla^2_\theta\Lambda_\theta$ are uniformly bounded over $\theta\in\real^l$ and $M>1$.
	\end{corollary}
	\begin{proof}
		The proof follows by \cref{lem:propgf,lem:propgmfm,lem:uniflowerbd_dg,lem:gradgfgmfm}, and simple algebra.
	\end{proof}

	%\subsection{Proof of \cref{cor:unifapprox_nablatheta}}\label{proof:cor:unifapprox_nablatheta}
	\begin{corollary}\label[corollary]{cor:unifapprox_nablatheta}
		$\nabla_\theta\Lambda_\theta^{(M)}\uparrow \nabla_\theta\Lambda_\theta$ as $M\uparrow \infty$ uniformly over $\theta\in\real^l$.
	\end{corollary}
	\begin{proof}
		This follows by \cref{lem:propgf,lem:propgmfm,lem:uniflowerbd_dg,lem:gradgfgmfm}, and \cref{cor:unifapprox_truncost}. Notice that for any $\theta \in \real^l$ and $\Lambda_\theta - \Lambda < \delta_0$, we have
\begin{align}
	\left\Vert f^{(M)}(\theta,\Lambda_\theta^{(M)}) - f(\theta,\Lambda_\theta)\right\Vert \leq \left\Vert f^{(M)}(\theta,\Lambda_\theta^{(M)}) - f^{(M)}(\theta,\Lambda_\theta)\right\Vert + \left\Vert f^{(M)}(\theta,\Lambda_\theta) - f(\theta,\Lambda_\theta)\right \Vert.
\end{align}
By \cref{lem:gradgfgmfm}, we can pick $M_0$ large enough such that $\left\Vert f^{(M)}(\theta,\Lambda_\theta) - f(\theta,\Lambda_\theta)\right\Vert$ is sufficiently small,  for all $M>M_0$. By \cref{cor:unifapprox_truncost}, we can pick $M_1$ large enough such that for all $M>M_1$, we have $\Lambda_\theta - \Lambda_\theta^{(M)} < \delta_0$. Notice that $\left\Vert f^{(M)}(\theta,\Lambda_\theta^{(M)}) - f^{(M)}(\theta,\Lambda_\theta)\right\Vert \leq \textrm{Const}\times \left|\Lambda_\theta^{(M)} - \Lambda_\theta\right|$ for all $M>M_1$, where in here, $\textrm{Const}>0$ is the uniform upper-bound of the set
\begin{align}
	\left\{\left\Vert\frac{\der{f^{(M)}(\theta,\Lambda)}}{\der{\Lambda}} \right\Vert\bigg| \theta\in\real^l, \Lambda_\theta - \Lambda < \delta_0, M\in(1,\infty) \right\}.
\end{align}
The existence of such upper-bound follows by \cref{lem:gradgfgmfm}; in particular, by the fact that $\nabla^2 g^{(M)}(\theta,\Lambda)$ is uniformly bounded over $\theta\in\real^l$, $\Lambda_\theta - \Lambda < \delta_0$, and $M\in(1,\infty)$. Using \cref{cor:unifapprox_truncost}, there exists $M_2 > M_1$ such that for all $M>M_2$, $\left\Vert f^{(M)}(\theta,\Lambda_\theta^{(M)}) - f^{(M)}(\theta,\Lambda_\theta)\right\Vert$ is sufficiently small.
	\end{proof}

	%\subsection{Proof of \cref{lem:prelim}}\label{proof:lem:prelim}
\end{document}